\documentclass[a4paper,11pt]{article} 

\usepackage{amsmath, amssymb, amsthm, mathtools}
\usepackage[mathscr]{euscript}     
\usepackage{dsfont} 
\usepackage{graphicx} 
\usepackage{xspace}
\usepackage{qip} 
\usepackage[section]{thmenvironments} 
\usepackage{layoutcommands} 

\usepackage{underscore} 
\usepackage[shortcuts]{extdash} 
\usepackage{authblk} 
\usepackage{appendix} 
\usepackage{microtype} 
\usepackage{enumitem} 
\usepackage[margin=10pt,font=small,labelfont=bf,labelsep=endash,hypcap=true]{caption}
\usepackage[margin=10pt,font=small,labelfont=bf,labelformat=parens,subrefformat=parens,labelsep=space,hypcap=false]{subcaption}
\usepackage{xcolor}
\usepackage{tikz} 
\usetikzlibrary{arrows,fit,positioning,shapes,decorations.pathmorphing} 

\usepackage[linktoc=all]{hyperref} 
\usepackage[open,numbered]{bookmark} 
\usepackage{cite} 
\usepackage{hyperlinks} 
\usepackage{eprint} 

\hypersetup{
  pdftitle={Quantum authentication with key recycling},
  pdfauthor={Christopher Portmann},
  pdfsubject={Quantum cryptography},
  pdfstartview={FitH}
}

\title{Quantum authentication with key recycling}

\author{Christopher Portmann\email{chportma@.ethz.ch}}

\affil{Institute for Theoretical Physics, ETH Zurich, 8093 Zurich, Switzerland.}

\date{\today}

\newcommand{\cauth}{\text{auth}}
\newcommand{\qauth}{\text{q-auth}}
\newcommand{\ecc}{\text{ecc}}
\newcommand{\key}{\text{key}}
\newcommand{\acc}{\textsf{acc}}
\newcommand{\rej}{\textsf{rej}}
\newcommand{\sym}[2]{\brs{#1,#2}_{\text{Sp}}}

\begin{document}

\maketitle

\begin{abstract}
  We show that a family of quantum authentication protocols introduced
  in [Barnum et al., FOCS 2002] can be used to construct a secure
  quantum channel and additionally recycle all of the secret key if
  the message is successfully authenticated, and recycle part of the
  key if tampering is detected. We give a full security proof that
  constructs the secure channel given only insecure noisy channels and
  a shared secret key. We also prove that the number of recycled key
  bits is optimal for this family of protocols, i.e., there exists an
  adversarial strategy to obtain all non-recycled bits. Previous works
  recycled less key and only gave partial security proofs, since they
  did not consider all possible distinguishers (environments) that may
  be used to distinguish the real setting from the ideal secure
  quantum channel and secret key resource.
\end{abstract}



\section{Introduction}
\label{sec:intro}

\subsection{Reusing a one-time pad}
\label{sec:intro.reusing}

A one-time pad can famously be used only once~\cite{Sha49}, i.e., a
secret key as long as the message is needed to encrypt it with
information\-/theoretic security. But this does not hold anymore if
the honest players can use quantum technologies to communicate. A
quantum key distribution (QKD) protocol~\cite{BB84,SBCDLP09} allows
players to expand an initial short secret key, and thus encrypt
messages that are longer than the length of the original key. Instead
of first expanding a key, and then using it for encryption, one can
also swap the order if the initial key is long enough: one first
encrypts a message, then recycles the key. This is possible due to the
same physical principles as QKD: quantum states cannot be cloned, so
if the receiver holds the exact cipher that was sent, the adversary
cannot have a copy, and thus does not have any information about the
key either, so it may be reused. This requires the receiver to verify
the authenticity of the message received, and if this process fails, a
net key loss occurs \--- the same happens in QKD: if an adversary
tampers with the communication, the players have to abort and also
lose some of the initial secret key.

\subsection{Quantum authentication and key recycling}
\label{sec:intro.auth}

Some ideas for recycling encryption keys using quantum ciphers were
already proposed in 1982~\cite{BBB82}. Many years later, Damg{\aa}rd
et al.~\cite{DPS05} (see also \cite{DPS14,FS17}) showed how to encrypt
a classical message in a quantum state and recycle the key. At roughly
the same time, the first protocol for authenticating quantum messages
was proposed by Barnum et al.~\cite{BCGST02}, who also proved that
quantum authentication necessarily encrypts the message as
well. Gottesman~\cite{Got03} then showed that after the message is
successfully authenticated by the receiver, the key can be leaked to
the adversary without compromising the confidentiality of the message.
And Oppenheim and Horodecki~\cite{OH05} adapted the protocol of
\cite{BCGST02} to recycle key.  But the security definitions in these
initial works on quantum authentication have a major flaw: they do not
consider the possibility that an adversary may hold a purification of
the quantum message that is encrypted. This was corrected by Hayden,
Leung and Mayers~\cite{HLM11}, who give a composable security
definition for quantum authentication with key recycling. They then
show that the family of protocols from~\cite{BCGST02} are secure, and
prove that one can recycle part of the key if the message is accepted.

The security proof from~\cite{HLM11} does however not consider all
possible environments. Starting in works by Simmons in the 80's and
then Stinson in the 90's (see, for example,
\cite{Sim85,Sim88,Sti90,Sti94}) the classical literature on
authentication studies two types of attacks: \emph{substitution
  attacks} \--- where the adversary obtains a valid pair of message
and cipher\footnote{Here we use the term \emph{cipher} to refer to the
  authenticated message, which is often a pair of the original message
  and a tag or message authentication code (MAC), but not
  necessarily.} and attempts to substitute the cipher with one that
will decode to a different message \--- and \emph{impersonation
  attacks} \--- where the adversary directly sends a forged cipher to
the receiver, without knowledge of a valid message\-/cipher pair. To
the best of our knowledge, there is no proof showing that security
against impersonation attacks follows from security against
substitution attacks, hence the literature analyzes both attacks
separately.\footnote{In fact, one can construct examples where the
  probability of a successful impersonation attack is higher than the
  probability of a successful substitution attack. This can occur,
  because any valid cipher generated by the adversary is considered a
  successful impersonation attack, whereas only a cipher that decrypts
  to a different message is considered a successful substitution
  attack.}  This is particularly important in the case of composable
security, which aims to prove the security of the protocol when used
in any arbitrary environment, therefore also in an environment that
first sends a forged cipher to the receiver, learns wether it is
accepted or rejected, then provides a message to the sender to be
authenticated, and finally obtains the cipher for this message. This
is all the more crucial when key recycling is involved, since the
receiver will already recycle (part of) the key upon receiving the
forged cipher, which is immediately given to the environment. The work
of Hayden et al.~\cite{HLM11} only considers environments that perform
substitution attacks \--- i.e., first provide the sender with a
message, then change the cipher, and finally learn the outcome of the
authentication as well as receive the recycled key. Hence they do not
provide a complete composable security proof of quantum
authentication, which prevents the protocol from being composed in an
arbitrary environment.\footnote{For example, QKD can be broken if the
  underlying authentication scheme is vulnerable to impersonation
  attacks, because Eve could trick Alice into believing that the
  quantum states have been received by Bob so that she releases the
  basis information.}

More recently, alternative security definitions for quantum
authentication have been proposed, both without~\cite{DNS12,BW16} and
with~\cite{GYZ16} key recycling (see also \cite{AM16}). These still
only consider substitution attacks, and furthermore, they are,
strictly speaking, not composable. While it is possible to prove that
these definitions imply security in a composable framework (if one
restricts the environment to substitution attacks), the precise way in
which the error $\eps$ carries over to the framework has not been
worked out in any of these papers. If two protocols with composable
errors $\eps$ and $\delta$ are run jointly (e.g., one is a subroutine
of the other), the error of the composed protocol is bounded by the
sum of the individual errors, $\eps+\delta$. If a security definition
does not provide a bound on the composable error, then one cannot
evaluate the new error after composition.\footnote{In an asymptotic
  setting, one generally does not care about the exact error, as long
  as it is negligible. But for any (finite) implementation, the exact
  value is crucial, since without it, it is impossible to set the
  parameters accordingly, e.g., how many qubits should one send to get
  an error $\eps \leq 10^{-18}$.}  For example, quantum authentication
with key recycling requires a backwards classical authentic channel,
so that the receiver may tell the sender that the message was
accepted, and allow her to recycle the key. The error of the complete
protocol is thus the sum of errors of the quantum authentication and
classical authentication protocols. Definitions such as those of
\cite{DNS12,BW16,GYZ16} are not sufficient to directly obtain a bound
on the error of such a composed protocol.

In the other direction, it is immediate that if a protocol is
$\eps$\=/secure according to the composable definition used in this
work, then it is secure according to \cite{DNS12,BW16,GYZ16} with the
same error $\eps$. More precisely, proving that the quantum
authentication scheme constructs a secure channel is sufficient to
satisfy \cite{DNS12,BW16} \--- i.e., the ideal functionality is a
secure channel which only allows the adversary to decide if the
message is delivered, but does not leak any information about the
message to the adversary except its length (confidentiality), nor does
it allow the adversary to modify the message (authenticity). And
proving that the scheme constructs a secure channel that additionally
generates fresh secret key is sufficient to satisfy the definition of
\emph{total authentication} from \cite{GYZ16}. Garg et
al.~\cite{GYZ16} also propose a definition of \emph{total
  authentication with key leakage}, which can be captured in a
composable framework by a secure channel that generates fresh key and
leaks some of it to the adversary. This is however a somewhat
unnatural ideal functionality, since it requires a deterministic
leakage function, which may be unknown or not exist, e.g., in concrete
protocols the bits leaked can depend on the adversary's behavior \---
this is the case for the \emph{trap code}~\cite{BGS13,BW16}, which we
discuss further in \secref{sec:conclusion}. The next natural step for
players in such a situation is to extract a secret key from the
partially leaked key, and thus the more natural ideal functionality is
what one obtains after this privacy amplification
step~\cite{BBCM95,RK05}: a secure channel that generates fresh secret
key, but where the key generated may be shorter than the key
consumed. The ideal functionality used in the current work provides
this flexibility: the amount of fresh key generated is a parameter
which may be chosen so as to produce less key than consumed, the same
amount, or even more.\footnote{One may obtain more key than consumed
  by using the constructed secure channel to share secret key between
  the players. We use this technique to compensate for key lost in a
  classical authentication subroutine, that cannot be recycled.}
Hence, with one security definition, we encompass all these different
cases \--- no key recycling, partial key recycling, total key
recycling, and even a net gain of secret key. Furthermore, having all
these notions captured by ideal functionalities makes for a
particularly simple comparison between the quite technical definitions
appearing in \cite{DNS12,BW16,GYZ16}.

\subsection{Contributions}
\label{sec:intro.contributions}

In this work we use the Abstract Cryptography (AC)
framework~\cite{MR11} to model the composable security of quantum
authentication with key recycling. AC views cryptography as a resource
theory: a protocol constructs a (strong) resource given some (weak)
resources. For example, the quantum authentication protocols that we
analyze construct two resources: a secure quantum channel \--- a
channel that provides both \emph{confidentiality} and
\emph{authenticity} \--- and a secret key resource that shares a fresh
key between both players. In order to construct these resources, we
require shared secret key, an insecure (noiseless) quantum channel and
a backwards authentic classical channel. These are all resources, that
may in turn be constructed from weaker resources, e.g., the classical
authentic channel can be constructed from a shared secret key and an
insecure channel, and noiseless channels are constructed from noisy
channels. Due to this constructive aspect of the framework, it is also
called \emph{constructive cryptography} in the
literature~\cite{Mau12,MR16}.

Although this approach is quite different from the Universal
Composability (UC) framework~\cite{Can01,Can13}, in the setting
considered in this work \--- with one dishonest player and where
recipients are denoted by classical strings\footnote{In a more general
  setting, a message may be in a superposition of ``sent'' and ``not
  sent'' or a superposition of ``sent to Alice'' and ``sent to Bob'',
  which cannot be modeled in UC, but is captured in
  AC~\cite{PMMRT17}.} \--- the two frameworks are essentially
equivalent and the same results could have been derived with a quantum
version of UC~\cite{Unr10}. In UC, the constructed resource would be
called \emph{ideal functionality}, and the resources used in the
construction are setup assumptions.

We thus first formally define the ideal resources constructed by the
quantum authentication protocol with key recycling \--- the secure
channel and key resource mentioned in this introduction \--- as well
as the resources required by this construction. We then prove that a
family of quantum authentication protocols proposed by Barnum et
al.~\cite{BCGST02} satisfy this construction, i.e., no distinguisher
(called environment in UC) can distinguish the real system from the
ideal resources and simulator except with an advantage $\eps$ that is
exponentially small in the security parameter. This proof considers
all distinguishers allowed by quantum mechanics, including those that
perform impersonation attacks.

We show that in the case where the message is accepted, every bit of
key may be recycled. And if the message is rejected, one may recycle
all the key except the bits used to one-time pad the
cipher.\footnote{Key recycling in the case of a rejected message is
  not related to any quantum advantage. A protocol does not leak more
  information about the key than (twice) the length of the cipher, so
  the rest may be reused. The same holds for classical
  authentication~\cite{Por14}.} We prove that this is optimal for the
family of protocols considered, i.e., an adversary may obtain all
non\-/recycled bits of key. This improves on previous results, which
recycled less key and only considered a subset of possible
environments. More specifically, Hayden et al.~\cite{HLM11}, while
also analyzing protocols from \cite{BCGST02}, only recycle part of the
key in case of an accept, and lose all the key in case of a
reject. Garg et al.~\cite{GYZ16} propose a new protocol, which they
prove can recycle all of the key in the case of an accept, but do not
consider key recycling in the case of a reject either. The protocols
we analyze are also more key efficient than that of~\cite{GYZ16}. We
give two instances which need $\Theta(m +\log 1/\eps)$ bits of initial
secret key, instead of the $\Theta((m +\log 1/\eps)^2)$ required by
\cite{GYZ16}, where $m$ is the length of the message and $\eps$ is the
error. Independently from this work, Alagic and Majenz~\cite{AM16}
proved that one of the instances analyzed here satisfies the weaker
security definition of \cite{GYZ16}.

Note that the family of protocols for which we provide a security
proof is a subset of the (larger) family introduced
in~\cite{BCGST02}. More precisely, Barnum et al.~\cite{BCGST02} define
quantum authentication protocols by composing a quantum one-time pad
and what they call a \emph{purity testing code} \--- which, with high
probability, will detect any noise that may modify the encoded message
\--- whereas we require a stricter notion, a \emph{strong purity
  testing code} \--- which, with high probability, will detect any
noise. This restriction on the family of protocols is necessary to
recycle all the key. In fact, there exists a quantum authentication
scheme, the \emph{trap code}~\cite{BGS13,BW16}, which is a member of
the larger class from~\cite{BCGST02} but not the stricter class
analyzed here, and which leaks part of the key to the adversary, even
upon a successful authentication of the message \--- this example is
discussed in \secref{sec:conclusion}.

We then give two explicit instantiations of this family of quantum
authentication protocols. The first is the construction used in
\cite{BCGST02}, which requires an initial key of length $2m+2n$, where
$m$ is the length of the message and $n$ is the security parameter,
and has error $\eps \leq 2^{-n/2+1} \sqrt{2m/n+2}$. The second is an
explicit unitary $2$\-/design~\cite{Dan05,DCEL09} discovered by
Chau~\cite{Cha05}, which requires $5m+4n$ bits of initial
key\footnote{The complete design would require $5m+5n$ bits of key,
  but we show that some of the unitaries are redundant when used for
  quantum authentication and can be dropped.} and has error
$\eps \leq 2^{-n/2+1}$. Both constructions have a net loss of $2m+n$
bits of key if the message fails authentication. Since several other
explicit quantum authentication protocols proposed in the literature
are instances of this family of schemes, our security proof is a proof
for these protocols as well \--- this is discussed further in
\secref{sec:conclusion}.

Finally, we show how to construct the resources used by the
protocol from nothing but insecure noisy channels and shared secret
key, and calculate the joint error of the composed protocols. We also
show how to compensate for the bits of key lost in the construction of
the backwards authentic channel, so that the composed protocol still
has a zero net key consumption if no adversary jumbles the
communication.

There is currently no work in the literature that analyzes the composable
security of quantum authentication without key recycling. And although
a security proof with key recycling is automatically a security proof
without key recycling, the parameters are not optimal \--- recycling
the key results in a larger error \--- and the proof given in this
paper is only valid for strong purity testing codes. So for
completeness, we provide a proof of security for quantum
authentication without key recycling in \appendixref{app:norecycle},
which is valid for weak purity testing codes and achieves an optimal
error.

\subsection{Structure of this paper}
\label{sec:intro.structure}

In \secref{sec:cc} we give a brief introduction to the main concepts
of AC, which are necessary to understand the notion of cryptographic
construction and corresponding security defintion. A more extended
introduction to AC is provided in \appendixref{app:ac}. In
\secref{sec:auth} we then define the resources constructed and used by
a quantum authentication scheme with key recycling. We introduce the
family of protocols from~\cite{BCGST02} that we analyze in this work,
and then prove that they construct the corresponding ideal
resources. We also prove that the number of recycled bits is
optimal. In \secref{sec:construction} we show how to construct the
various resources used by the quantum authentication protocol, and put
the pieces together to get a security statement for the joint protocol
that constructs the secure quantum channel and secret key resource
from nothing but noisy insecure channels and shared secret
key. Finally, in \secref{sec:conclusion} we discuss the relation
between some quantum authentication schemes that have appeared in the
literature and those analyzed here, as well as some open problems. An
overview of the appendices is given on \pref{app}.

\section{Constructive cryptography}
\label{sec:cc}

As already mentioned in \secref{sec:intro.contributions}, the AC
framework~\cite{MR11} models cryptography as a resource theory. In
this section we give a brief overview of how these constructive
statements are formalized. We illustrate this with an example taken
from \cite{Por14}, namely authentication of classical messages with
message authentication codes (MAC).  An expanded version of this
introduction to AC is provided in \appendixref{app:ac}.

In an $n$ player setting, a \emph{resource} is an object with $n$
interfaces, that allows every player to input messages and receive
other messages at her interface. The objects depicted in
\figref{fig:channel.dishonest} are examples of resources. The insecure
channel in \figref{fig:channel.dishonest.insecure} allows Alice to
input a message at her interface on the left and allows Bob to receive
a message at his interface on the right. Eve can intercept Alice's
message and insert a message of her choosing at her interface. The
authentic channel resource depicted in
\figref{fig:channel.dishonest.authentic} also allows Alice to send a
message and Bob to receive a message, but Eve's interface is more
limited than for the insecure channel: she can only decide if Bob
receives the message, an error symbol or nothing at all \--- by
inputing $0$, $1$, or nothing, respectively, at her interface \--- but
not tamper with the message being sent. The key resource drawn in
\figref{fig:resource.key} provides each player with a secret key when
requested. If two resources $\aK$ and $\aC$ are both available to the
players, we write $\aK \| \aC$ for the resource resulting from their
parallel composition \--- this is to be understood as the resources
being merged into one: the interfaces belonging to player $i$ are
simultaneously accessible to her as one new interface, which we depict
in \figref{fig:resource.composition}. In \appendixref{app:ac} we
provide a more detailed description of the resources from
\figref{fig:channel.dishonest} along a discussion of how to model them
mathematically.

\begin{figure}[tbp]
\subcaptionbox[Insecure channel with
adversary]{\label{fig:channel.dishonest.insecure}An insecure channel
  from Alice (on the left) to Bob (on the right) allows Eve (below) to
  intercept the message and insert a message of her own.}[.5\textwidth][c]{
\begin{tikzpicture}[
resource/.style={draw,thick,minimum width=3.2cm,minimum height=1cm},
filter/.style={draw,thick,minimum width=1.6cm,minimum height=.8cm},
sArrow/.style={->,>=stealth,thick},
sLine/.style={-,thick}]

\small

\def\t{2.2} 
\def\v{1.4} 
\def\w{1.2} 
\def\s{.4}

\node[resource] (ch) at (0,0) {};
\node[yshift=-1.5,above right] at (ch.north west) {\footnotesize
  Insecure channel $\aC$};
\node (alice) at (-\t,0) {};
\node[yshift=-1.5,above] at (alice) {\footnotesize
  Alice};
\node (bob) at (\t,0) {};
\node[yshift=-1.5,above] at (bob) {\footnotesize
  Bob};
\node (eve) at (0,-\w) {\footnotesize
  Eve};

\draw[sArrow] (alice.center) to (-\s,0) to (-\s,0 |- eve.north);
\draw[sArrow] (\s,0 |- eve.north) to (\s,0) to (bob.center);

\end{tikzpicture}}
\subcaptionbox[Authentic channel with
adversary]{\label{fig:channel.dishonest.authentic}An authentic channel
from Alice to Bob allows Eve (below) to
receive a copy of the message and choose whether Bob receives it, an
error symbol, or nothing at all.}[.5\textwidth][c]{
\begin{tikzpicture}[
resource/.style={draw,thick,minimum width=3.2cm,minimum height=1cm},
filter/.style={draw,thick,minimum width=1.6cm,minimum height=.8cm},
sArrow/.style={->,>=stealth,thick},
sLine/.style={-,thick}]

\small

\def\t{2.2} 
\def\v{1.4} 
\def\w{1.2} 
\def\s{.4}

\node[resource] (ch) at (0,0) {};
\node[yshift=-1.5,above right] at (ch.north west) {\footnotesize
  Authentic channel $\aA$};
\node (alice) at (-\t,0) {};
\node[yshift=-1.5,above] at (alice) {\footnotesize
  Alice};
\node (bob) at (\t,0) {};
\node[yshift=-1.5,above] at (bob) {\footnotesize
  Bob};
\node (eve) at (0,-\w) {\footnotesize
  Eve};

\draw[sLine] (alice.center) to (0,0) to node (handle) {} (330:2*\s);
\draw[sArrow] (2*\s,0) to (bob.center);
\draw[sArrow] (-\s,0) to (-\s,0 |- eve.north);
\draw[double,thick] (handle.center |- eve.north) to node[auto,swap,pos=.25] {$0,1$} (handle.center);

\end{tikzpicture}}

\vspace{6pt}

\subcaptionbox[Secret key resource]{\label{fig:resource.key}A secret
  key resource distributes a perfectly uniform key $k$ to the players
  when they send a request \tt{req.}}[.5\textwidth][c]{
\begin{tikzpicture}[
resource/.style={draw,thick,minimum width=3.2cm,minimum height=1cm},
sArrow/.style={->,>=stealth,thick}]

\small

\def\t{2.3} 
\def\v{.2}

\node[resource] (keyBox) at (0,0) {};
\node[draw] (key) at (0,-\v) {key};
\node[yshift=-1.5,above right] at (keyBox.north west) {\footnotesize
  Secret key $\aK$};
\node (alice) at (-\t,0) {};
\node (bob) at (\t,0) {};
\node (up) at (0,\v) {};
\node (down) at (0,-\v) {};

\draw[sArrow] (alice |- up) to node[auto] {\footnotesize req.}
(keyBox.west |- up);
\draw[sArrow] (bob |- up) to node[auto,swap] {\footnotesize req.}
(keyBox.east |- up);
\draw[sArrow] (key.west |- down) to node[auto,pos=.8] {$k$} (alice |- down);
\draw[sArrow] (key.east |- down) to node[auto,swap,pos=.8] {$k$} (bob |- down);

\end{tikzpicture}}
\subcaptionbox[Insecure channel and
key]{\label{fig:resource.composition}If two resources $\aK$ and $\aC$
  are available to the players, we denote the composition of the two
  as the new resource $\aK \| \aC$.}[.5\textwidth][c]{
\begin{tikzpicture}[
resource/.style={draw,thick,minimum width=3.2cm,minimum height=1cm},
filter/.style={draw,thick,minimum width=1.6cm,minimum height=.8cm},
sArrow/.style={->,>=stealth,thick},
sLine/.style={-,thick}]

\small

\def\t{2.5} 
\def\w{1.2} 
\def\s{.4}
\def\u{1.5}
\def\v{.2}

\node[resource] (ch) at (0,0) {};
\node[yshift=-1.5,above right] at (ch.north west) {\footnotesize
  Insecure channel $\aC$};
\node (alice) at (-\t,0) {};
\node (bob) at (\t,0) {};

\draw[sArrow] (alice.center) to (-\s,0) to (-\s,-\w);
\draw[sArrow] (\s,-\w) to (\s,0) to (bob.center);

\node[resource] (keyBox) at (0,\u) {};
\node[draw] (key) at (0,\u-\v) {key};
\node[yshift=-1.5,above right] at (keyBox.north west) {\footnotesize
  Secret key $\aK$};
\node[above] (keyLabel) at (keyBox.north) {};
\node (up) at (0,\u+\v) {};
\node (down) at (0,\u-\v) {};

\draw[sArrow] (alice |- up) to node[auto,pos=.4] {\footnotesize req.}
(keyBox.west |- up);
\draw[sArrow] (bob |- up) to node[auto,swap,pos=.4] {\footnotesize req.}
(keyBox.east |- up);
\draw[sArrow] (key.west |- down) to node[auto,pos=.8] {$k$} (alice |- down);
\draw[sArrow] (key.east |- down) to node[auto,swap,pos=.8] {$k$} (bob |- down);

\node[draw,thick,fit=(keyBox)(ch)(keyLabel),inner sep=6] (comp) {};
\node[yshift=-1.5,above right] at (comp.north west) {\footnotesize
  Composed resource $\aK \| \aC$};

\end{tikzpicture}}

\caption[Examples of resources]{\label{fig:channel.dishonest}Some examples of
  resources. The insecure channel on the top left could transmit either
  classical or quantum messages. The authentic channel on the top right is
  necessarily classical, since it clones the message.}
\end{figure}

\emph{Converters} capture operations that a player might perform
locally at her interface. For example, if the players share a key
resource and an insecure channel, Alice might decide to append a MAC
to her message. This is modeled as a converter $\pi^\cauth_A$ that
obtains the message $x$ at the outside interface, obtains a key at the
inside interface from a key resource $\aK$ and sends $(x,h_k(x))$ on
the insecure channel $\aC$, where $h_k$ is taken from a family of
strongly $2$\-/universal hash functions~\cite{WC81,Sti94}. We
illustrate this in \figref{fig:classicalauth.real}. Converters are
always drawn with rounded corners. If a converter $\alpha_i$ is
connected to the $i$ interface of a resource $\aR$, we write
$\alpha_i\aR$ or $\aR\alpha_i$ for the new resource obtained by
connecting the two.\footnote{In this work we adopt the convention of
  writing converters at the $A$ and $B$ interfaces on the left and
  converters at the $E$ interface on the right, though there is no
  mathematical difference between $\alpha_i\aR$ and $\aR\alpha_i$.}

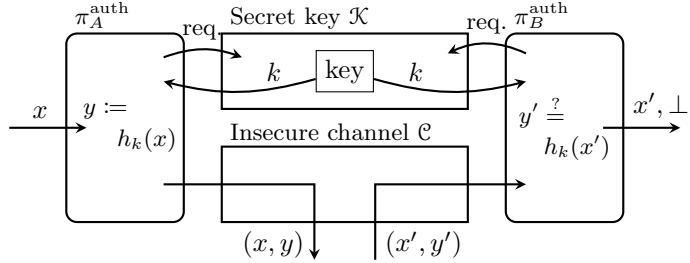
\begin{figure}[tb]
\begin{centering}

\begin{tikzpicture}[
thinResource/.style={draw,thick,minimum width=1.618*2cm,minimum height=1cm},
protocol/.style={draw,thick,minimum width=1.545cm,minimum
  height=2.5cm,rounded corners},
pnode/.style={minimum width=1cm,minimum height=.5cm},
sArrow/.style={->,>=stealth,thick}]

\small

\def\t{4.413} 
\def\u{2.89} 
\def\v{.75}
\def\i{.846} 

\node[pnode] (a1) at (-\u,\v) {};
\node[pnode] (a2) at (-\u,0) {};
\node[pnode] (a3) at (-\u,-\v) {};
\node[protocol,text width=1.1cm] (a) at (-\u,0) {\footnotesize $y
  \coloneqq $\\$\quad \ h_k(x)$};
\node[yshift=-2,above right] at (a.north west) {\footnotesize
  $\pi^{\cauth}_A$};
\node (alice) at (-\t,0) {};

\node[pnode] (b1) at (\u,\v) {};
\node[pnode] (b2) at (\u,0) {};
\node[pnode] (b3) at (\u,-\v) {};
\node[protocol,text width=1.2cm] (b) at (\u,0) {\footnotesize $y' \stackrel{?}{=}$\\$\quad h_k(x')$};
\node[yshift=-2,above right] at (b.north west) {\footnotesize $\pi^{\cauth}_B$};
\node (bob) at (\t,0) {};

\node[thinResource] (keyBox) at (0,\v) {};
\node[pnode] (keyInnerBoxLeft) at (-\i,\v) {};
\node[pnode] (keyInnerBoxRight) at (\i,\v) {};
\node[draw] (key) at (0,\v) {key};
\node[yshift=-2,above right] at (keyBox.north west) {\footnotesize
  Secret key $\aK$};
\node[thinResource] (channel) at (0,-\v) {};
\node[yshift=-1.5,above right] at (channel.north west) {\footnotesize
  Insecure channel $\aC$};
\node (eveleft) at (-.4,-1.75) {};
\node (everight) at (.4,-1.75) {};
\node (ajunc) at (eveleft |- a3) {};
\node (bjunc) at (everight |- b3) {};

\draw[sArrow,bend left=15] (key) to node[auto,swap,pos=.15,yshift=-2] {$k$} (a1);
\draw[sArrow,bend right=15] (key) to node[auto,pos=.15,yshift=-2] {$k$} (b1);
\draw[sArrow,bend left=20] (a1) to node[auto,pos=.5] {\footnotesize req.} (keyInnerBoxLeft);
\draw[sArrow,bend right=25] (b1) to node[auto,swap,pos=.5] {\footnotesize req.} (keyInnerBoxRight);

\draw[sArrow] (alice.center) to  node[auto,pos=.4] {$x$} (a2);
\draw[sArrow] (b2) to node[auto,pos=.75] {$x',\bot$} (bob.center);

\draw[sArrow] (a3) to (ajunc.center)
to node[pos=.8,auto,swap] {$(x,y)$} (eveleft.center);
\draw[sArrow] (everight.center) to node[pos=.2,auto,swap] {$(x',y')$}
(bjunc.center) to (b3);

\end{tikzpicture}

\end{centering}
\caption[Real classical authentication
system]{\label{fig:classicalauth.real}The real system for a MAC
  protocol. Alice authenticates her message by appending a MAC to
  it. Bob checks if the MAC is correct and either accepts or rejects
  the message.}
\end{figure}

A protocol is then defined by a set of converters, one for every
honest player. Another type of converter that we need is a
\emph{filter}. The resources illustrated in
\figref{fig:channel.dishonest} depict a setting with an adversary that
has some control over these resources. For a cryptographic protocol to
be useful it is not sufficient to provide guarantees on what happens
when an adversary is present, one also has to provide a guarantee on
what happens when no adversary is present, e.g., if no adversary
tampers with the message on the insecure channel, then Bob will
receive the message that Alice sent. We model this setting by covering
the adversarial interface with a filter that emulates an honest
behavior. In \figref{fig:channel.honest} we draw
an insecure and an authentic channel with filters $\sharp_E$ and
$\lozenge_E$ that transmit the message to Bob. In the case of the
insecure channel, one may want to model an honest noisy channel when
no adversary is present. This is done by having the filter $\sharp_E$
add some noise to the message. A dishonest player removes this and has
access to a noiseless channel as in
\figref{fig:channel.dishonest.insecure}.

\begin{figure}[tb]
\subcaptionbox[Insecure channel without
adversary]{\label{fig:channel.honest.insecure}When no adversary is
  present, Alice's message is delivered to Bob. In the case of a noisy
channel, this noise is introduced by the filter $\sharp_E$.}[.5\textwidth][c]{
\begin{tikzpicture}[
resource/.style={draw,thick,minimum width=3.2cm,minimum height=1cm},
filter/.style={draw,thick,minimum width=1.6cm,minimum height=.8cm,rounded corners},
sArrow/.style={->,>=stealth,thick},
sLine/.style={-,thick}]

\small

\def\t{2.2} 
\def\v{1.4} 
\def\s{.4}

\node[resource] (ch) at (0,0) {};
\node[yshift=-1.5,above right] at (ch.north west) {\footnotesize
  Insecure channel $\aC$};
\node[filter] (fil) at (0,-\v) {};
\node[xshift=2,below left] at (fil.north west) {\footnotesize $\sharp_E$};
\node (alice) at (-\t,0) {};
\node (bob) at (\t,0) {};

\draw[sLine] (alice.center) to (-\s,0) to (-\s,-\v);
\draw[sLine,decorate,decoration={snake,amplitude=1mm,segment length=1.8mm}] (-\s,-\v) to (\s,-\v);
\draw[sArrow] (\s,-\v) to (\s,0) to (bob.center);

\end{tikzpicture}}
\subcaptionbox[Authentic channel without
adversary]{\label{fig:channel.honest.authentic}When no adversary is
  present, Bob receives the message sent by Alice.}[.5\textwidth][c]{
\begin{tikzpicture}[
resource/.style={draw,thick,minimum width=3.2cm,minimum height=1cm},
filter/.style={draw,thick,minimum width=1.6cm,minimum height=.8cm,rounded corners},
sArrow/.style={->,>=stealth,thick},
sLine/.style={-,thick}]

\small

\def\t{2.2} 
\def\v{1.4} 
\def\s{.4}

\node[resource] (ch) at (0,0) {};
\node[yshift=-1.5,above right] at (ch.north west) {\footnotesize
  Authentic channel $\aA$};
\node[filter] (fil) at (0,-\v) {};
\node[xshift=2,below left] at (fil.north west) {\footnotesize $\lozenge_E$};
\node (alice) at (-\t,0) {};
\node (bob) at (\t,0) {};

\draw[sArrow] (alice.center) to (bob.center);
\draw[sArrow] (-\s,0 |- ch.center) to (-\s,0 |- fil.center);
\draw[double,thick] (\s,0 |- fil.center) to node[auto,swap,pos=.45,xshift=-1] {$0$} (\s,0 |- ch.center);

\end{tikzpicture}}
\caption[Some channels without
adversary]{\label{fig:channel.honest}Channels with filters. The two
  channels from Figures~\ref{fig:channel.dishonest.insecure} and
  \ref{fig:channel.dishonest.authentic} are represented with filters
  on Eve's interface emulating an honest behavior, i.e., when no
  adversary is present.}
\end{figure}
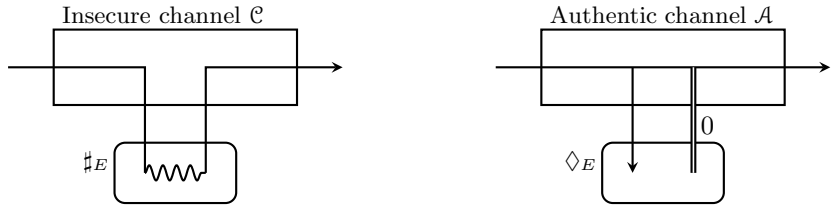

We use the term \emph{filtered resource} to refer to a pair of a
resource $\aR$ and a filter $\sharp_E$, and often write
$\aR_\sharp = (\aR,\sharp_E)$. Such an object can be thought of as
having two modes: it is characterized by the resource $\aR\sharp_E$
when no adversary is present and by the resource $\aR$ when the
adversary is present.

The final object that is required by the AC framework to define the
notion of construction and prove that it is composable, is a
(pseudo\-/)metric defined on the space of resources that measures how
close two resources are. In the following, we use a distinguisher
based metric, i.e., the maximum advantage a distinguisher has in
guessing whether it is interacting with resource $\aR$ or $\aS$, which
we write $d(\aR,\aS)$. More specifically, let $\aD$ be a
distinguisher, and le $\aD[\aR]$ and
$\aD[\aS]$ be the binary random variables corresponding to $\aD$'s
output when connected to $\aR$ and $\aS$, respectively. Then the
distinguishing advantage between $\aR$ and $\aS$ is defined as
\[ d(\aR,\aS) \coloneqq \sup_{\aD} \left| \Pr\left[ \aD[\aR] = 0
  \right] - \Pr\left[ \aD[\aS] = 0 \right] \right|\,.\]
Since we study information\-/theoretic security in this work, the
supremum is taken over the set of all possible distinguishers allowed
by quantum mechanics. This is discussed further in
\appendixref{app:ac.instantiation}.

We are now ready to define the security of a cryptographic
protocol. We do so in the three player setting, for honest Alice and
Bob, and dishonest Eve. Thus, in the following, all resources have
three interfaces, denoted $A$, $B$ and $E$, and a protocol is then
given by a pair of converters $(\pi_A,\pi_B)$ for the honest
players. We refer to \cite{MR11} for the general case, when arbitrary
players can be dishonest.

\begin{deff}[Cryptographic security~\cite{MR11}] \label{def:security}
  Let $\pi_{AB} = (\pi_A,\pi_B)$ be a protocol and $\aR_\sharp =
  (\aR,\sharp)$ and $\aS_\lozenge = (\aS,\lozenge)$ denote two
  filtered resources.  We say that \emph{$\pi_{AB}$ constructs
    $\aS_{\lozenge}$ from $\aR_\sharp$ within $\eps$}, which we write
  $\aR_\sharp \xrightarrow{\pi,\eps} \aS_\lozenge$, if the two
  following conditions hold:
\begin{enumerate}[label=\roman*), ref=\roman*]
\item \label{eq:def.cor} We have
  \[d(\pi_{AB}\aR\sharp_E,\aS\lozenge_E) \leq \eps \,.\]
\item \label{eq:def.sec} There exists a converter\footnote{For a
    protocol with information\-/theoretic security to be composable
    with a protocol that has computational security, one additionally
    requires the simulator to be efficient.} $\sigma_E$ \--- which we
  call simulator \--- such that
  \[  d(\pi_{AB}\aR,\cS\sigma_E) \leq \eps \,.\]
\end{enumerate}
If it is clear from the context what filtered resources $\aR_\sharp$
and $\aS_\lozenge$ are meant, we simply say that $\pi_{AB}$ is
$\eps$\=/secure.
\end{deff}

The first of these two conditions measures how close the constructed
resource is to the ideal resource in the case where no malicious
player is intervening, which is often called \emph{correctness} in the
literature. The second condition captures \emph{security} in the
presence of an adversary. For example, to prove that the MAC protocol
$\pi^\cauth_{AB}$ constructs an authentic channel $\aA_{\lozenge}$
from a (noiseless) insecure channel $\aC_\square$ and a secret key
$\aK$ within $\eps$, we need to prove that the real system (with
filters) $\pi^\cauth_{AB}(\aK\|\aC\square_E)$ cannot be distinguished
from the ideal system $\aA\lozenge_E$ with advantage greater than
$\eps$, and we need to find a converter $\sigma^\cauth_E$ such that
the real system (without filters) $\pi^\cauth_{AB}(\aK\|\aC)$ cannot
be distinguished from the ideal system $\aA\sigma^\cauth_E$ with
advantage greater than $\eps$. For the MAC protocol, correctness is
satisfied with error $0$ and the simulator $\sigma^\cauth_E$ drawn in
\figref{fig:classicalauth.ideal} satisfies the second requirement if
the family of hash functions $\{h_k\}_k$ is $\eps$\=/almost strongly
$2$\-/universal~\cite{Por14}.

\begin{figure}[tb]
\begin{centering}

\begin{tikzpicture}[
thinResource/.style={draw,thick,minimum width=1.618*2cm,minimum height=1cm},
simulator/.style={draw,thick,minimum width=1.618*2cm,minimum height=1.7cm,rounded corners},
snode/.style={minimum width=1.1cm,minimum height=1.2cm},
sArrow/.style={->,>=stealth,thick}]

\small

\def\t{2.368} 
\def\u{-1.1}
\def\v{.75}
\def\w{-2.45} 

\node[thinResource] (channel) at (0,\v) {};
\node[yshift=-1.5,above right] at (channel.north west) {\footnotesize
  Authentic channel $\aA$};
\node (alice) at (-\t,\v) {};
\node (bob) at (\t,\v) {};

\node[simulator] (sim) at (0,\u) {};
\node[xshift=1.5,below left] at (sim.north west) {\footnotesize
  $\sigma^{\cauth}_E$};
\node[snode,ellipse] (sleft) at (-.809,\u) {};
\node[snode] (sright) at (.809,\u) {};
\draw[dashed] (sim.north) to (sim.south);

\node (ajunc) at (sleft |- alice) {};
\node (bjunc) at (sright |- bob) {};

\draw[thick] (alice.center) to node[pos=.15,auto] {$x$} (.4,\v) to node[pos=.54] (ejunc) {} +(160:-.8);
\draw[sArrow] (ajunc.center) to node[pos=.63,auto,swap] {$x$} (sleft);
\draw[sArrow] (1.2,\v) to node[pos=.75,auto] {$x,\bot$} (bob.center);
\draw[double] (sright) to node[pos=.4,auto,swap] {$0,1$} (ejunc.center);

\node (sltext) at (-.809,\u) {\footnotesize $y = h_k(x)$};
\node[text width=1.2cm] (srtext) at (.809,\u) {\footnotesize $(x,y)
  \stackrel{?}{=}$\\$\quad \ (x',y')$};

\node (eveleft) at (sleft |- 0,\w) {};
\node (everight) at (sright |- 0,\w) {};
\draw[sArrow] (sleft) to node[pos=.75,auto,swap] {$(x,y)$} (eveleft.center);
\draw[sArrow] (everight.center) to node[pos=.25,auto,swap] {$(x',y')$}
(sright);

\node[draw,fill=white] (key) at (.15,\u+.8) {key};
\draw[sArrow] (key) to (sleft);

\end{tikzpicture}

\end{centering}
\caption[Ideal classical authentication
system]{\label{fig:classicalauth.ideal}The ideal system with simulator
  for a MAC protocol. The simulator $\sigma^\cauth_E$ picks its own
  key and generates the MAC. If the value input by Eve is different
  from the output at her interface (or is input before an output is
  generated), the simulator prevents Bob from getting Alice's
  message.}
\end{figure}
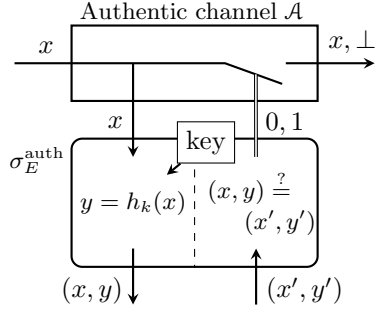

\begin{rem}
  \label{rem:efficiency} The protocols and simulators discussed in
  this work are all efficient. The protocols we consider are either
  trivially efficient or taken from other work, in which case we refer
  to these other works for proofs of efficiency. The efficiency of the
  simulator used to prove the security of quantum authentication has
  been analyzed in \cite{BW16}. All other simulators used in the
  security proofs run the corresponding honest protocols, and are thus
  efficient because the protocols are. We therefore do not discuss
  efficiency any further in this work.
\end{rem}

\section{Quantum authentication}
\label{sec:auth}

We start with some technical preliminaries in
\secref{sec:auth.technical}, where we introduce (strong) purity
testing codes, which are a key component of the family of quantum
authentication protocols of \cite{BCGST02}. In
\secref{sec:auth.secure} we give a constructive view of quantum
authentication with key recycling: we define the resources that such a
protocol is expected to construct, as well as the resources that are
required to achieve this. In \secref{sec:auth.protocol} we describe
the family of protocols that we analyze in this work, along with a
variant in which the order of the encryption and encoding operations
has been swapped, which we prove to be equivalent. In
\secref{sec:auth.proof} we give a security proof for the family of
quantum authentication protocols defined earlier. And in
\secref{sec:auth.optimality} we show that the number of recycled key
bits is optimal. Finally, in \secref{sec:auth.explicit} we give two
explicit constructions of purity testing codes and get the exact
parameters of the quantum authentication protocols with these codes.

\subsection{Technical preliminaries}
\label{sec:auth.technical}

\paragraph{Pauli operators.} To denote a Pauli operator on $n$ qubits
we write either $P_{x,z}$ or $P_\ell$, where $x$ and $z$ are $n$-bit
strings indicating in which positions bit and phase flips occur, and
$\ell = (x,z)$ is the concatenation of $x$ and $z$, which is used when
we do not need to distinguish between $x$ and $z$. Two Pauli operators
$P_{j}$ and $P_{\ell}$ with $j = (x,z)$ and $\ell = (x',z')$ commute
(anti\-/commute) if the symplectic inner product
\begin{equation}
\label{eq:symplectic}
\sym{j}{\ell} \coloneqq x \cdot z' - z \cdot x'
\end{equation} is $0$ (is $1$),
where $x \cdot z$ is the scalar product of the vectors and the arithmetic is done modulo
$2$. Hence, for any $P_j$ and $P_\ell$
\[ P_j P_\ell = (-1)^{\sym{j}{\ell}} P_\ell P_j\,.\]
We use several times the following equality
\begin{equation}
\label{eq:symplectic.trick}
\sum_{j \in \{0,1\}^n} (-1)^{\sym{j}{\ell}} = \begin{cases} 2^n & \text{if $\ell = 0$},
    \\ 0 & \text{otherwise}, \end{cases}
\end{equation}
where $\ell = 0$ means that all bits of the string $\ell$ are $0$.

\paragraph{Purity testing code.} 
An error correcting code (ECC) that encodes an $m$ qubit message in a
$m+n$ qubit code word is generally defined by an isomorphism from
$\complex^{2^m}$ to $\complex^{2^{m+n}}$. In this work we define an
ECC by a unitary $U : \complex^{2^{m+n}} \to \complex^{2^{m+n}}$.  The
code word for a state $\ket{\psi}$ is obtained by appending a $n$
qubit state $\zero$ to the message, and applying $U$, i.e., the
encoding of $\ket{\psi}$ is $U (\ket{\psi} \tensor \zero)$.  We do not
need to use the decoding properties of ECCs in this work, we only use
the them to detect errors, i.e., given a state
$\ket{\varphi} \in \complex^{2^{m+n}}$, we apply the inverse unitary
$\hconj{U}$ and measure the last $n$ qubits to see if they are $\zero$
or not.

The first property we require of our codes, is that they map any Pauli
error $P_\ell$ into another Pauli error $P_{\ell'}$, i.e., 
\begin{equation}
\label{eq:stabilizer.pauli} \hconj{U} P_\ell U = e^{i\theta_\ell} P_{\ell'} \,,
\end{equation}
for some global phase $e^{i\theta_\ell}$. This is always the case for
any $U$ that can be implemented with Clifford operators. In
particular, all stabilizer codes have this property, which are used in
\cite{BCGST02} to define purity testing codes. Note that the mapping
from $\ell$ to $\ell'$ defined by \eqnref{eq:stabilizer.pauli} is a
permutation on the set of indices $\ell \in \{0,1\}^{2m+2n}$ that
depends only on the choice of code.

A code will detect an error $P_\ell$ if
$P_{\ell'} = P_{x,z} \tensor P_{s,z'}$ for $s \neq 0$, where $P_{x,z}$
acts on the first $m$ qubits and $P_{s,z'}$ on the last $n$. Measuring
these last qubits would yield the syndrome $s$, since $P_{s,z'}$ flips
the bits in the positions corresponding to the bits of $s$. And an
error $P_\ell$ will act trivially on the message if
$P_{\ell'} = P_{0,0} \tensor P_{s,z}$. In particular, if
$P_{\ell'} = P_{0,0} \tensor P_{0,z}$, then this error will not be
detected, but not change the message either.

For a code indexed by a key $k$, we denote by $\cP_k$ the set of Pauli
errors that are not detected by this code, and by
$\cQ_k \subset \cP_k$ we denote the undetected errors which act
trivially on the message. A purity testing code is a set of codes
$\{U_k\}_{k \in \cK}$ such that when a code $U_k$ is selected
uniformly at random, it will detect with high probability all Pauli
errors which act non\-/trivially on the message.

\begin{deff}[Purity testing code~\cite{BCGST02}.] 
  A purity testing code with error $\eps$ is a set of codes
  $\{U_k\}_{k \in \cK}$, such that for all Pauli operators $P_\ell$,
  \[ \frac{\left|\left\{ k \in \cK : P_\ell \in \cP_k \setminus \cQ_k
      \right\}\right|}{|\cK|} \leq \eps \,.\]
\end{deff}

As mentioned in \secref{sec:intro.contributions}, we use a stricter
definition of purity testing code in this work. We require that all
non\-/identity Paulis get detected with high probability, even those
that act trivially on the message. Intuitively, the reason for this is
that, with the original definition of purity testing, if the adversary
introduces some noise $P_\ell$, by learning whether the message was
accepted or not, she will learn whether that error acts trivially on
the message or not, and thus learn something about the ECC used. This
means that the adversary learns something about the key used to choose
the ECC, and hence it cannot be recycled in its entirety.\footnote{We
  conjecture that in this case only $1$ bit of the key is leaked, see
  the discussion in \secref{sec:conclusion}.}

\begin{deff}[Strong purity testing code.] 
  A strong purity testing code with error $\eps$ is a set of codes
  $\{U_k\}_{k \in \cK}$, such that for all non\-/identity Pauli
  operators $P_\ell$,
  \[ \frac{\left|\left\{ k \in \cK : P_\ell \in \cP_k
      \right\}\right|}{|\cK|} \leq \eps \,.\]
\end{deff}

In \secref{sec:auth.explicit} we provide explicit constructions of
strong purity testing codes.

\subsection{Secure channel \& secret key resource}
\label{sec:auth.secure}

The main result in this paper is a proof that the family of quantum
authentication protocols of Barnum et al.~\cite{BCGST02} restricted to
strong purity testing codes can be used to construct a resource that
corresponds to the parallel composition of a secure quantum channel
$\aS^m$ and a secret key resource $\bar{\aK}^{\nu_{\rej},\nu_{\acc}}$, which
are illustrated in \figref{fig:goal.new} and explained in more detail
in the following paragraphs.

\begin{figure}[tbp]
  \subcaptionbox[Secure channel]{\label{fig:channel.secure}A secure
    channel $\aS^m$ is very similar to the authentic channel from
    \figref{fig:channel.dishonest.authentic}. It allows Alice to send
    an $m$\=/qubit message, and Eve to decide if Bob gets it. But this
    time, Eve only receives the size of the message that has been sent
    (denoted by the dashed arrow), not a copy.}[.5\textwidth][c]{
\begin{tikzpicture}[
resource/.style={draw,thick,minimum width=3.2cm,minimum height=1.6cm},
sArrow/.style={->,>=stealth,thick},
sLine/.style={-,thick},
sLeak/.style={->,>=stealth,thick,dashed}]

\small

\def\t{2.3} 
\def\w{1.5} 
\def\s{.4}

\node[resource] (ch) at (0,0) {};
\node[yshift=-1.5,above right] at (ch.north west) {\footnotesize
  Secure channel $\aS^m$};
\node (alice) at (-\t,0) {};
\node[yshift=-1.5,below] at (alice) {\footnotesize
  Alice};
\node (bob) at (\t,0) {};
\node[yshift=-1.5,below] at (bob) {\footnotesize
  Bob};
\node (eve) at (0,-\w) {\footnotesize
  Eve};

\draw[sLine] (alice) to node[auto,pos=.15] {$\rho$} (0,0) to node (handle) {} (330:2*\s);
\draw[sArrow] (2*\s,0) to node[auto,pos=.85] {$\rho,\bot$} (bob);
\draw[sLeak] (-\s,0) to node[auto,swap,pos=.82] {$m$}(-\s,0 |- eve.north);
\draw[double,thick] (handle.center |- eve.north) to node[auto,swap,pos=.15] {$0,1$} (handle.center);

\end{tikzpicture}} \subcaptionbox[Secret key resource of variable
length]{\label{fig:resource.key.switch}A slightly weaker secret key
  resource than that from \figref{fig:resource.key},
  $\bar{\aK}^{\nu_{\rej},\nu_{\acc}}$. It allows Eve to choose the
  length of the key generated, either $|k| = \nu_{\rej}$ or
  $|k| = \nu_{\acc}$. Furthermore, Eve can prevent Alice from getting
  the key at all.}[.5\textwidth][c]{
\begin{tikzpicture}[
resource/.style={draw,thick,minimum width=3.2cm,minimum height=1.6cm},
sArrow/.style={->,>=stealth,thick},
sLine/.style={-,thick}]]

\small

\def\t{2.3} 
\def\v{.3}
\def\w{1.5} 
\def\s{.4}

\node[resource] (keyBox) at (0,0) {};
\node[draw] (key) at (\s,-\v) {key};
\node[yshift=-1.5,above right] at (keyBox.north west) {\footnotesize
  Secret key $\bar{\aK}^{\nu_{\rej},\nu_{\acc}}$};
\node (alice) at (-\t,0) {};
\node (bob) at (\t,0) {};
\node (up) at (0,\v) {};
\node (down) at (0,-\v) {};
\node (bob) at (\t,0) {};
\node (alice) at (-\t,0) {};
\node[yshift=-1.5,below] at (alice |- down) {\footnotesize
  Alice};
\node (bob) at (\t,0) {};
\node[yshift=-1.5,below] at (bob |- down) {\footnotesize
  Bob};
\node (eve) at (0,-\w) {\footnotesize
  Eve};

\draw[sArrow] (alice |- up) to node[auto] {\footnotesize req.}
(keyBox.west |- up);
\draw[sArrow] (bob |- up) to node[auto,swap] {\footnotesize req.}
(keyBox.east |- up);
\draw[sLine] (key.west |- down) to (0,0 |- down) to node (handle) {} +(210:2*\s);
\draw[sArrow] (-2*\s,0 |- down) to node[auto,pos=.8,swap] {$k,\bot$} (alice |- down);
\draw[sArrow] (key.east |- down) to node[auto,pos=.78] {$k$} (bob |- down);
\draw[double,thick] (handle.center |- eve.north) to node[auto,pos=.25] {$0,1$} (handle.center);
\draw[double,thick] (key.center |- eve.north) to node[auto,swap,pos=.25] {$0,1$} (key);

\end{tikzpicture}}

\vspace{6pt}

\subcaptionbox[No adversary]{\label{fig:goal.honest}When no adversary
  is present, the filter $\flat_E$ covers Eve's interface of the
  resource $\aS^m \| \bar{\aK}^{\nu_{\acc},\nu_{\rej}}$. Once $\flat_E$ is
  notified that a message has been sent, it allows the message through
  and notifies the secret key resource to prepare a key of length
  $\nu_{\acc}$.}[\textwidth][c]{
\begin{tikzpicture}[
resource/.style={draw,thick,minimum width=3.2cm,minimum height=1.6cm},
sArrow/.style={->,>=stealth,thick},
sLine/.style={-,thick},
sLeak/.style={->,>=stealth,thick,dashed},
filter/.style={draw,thick,minimum width=7.9cm,rounded corners}]

\small

\def\t{2.3} 
\def\v{.3}
\def\w{1.5} 
\def\s{.4}
\def\x{5.5}

\node[filter] (fil) at (\x/2,-\w) {};
\node[xshift=-1.5,right] at (fil.east) {\footnotesize $\flat_E$};

\node[resource] (ch) at (0,0) {};
\node[yshift=-1.5,above right] at (ch.north west) {\footnotesize
  Secure channel $\aS^m$};
\node (alice) at (-\t,0) {};
\node (bob) at (\t,0) {};

\draw[sArrow] (alice) to node[auto,pos=.07] {$\rho$} node[auto,pos=.93] {$\rho$} (bob);
\draw[sLeak] (-\s,0) to node[auto,swap,pos=.8] {$m$} (-\s,0 |- fil.north);
\draw[double,thick] (\s,0) to node[auto,pos=.8] {$0$} (\s,0 |- fil.north);

\node[resource] (keyBox) at (\x,0) {};
\node[draw] (key) at (\s+\x,-\v) {key};
\node[yshift=-1.5,above right] at (keyBox.north west) {\footnotesize
  Secret key $\bar{\aK}^{\nu_{\rej},\nu_{\acc}}$};
\node (alice2) at (-\t+\x,0) {};
\node (bob2) at (\t+\x,0) {};
\node (up) at (0,\v) {};
\node (down) at (0,-\v) {};

\draw[sArrow] (alice2 |- up) to node[auto] {\footnotesize req.}
(keyBox.west |- up);
\draw[sArrow] (bob2 |- up) to node[auto,swap] {\footnotesize req.}
(keyBox.east |- up);
\draw[sArrow] (key.west |- down) to node[auto,pos=.85,swap] {$k$} (alice2 |- down);
\draw[sArrow] (key.east |- down) to node[auto,pos=.78] {$k$}  (bob2 |- down);
\draw[double,thick] (\x-\s,0 |- fil.north) to node[auto,swap,pos=.25] {$0$} (\x-\s,-\v);
\draw[double,thick] (\s+\x,0 |- fil.north) to node[auto,swap,pos=.32] {$0$} (key);

\end{tikzpicture}}

\caption[Resources constructed]{\label{fig:goal.new}We depict here the
  filtered resource
  $(\aS^m \| \bar{\aK}^{\nu_{\acc},\nu_{\rej}},\flat_E)$ constructed
  by the quantum authentication protocols analyzed in this work. It
  can be seen as the composition of a secure channel $\aS^m$ (drawn in
  \subref{fig:channel.secure}) and a secret key resource
  $\bar{\aK}^{\nu_{\acc},\nu_{\rej}}$ (drawn in
  \subref{fig:resource.key.switch}). The filter $\flat_E$ that
  emulates an honest behavior is drawn in \subref{fig:goal.honest}.}
\end{figure}

The secure quantum channel, $\aS^m$, drawn in
\figref{fig:channel.secure}, allows an $m$\=/qubit message $\rho$ to be
transmitted from Alice to Bob, which Alice may input at her
interface. Since in general the players cannot prevent Eve from
learning that a message has been sent, Eve's interface has one output
denoted by a dashed arrow, which notifies her that Alice has sent an
$m$\=/qubit message. But the players cannot prevent Eve from jumbling
the communication lines either, which is captured in the resource
$\aS^m$ by allowing Eve to input a bit that decides if Bob gets the
message or an error symbol $\bot$ \--- Eve may also decide not to
provide this input (Eve cuts the communication lines), in which case
the system is left waiting and Bob obtains neither the message nor an
error. Note that the order in which messages are input to the resource
$\aS^m$ is not fixed, Eve may well provide her bit before Alice inputs
a message. In this case, Bob immediately receives an error $\bot$
regardless of the value of Eve's bit.

The secret key resource, $\bar{\aK}^{\nu_{\rej},\nu_{\acc}}$, depicted
in \figref{fig:resource.key.switch} distributes a uniformly random key
to Alice and Bob. Unlike the simplified key resource from
\figref{fig:resource.key}, here the adversary has some control over
the length of the key produced. This is because in the real setting
Eve can prevent the full key from being recycled by jumbling the
message. This is reflected at Eve's interface of
$\bar{\aK}^{\nu_{\rej},\nu_{\acc}}$ allowing her to decide if the key
generated is of length $\nu_{\rej}$ or $\nu_{\acc}$.  Furthermore, if
in the real setting Alice were to recycle her key before Bob receives
the cipher, Eve could use the information from the recycled key to
modify the cipher without being detected. So Alice must wait for a
confirmation of reception from Bob, which Eve can jumble, preventing
Alice from ever recycling the key. This translates in the ideal
setting to Eve having another control bit, deciding whether Alice
receives the key or an error $\bot$. Note that if Eve provides her two
bits in the wrong order, Alice always gets an error $\bot$. This key
resource is modeled so that the honest players must request the key to
obtain its value. If Bob does this before Eve has provided the bit
deciding the key length, he gets an error instead of a key. If Alice
makes the request before Eve has provided both her bits, she also gets
an error. Otherwise they get the key $k$.

If no adversary is present, a filter $\flat_E$ covers Eve's interface
of the resources $\aS^m$ and $\bar{\aK}^{\nu_{\rej},\nu_{\acc}}$,
which is drawn in \figref{fig:goal.honest}. This filter provides the
inputs to the resources that allow Bob to get Alice's message and
generate a key of length $\nu_{\acc}$ that is made available to both
players.

To construct the filtered resource
$(\aS^m \| \bar{\aK}^{\nu_{\rej},\nu_{\acc}})_\flat$, the quantum
authentication protocol will use a shared secret key to encrypt and
authenticate the message. This means that the players must share a
secret key resource. For simplicity we assume the players have access
to a resource $\aK^\mu$ as depicted in \figref{fig:resource.key}, that
always provides them with a key of length $\mu$.\footnote{Since Eve's
  interface of $\aK^\mu$ is empty, this resource has a trivial empty
  filter, which we do not write down.} Note that the security of the
protocol is not affected if the players only have a weaker resource
which might shorten the key or not deliver it to both players \---
such as the one constructed by the protocol, namely
$\bar{\aK}^{\nu_{\rej},\nu_{\acc}}$ \--- because if either of the
players does not have enough key, they simply abort, which is an
outcome Eve could already achieve by cutting or jumbling the
communication.

They also need to share an insecure quantum channel, which is used to
send the message, and is illustrated in
\figref{fig:channel.dishonest.insecure} without a filter and in
\figref{fig:channel.honest.insecure} with a filter. The authentication
protocol we consider is designed to catch any error, so if it is used
over a noisy channel, it will always abort, even though no adversary
is tampering with the message. We thus assume that the players share a
noiseless channel, which we denote $\aC_{\square}$, i.e., $\aC$ is
controlled by the adversary as in
\figref{fig:channel.dishonest.insecure}. But if no adversary is
present, the filter $\square_E$ is noiseless. In
\secref{sec:construction.noiseless} we explain how to compose the
protocol with an error correcting code so as to run it over a noisy
channel.

Finally, the players need a backwards authentic channel, that can send
one bit of information from Bob to Alice. This is required so that
Alice may learn whether the message was accepted and recycle the
corresponding amount of key. The authentic channel and its filter
$\aA_\lozenge$ are drawn in
Figures~\ref{fig:channel.dishonest.authentic} and
\ref{fig:channel.honest.authentic}. Putting all this together in the case
of an active adversary, we get \figref{fig:auth.real}, where the
converters for Alice's and Bob's parts of the quantum authentication
protocol are labeled $\pi^{\qauth}_A$ and $\pi^{\qauth}_B$,
respectively.

\begin{figure}[tb]
\begin{centering}

\begin{tikzpicture}[
resourceLong/.style={draw,thick,minimum width=3.5cm,minimum height=1cm},
resource/.style={draw,thick,minimum width=1.8cm,minimum height=1cm},
sArrow/.style={->,>=stealth,thick},
sLine/.style={-,thick},
protocol/.style={draw,thick,minimum width=1.6cm,minimum height=4cm,rounded corners},
pnode/.style={minimum width=1cm,minimum height=1cm}]

\small

\def\t{4.6} 
\def\a{3.05} 
\def\v{1.5}
\def\w{.4}
\def\x{.85}
\def\z{2.75} 

\node[resourceLong] (auth) at (0,0) {};
\node[yshift=-1.5,above right] at (auth.north west) {\footnotesize
  Authentic channel $\aA$};
\node[resource] (ch) at (-\x,-\v) {};
\node[yshift=-1.5,above] at (ch.north) {\footnotesize
  Insecure ch.\ $\aC$};
\node[resourceLong] (key) at (0,\v) {};
\node[yshift=-1.5,above right] at (key.north west) {\footnotesize
  Secret key $\aK^\mu$};
\node[protocol] (protA) at (-\a,0) {};
\node[yshift=-1.5,above right] at (protA.north west) {\footnotesize $\pi^\qauth_A$};
\node[pnode] (a1) at (-\a,\v) {};
\node[pnode] (a2) at (-\a,0) {};
\node[pnode] (a3) at (-\a,-\v) {};
\node[protocol] (protB) at (\a,0) {};
\node[yshift=-1.5,above left] at (protB.north east) {\footnotesize $\pi^\qauth_B$};
\node[pnode] (b1) at (\a,\v) {};
\node[pnode] (b2) at (\a,0) {};
\node[pnode] (b3) at (\a,-\v) {};

\node (aliceUp) at (-\t,\v) {};
\node (aliceMiddle) at (-\t,0) {};
\node (aliceDown) at (-\t,-\v) {};
\node (bobUp) at (\t,\v) {};
\node (bobMiddle) at (\t,0) {};
\node (bobDown) at (\t,-\v) {};
\node (eveLeftL) at (-\x-\w,-\z) {};
\node (eveLeftR) at (-\x+\w,-\z) {};
\node (eveRightL) at (\x-\w,-\z) {};
\node (eveRightR) at (\x+\w,-\z) {};

\draw[sArrow] (aliceDown) to node[auto,pos=.4] {$\rho$} (a3);
\draw[sArrow] (a3) to (eveLeftL |- aliceDown) to (eveLeftL);
\draw[sArrow] (eveLeftR) to 
(eveLeftR |- bobDown) to (b3);
\draw[sArrow] (b3) to node[auto,pos=.8] {$\rho',\bot$} (bobDown);

\draw[sArrow] (aliceUp) to node[auto,pos=.3] {\footnotesize req.} (a1);
\draw[sArrow] (bobUp) to node[auto,swap,pos=.3] {\footnotesize req.} (b1);
\draw[sArrow] (a2) to node[auto,swap,pos=.8] {$k',\bot$} (aliceMiddle);
\draw[sArrow] (b2) to node[auto,pos=.6] {$k'$} (bobMiddle);

\node[draw] (key) at (0,\v) {key};
\draw[sArrow,bend left=15] (key) to node[auto,swap,pos=.2] {$k$} (a1);
\draw[sArrow,bend right=15] (key) to node[auto,pos=.2] {$k$} (b1);
\node[pnode] (kLeft) at (-1.1,\v) {};
\node[pnode] (kRight) at (1.1,\v) {};
\draw[sArrow,bend left=25] (a1) to node[auto,pos=.5] {\footnotesize req.} (kLeft);
\draw[sArrow,bend right=25] (b1) to node[auto,swap,pos=.5] {\footnotesize req.} (kRight);

\draw[sLine] (b2) to (\x,0) to node[pos=.53] (handle) {} +(200:2*\w);
\draw[sArrow] (\x-2*\w,0) to (a2);
\draw[sArrow] (eveRightR |- bobMiddle) to 
(eveRightR);
\draw[double] (eveRightL) to node[auto,swap,pos=.1,xshift=-2] {$0,1$} (handle.center);

\end{tikzpicture}

\end{centering}
\caption[Real system for quantum
authentication]{\label{fig:auth.real}The real system for quantum
  authentication with key recycling. Upon receiving a message $\rho$,
  $\pi^\qauth_A$ encrypts it with a key that it obtains from $K^\mu$
  and sends it on the insecure channel. Upon receiving a quantum state
  on the insecure channel, $\pi^\qauth_B$ checks whether it is valid,
  and outputs the corresponding message $\rho'$ or an error message
  $\bot$. It may then recycle (part of) the key, $k'$, and uses the
  authentic channel to notify $\pi^{\qauth}_A$ whether the message was
  accepted or not.  $\pi^{\qauth}_A$ then recycles the key as
  well. Concrete protocols for this are given in
  \secref{sec:auth.protocol}.}
\end{figure}
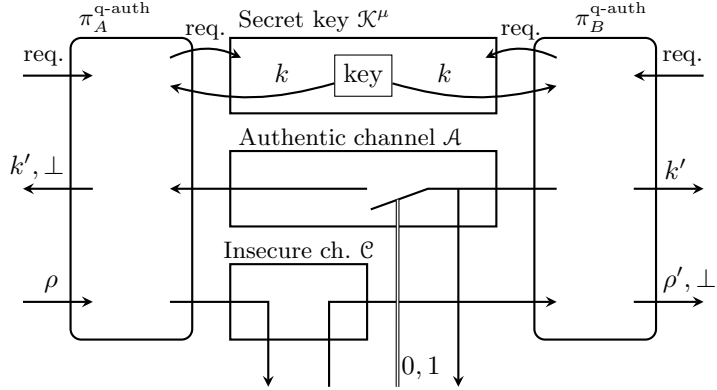

According to \defref{def:security}, a protocol
$\pi^\qauth_{AB} = (\pi^\qauth_A,\pi^\qauth_B)$ is then a quantum
authentication protocol (with key recycling) with error
$\eps^{\qauth}$ if it constructs
$(\aS^m \| \bar{\aK}^{\nu_{\rej},\nu_{\acc}})_\flat$ from
$\aC_\square \| \aA_\lozenge \| \aK^\mu$ within $\eps^{\qauth}$, i.e.,
\begin{equation}
\label{eq:security.auth}
\aC_\square \| \aA_\lozenge \| \aK^{\mu}
\xrightarrow{\pi^\qauth_{AB},\eps^{\qauth}}
(\aS^m \| \bar{\aK}^{\nu_{\rej},\nu_{\acc}})_\flat\,.
\end{equation}

In \secref{sec:auth.protocol} we describe the protocol, and in
\secref{sec:auth.proof} we prove that \eqnref{eq:security.auth} is
satisfied and provide the parameters
$\mu,\nu_\rej,\nu_\acc,\eps^\qauth$.

\subsection{Generic protocol}
\label{sec:auth.protocol}

The family of quantum authentication protocols from \cite{BCGST02}
consists in first encrypting the message to be sent with a quantum
one-time pad, then encoding it with a purity testing code and a random
syndrome. We do the same, but with a strong purity testing code. We
also extend the protocol so that the players recycle all the key if
the message is accepted, and the key used to select the strong purity
testing code if the message is rejected. So that Alice may also
recycle the key, Bob uses the backwards authentic classical channel to
notify her of the outcome. We refer to this as the
``encrypt\-/then\-/encode'' protocol, the details of which are
provided in \figref{fig:protocol.1}.

\begin{figure}[tb]
\begin{boxenvtitle}{Quantum authentication \--- encrypt\-/then\-/encode}
\begin{enumerate}
\item Alice and Bob obtain uniform keys $k$, $\ell$, and $s$ from the
  key resource, where $k$ is long enough to choose an element from a
  strong purity testing code that encodes $m$ qubits in $m+n$ qubits,
  $\ell$ is $2m$ bits and $s$ is $n$ bits.
\item Alice encrypts the message $\rho^A$ she receives with a quantum
  one-time pad using the key $\ell$. She then appends an $n$ qubit
  state $\proj{s}^S$, and encodes the whole thing with a strong purity
  testing code, obtaining the cipher
  $\sigma^{AS} = U_k (P_\ell \rho^A P_\ell \tensor \proj{s}^S)
  \hconj{U}_k$.
\item Alice sends $\sigma^{AS}$ to Bob on the insecure channel.
\item Bob receives a message $\tilde{\sigma}^{AS}$, he applies
  $\hconj{U}_k$, decrypts the $A$ part and measures the $S$ part in
  the computational basis.
\item If the result of the measurement is $s$, he accepts the message
  and recycles $k$, $\ell$ and $s$. If the result is not $s$, he
  rejects the message, and recycles $k$.
\item Bob sends Alice a bit on the backwards authentic channel to tell
  her if he accepted or rejected the message.
\item When Alice receives Bob's bit, she either recycles all the keys
  or only $k$.
\end{enumerate}
\end{boxenvtitle}

\caption[Quantum authentication protocol \---
encrypt\-/then\-/encode]{\label{fig:protocol.1}This protocol is
  identical to the scheme from \cite{BCGST02}, except that the players
  use a strong purity testing code, recycle key, and have a backwards
  authentic channel so that Alice may learn the outcome.}
\end{figure}

Alternatively, one may perform the encoding and encryption in the
opposite order: Alice first encodes her message with the strong purity
testing code with syndrome $0$, then does a quantum one-time pad on
the resulting $m+n$ qubit state. This ``encode\-/then\-/encrypt''
protocol is described in \figref{fig:protocol.2}.

\begin{figure}[tb]
\begin{boxenvtitle}{Quantum authentication \--- encode\-/then\-/encrypt}
\begin{enumerate}
\item Alice and Bob obtain uniform keys $k$ and $\ell$ from the key
  resource, where $k$ is long enough to choose an element from a
  strong purity testing code that encodes $m$ qubits in $m+n$ qubits
  and $\ell$ is $2m+2n$ bits long.
\item Alice appends a $n$ qubit state $\proj{0}$ to the message
  $\rho^A$ she receives, encodes it with a strong purity testing code
  chosen according to the key $k$, and encrypts the whole thing with a
  quantum one-time pad using the key $\ell$. She thus obtains the
  cipher
  $\sigma^{AS} = P_\ell U_k (\rho^A \tensor \proj{0}^S) \hconj{U}_k
  P_\ell$.
\item Alice sends $\sigma^{AS}$ to Bob on the insecure channel.
\item Bob receives a message $\tilde{\sigma}^{AS}$, he applies
  $P_\ell$, then
  $\hconj{U}_k$, and measures the $S$ part in
  the computational basis.
\item If the result of the measurement is $0$, he accepts the message
  and recycles $k$ and $\ell$. Otherwise, he rejects
  the message, and recycles $k$.
\item Bob sends Alice a bit on the backwards authentic channel to tell
  her if he accepted or rejected the message.
\item When Alice receives Bob's bit, she either recycles all the keys
  or only $k$.
\end{enumerate}
\end{boxenvtitle}

\caption[Quantum authentication protocol \---
encode\-/then\-/encrypt]{\label{fig:protocol.2}This protocol is
 similar to the protocol from \figref{fig:protocol.1}, except that the
 order of the encryption and encoding have been reversed. To do this,
 the players need an extra $n$ bits of key.}
\end{figure}

The pseudo\-/code described in Figures~\ref{fig:protocol.1} and
\ref{fig:protocol.2} can easily be translated into converters as used
in the AC formalism, i.e., the objects $\pi^{\qauth}_A$ and
$\pi^{\qauth}_B$ from \figref{fig:auth.real}. More precisely, if
$\pi^{\qauth}_A$ receives a message at its outer interface, it
requests a key from the key resource, encrypts the message as
described and sends the cipher on the insecure channel. It may receive
three symbols from the backwards authentic channel: an error $\bot$,
in which case it does not recycle any key, a message $0$ saying that
$\pi^{\qauth}_B$ did not receive the correct state, in which case it
recycles the part of the key used to choose the code, or a message $1$
saying that $\pi^{\qauth}_B$ did receive the correct state, in which
case it recycles all the key. If $\pi^{\qauth}_A$ first receives a
message on the backwards authentic channel before receiving a message
to send, it will not recycle any key. Similarly, when $\pi^{\qauth}_B$
receives a cipher on the insecure channel, it requests a key from the
key resource, performs the decryption, outputs either the message or
an error depending on the result of the decryption, and sends this
result back to $\pi^{\qauth}_A$ on the authentic channel.

The encode\-/then\-/encrypt protocol uses $n$ bits more key, and since
these bits are not recycled in case of a reject, it is preferable to
use the encrypt\-/then\-/encode protocol. These protocols are however
identical: no external observer can detect which of the two is being
run. This holds, because the encode\-/then\-/encrypt protocol performs
phase flips on a syndrome that is known to be in a computational basis
state $\ket{s}$. Thus, they have no effect and can be
skipped. Likewise, Bob performs phase flips on $S$ before measuring in
the computational basis \--- he might as well skip these phase flips,
since they have no effect either. We formalize this statement by
proving (in \lemref{lem:ete}) that the converters corresponding to the
two different protocols are indistinguishable. This result is similar
in spirit to proofs that some prepare\-/and\-/measure quantum key
distribution (QKD) protocols are indistinguishable from
entanglement-\/based QKD protocols, and thus security proofs for one
are security proofs for the other~\cite{SP00}.

Since these two protocols are indistinguishable, we provide a security
proof in \secref{sec:auth.proof} for the encode\-/then\-/encrypt
protocol. However, in \secref{sec:auth.explicit}, when we count the
number of bits of key consumed, we count those of the
encrypt\-/then\-/encode protocol.

\begin{lem}
\label{lem:ete}
Let $(\bar{\pi}^\qauth_A,\bar{\pi}^\qauth_B)$ and
$(\pi^\qauth_A,\pi^\qauth_B)$ denote the pairs of converters modeling
Alice's and Bob's behavior in the encrypt\-/then\-/encode and
encode\-/then\-/encrypt protocols, respectively. Then 
\[d(\bar{\pi}^\qauth_A,\pi^\qauth_A) = d(\bar{\pi}^\qauth_B,\pi^\qauth_B) = 0\,.\]
\end{lem}

\begin{proof}
We start with Alice's part of the protocol. Let $\bar{\pi}^\qauth_A$
and $\pi^\qauth_A$ receive keys $k$, $\ell$ and $s$ as in the protocol
from \figref{fig:protocol.1}, as well as an extra key $z$ of length
$n$ that is needed by $\pi^\qauth_A$, since it requires more key. The
distinguisher prepares a state $\rho^{RA}$, and sends the $A$ part to
the system. $\bar{\pi}^\qauth_A$ outputs
\begin{multline*}
  U^{AS}_k P^A_\ell \left( \rho^{RA} \tensor \proj{s}^S \right)
  P^A_\ell \hconj{\left(U^{AS}_k\right)} \\ \begin{split} & = U^{AS}_k
    \left(P^A_\ell \tensor P^S_{s,0} \right) \left( \rho^{RA} \tensor
      \proj{0}^S \right) \left(P^A_\ell \tensor P^S_{s,0} \right)
    \hconj{\left(U^{AS}_k\right)} \\
    & = U^{AS}_k \left(P^A_\ell \tensor P^S_{s,z} \right) \left(
      \rho^{RA} \tensor \proj{0}^S \right) \left(P^A_\ell \tensor
      P^S_{s,z} \right)
    \hconj{\left(U^{AS}_k\right)} \\
    & = P^{AS}_{\ell'}U^{AS}_k \left( \rho^{RA} \tensor \proj{0}^S
    \right) \hconj{\left(U^{AS}_k\right)}
    P^{AS}_{\ell'}\,, \end{split}
\end{multline*}
where in the last line we used \eqnref{eq:stabilizer.pauli}. This is
exactly the state output by $\pi^\qauth_A$ if when receiving the key
$k,\ell,s,z$, the protocol uses the Pauli $P_{\ell'}$ for the quantum
one-time pad.

For Bob's part of the protocol, let the distinguisher prepare a state
$\sigma^{RAS}$ and send the $AS$ part to the system. The subnormalized
state held jointly by $\bar{\pi}^\qauth_B$ and the distinguisher after
decoding and performing the measurement is given by
\begin{multline*}
  \bra{s} P^A_\ell \hconj{\left(U^{AS}_k\right)} \sigma^{RAS} U^{AS}_k
  P^A_\ell \ket{s} \\ \begin{split} & = \bra{0} \left(P^A_\ell \tensor P^S_{s,0}\right) \hconj{\left(U^{AS}_k\right)} \sigma^{RAS} U^{AS}_k
                     \left(P^A_\ell \tensor P^S_{s,0}\right) \ket{0} \\
                   & = \bra{0} \left(P^A_\ell \tensor P^S_{s,z}\right) \hconj{\left(U^{AS}_k\right)} \sigma^{RAS} U^{AS}_k
                     \left(P^A_\ell \tensor P^S_{s,z}\right) \ket{0} \\
                   & = \bra{0}
                     \hconj{\left(U^{AS}_k\right)} P^{AS}_{\ell'} \sigma^{RAS} P^{AS}_{\ell'} U^{AS}_k
                     \ket{0} \,. \end{split}
\end{multline*}
We again obtain the state that is jointly held by $\pi^\qauth_B$ and
the distinguisher if when receiving the key $k,\ell,s,z$, the protocol
uses the Pauli $P_{\ell'}$ for the quantum one-time pad.
\end{proof}

\begin{rem}
\label{rem:classicalmessage}
If part of the message is classical \--- i.e., it is
diagonal in the computational basis and known not to have a purification
held be the distinguisher \--- then running the same proof as
\lemref{lem:ete}, one can show that it is sufficient to perform bit flips
on that part of the message, the phase flips are unnecessary. This is
the case with the protocol from \secref{sec:construction.morekey},
that generates a key $x$ locally and sends it to Bob using a quantum
authentication scheme. We use this to save some bits of key in the bounds
from \corref{cor:construction.1} and \corref{cor:construction.2}.
\end{rem}

\subsection{Security proof}
\label{sec:auth.proof}

As stated in \secref{sec:auth.secure} we wish to prove that the
quantum authentication protocol considered constructs a secure channel
$\aS^m$ and secret key resource
$\bar{\aK}^{\nu_{\rej},\nu_{\acc}}$. We will however prove this as a
corollary of a theorem that makes a slightly stronger statement,
namely that the protocol constructs a filtered resource
$\aQ^{m,\nu_{\rej},\nu_{\acc}}_\natural$ (depicted in
\figref{fig:goal.alt}) that is equivalent to
$(\aS^m \| \bar{\aK}^{\nu_{\rej},\nu_{\acc}})_\flat$, except for the
fact that it provides one switch at the adversarial interface that
decides both whether the message is delivered and the length of the
recycled key. $\bar{\aK}^{\nu_{\rej},\nu_{\acc}}$ and $\aS^m$ each
provide Eve with an input bit to decide the length of the new key and
whether the message is delivered, respectively, but in the protocol
the two bits are correlated, since the players only recycle the full
key if the message is successfully authenticated. One can thus make a
slightly stronger statement, in which the ideal resource constructed
only allows Eve to input one bit that decides both these things, which
is what is achieved by $\aQ^{m,\nu_{\rej},\nu_{\acc}}_\natural$.

\begin{figure}[tb]
\begin{centering}

\begin{tikzpicture}[
resource/.style={draw,thick,minimum width=3.2cm,minimum height=2cm},
sArrow/.style={->,>=stealth,thick},
sLine/.style={-,thick},
sLeak/.style={->,>=stealth,thick,dashed},
filter/.style={draw,thick,minimum width=2.4cm,rounded corners}]

\small

\def\t{2.35} 
\def\v{.6}
\def\w{.4}
\def\x{.8}
\def\z{1.75} 
\def\s{5.7}

\node[resource] (channel) at (0,0) {};
\node[yshift=-1.5,above right] at (channel.north west) {\footnotesize
  Secure ch.\ \& key $\aQ^{m,\nu_{\rej},\nu_{\acc}}$};
\node (aliceUp) at (-\t,\v) {};
\node (aliceMiddle) at (-\t,0) {};
\node (aliceDown) at (-\t,-\v) {};
\node (bobUp) at (\t,\v) {};
\node (bobMiddle) at (\t,0) {};
\node (bobDown) at (\t,-\v) {};
\node (eveLeft) at (-\x,-\z) {};
\node (eveCenter) at (0,-\z) {};
\node (eveRight) at (\x,-\z) {};

\draw[sLine] (aliceDown) to node[auto,pos=.1,yshift=-1] {$\rho$} (\x-\w,-\v) to +(340:2*\w);
\draw[sArrow]  (\x+\w,-\v) to node[auto,pos=.8,yshift=-1] {$\rho,\bot$} (bobDown);
\draw[sLeak] (eveLeft |- aliceDown) to node[auto,swap,pos=.7] {$m$}
(eveLeft);

\node[draw] (key) at (\x,0) {key};
\draw[sArrow] (key) to node[auto,pos=.2] {$k$} (bobMiddle);
\draw[sLine] (key) to (\w,0) to node[auto,swap,pos=.1] {$k$}
node[pos=.53] (handle) {} +(200:2*\w);
\draw[sArrow] (-\w,0) to (aliceMiddle);

\draw[double] (eveRight) to node[pos=.19,xshift=-1,auto,swap] {$0,1$}
node[circle,fill,pos=.65,inner sep=1.5] {}  (key);
\draw[double] (eveCenter) to node[pos=.17,xshift=-1,auto,swap]
{$0,1$} (handle.center);

\draw[sArrow] (aliceUp) to node[auto,yshift=-1] {\footnotesize req.} (channel.west |- aliceUp);
\draw[sArrow] (bobUp) to node[auto,swap,yshift=-1] {\footnotesize req.} (channel.east |- bobUp);

\node[resource] (channelS) at (\s,0) {};
\node[yshift=-1.5,above right] at (channelS.north west) {\footnotesize
  Secure ch.\ \& key $\aQ^{m,\nu_{\rej},\nu_{\acc}}$};
\node (aliceUpS) at (-\t+\s,\v) {};
\node (aliceMiddleS) at (-\t+\s,0) {};
\node (aliceDownS) at (-\t+\s,-\v) {};
\node (bobUpS) at (\t+\s,\v) {};
\node (bobMiddleS) at (\t+\s,0) {};
\node (bobDownS) at (\t+\s,-\v) {};
\node[filter] (fil) at (\s,-\z) {};
\node[xshift=-1.5,right] at (fil.east) {\footnotesize $\natural_E$};
\node (eveLeftS) at (-\x+\s,-\z) {};
\node (eveCenterS) at (\s,-\z) {};
\node (eveRightS) at (\x+\s,-\z) {};

\draw[sLine] (aliceDownS) to node[auto,pos=.07,yshift=-1] {$\rho$} node[auto,pos=.93,yshift=-1] {$\rho$} (bobDownS);
\draw[sLeak] (eveLeftS |- aliceDownS) to node[auto,swap,pos=.7] {$m$}
(eveLeftS |- fil.north);

\node[draw] (keyS) at (\x+\s,0) {key};
\draw[sArrow] (keyS) to node[auto,pos=.2] {$k$} (bobMiddleS);
\draw[sArrow] (keyS) to node[auto,swap,pos=.1] {$k$} (aliceMiddleS);

\draw[double] (eveRightS |- fil.north) to node[pos=.22,xshift=-1,auto,swap] {$0$} (keyS);
\draw[double] (fil.north) to node[pos=.18,xshift=-1,auto,swap] {$0$} (\s,0);
\node[circle,fill,inner sep=1.5] at (eveRightS |- bobDownS) {};

\draw[sArrow] (aliceUpS) to node[auto,yshift=-1] {\footnotesize req.} (channelS.west |- aliceUpS);
\draw[sArrow] (bobUpS) to node[auto,swap,yshift=-1] {\footnotesize req.} (channelS.east |- bobUpS);
\end{tikzpicture}

\end{centering}
\caption[Secure channel that generates key]{\label{fig:goal.alt} The
  resource $\aQ^{m,\nu_{\rej},\nu_{\acc}}$ is a restriction of the
  resource $\aS^m \| \aK^{\nu_{\rej},\nu_{\acc}}$ in which Eve's
  interface only allows $1$ bit to be input to decide both the length
  of the key and whether the message is received by Bob.}
\end{figure}
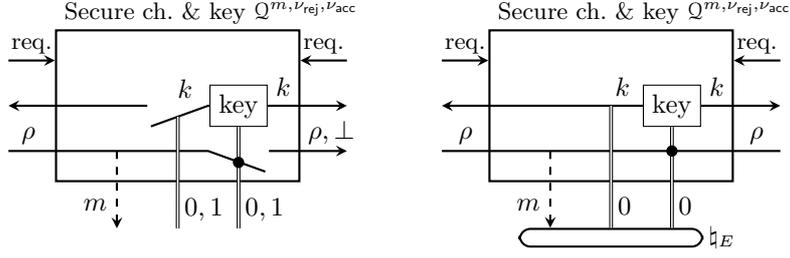

The parameteres of the construction are determined as follows.  Let
$\{U_k\}_{k \in \cK}$ be a strong purity testing code of size
$\log |\cK| = \nu$ and with error $\eps$ that encodes an $m$ qubit
message in an $m+n$ qubit cipher. And let
$\pi^{\qauth}_{AB} = (\pi^{\qauth}_A,\pi^{\qauth}_A)$ denote Alice and
Bob's converters when running the encode\-/then\-/encrypt protocol
from \figref{fig:protocol.2}. We are now ready to state the main
theorem, namely that $\pi^{\qauth}_{AB}$ is a secure authentication
scheme with key recycling.

\begin{thm}
\label{thm:auth}
Let $\pi^{\qauth}_{AB}$ denote converters corresponding to the
protocol from \figref{fig:protocol.2}. Then
$\pi^{\qauth}_{AB}$ constructs the secure channel and secret key
filtered resource
$\aQ^{m,\nu,\nu+2m+2n}_\natural$, given an
insecure quantum channel $\aC_{\square}$, a backwards authentic channel
$\aA_{\lozenge}$ and a secret key $\aK^{\nu+2m+2n}$, i.e.,
\[\aC_\square \| \aA_\lozenge \| \aK^{\nu+2m+2n}
\xrightarrow{\pi^\qauth_{AB},\eps^{\qauth}} \aQ^{m,\nu,\nu+2m+2n}_\natural\,,\]
with $\eps^{\qauth} = \sqrt{\eps}+\eps/2$, where $\eps$ is the error
of the strong purity testing code.
\end{thm}

In order to prove this theorem, we need to find a simulator such that
the real and ideal systems are indistinguishable except with advantage
$\sqrt{\eps}+\eps/2$. The simulator that we use is illustrated in
\figref{fig:auth.ideal}, and works as follows. When it receives a
notification from the ideal resource that a message is sent, it
generates EPR pairs $\ket{\Phi}^{CR}$ and outputs half of each pair
(the $C$ register) at its outer interface. Once it receives a modified
cipher (denoted $C'$ in the picture), it measures this state and the
half of the EPR pairs it kept in the Bell basis to decide if they were
modified. It accordingly activates the switch on the resource
$\aQ^{m,\nu,\nu+2m+2n}$ controlling whether Bob gets the message and
the length of the key generated, and outputs the bit of backward
communication from Bob to Alice \--- which is always leaked to Eve. If
it first receives the register $C'$ before generating the EPR pairs,
it always notifies the ideal resource to output an error and outputs
$0$ as the leak on the backwards authentic channel.

\begin{figure}[tb]
\begin{centering}

\begin{tikzpicture}[
resource/.style={draw,thick,minimum width=3.2cm,minimum height=2cm},
sArrow/.style={->,>=stealth,thick},
sLine/.style={-,thick},
sLeak/.style={->,>=stealth,thick,dashed},
simulator/.style={draw,thick,minimum width=3.2cm,minimum height=1cm,rounded corners}]

\small

\def\t{2.35} 
\def\v{.6}
\def\w{.4}
\def\x{.4}
\def\z{3.25} 
\def\s{2} 

\node[resource] (channel) at (0,0) {};
\node[yshift=-1.5,above right] at (channel.north west) {\footnotesize
  Secure ch.\ \& key $\aQ^{m,\nu,\nu+2m+2n}$};
\node[simulator] (sim) at (0,-\s) {};
\node[xshift=-1.5,below right] at (sim.north east) {\footnotesize $\sigma^\qauth_E$};
\node (aliceUp) at (-\t,\v) {};
\node (aliceMiddle) at (-\t,0) {};
\node (aliceDown) at (-\t,-\v) {};
\node (bobUp) at (\t,\v) {};
\node (bobMiddle) at (\t,0) {};
\node (bobDown) at (\t,-\v) {};
\node (eveLeftL) at (-3*\x,-\z) {};
\node (eveLeftR) at (-\x,-\z) {};
\node (eveRightL) at (\x,-\z) {};
\node (eveRightR) at (3*\x,-\z) {};

\node at (-2.5*\x,-\s) {$\ket{\Phi}^{CR}$};

\draw[sLine] (aliceDown) to node[auto,pos=.1,yshift=-1] {$\rho$} (\x-\w,-\v) to +(340:2*\w);
\draw[sArrow]  (\x+\w,-\v) to node[auto,pos=.8,yshift=-1] {$\rho,\bot$} (bobDown);
\draw[sLeak] (eveLeftL |- aliceDown) to node[auto,swap,pos=.7] {$m$}
(eveLeftL |-sim.north);
\draw[sArrow] (eveLeftL |- sim.south) to node[auto,pos=.5,xshift=-1] {$C$} (eveLeftL);

\node[draw] (key) at (\x,0) {key};
\draw[sArrow] (key) to node[auto,pos=.2] {$k$} (bobMiddle);
\draw[sLine] (key) to (-\x+\w,0) to node[auto,swap,pos=.1] {$k$}
node[pos=.53] (handle) {} +(200:2*\w);
\draw[sArrow] (-\x-\w,0) to (aliceMiddle);

\draw[double] (eveRightL |- sim.north) to node[pos=.2,xshift=-1,auto,swap] {$0,1$}
node[circle,fill,pos=.61,inner sep=1.5] {}  (key);
\draw[double] (eveLeftR) to node[pos=.08,xshift=-1,auto,swap]
{$0,1$} (handle.center);

\draw[sArrow] (eveRightR |- sim.south) to  node[auto,pos=.6] {$0,1$} (eveRightR);
\draw[sArrow] (eveRightL) to node[auto,pos=.5,xshift=-1,swap] {$C'$} (eveRightL |- sim.south);

\draw[sArrow] (aliceUp) to node[auto,yshift=-1] {\footnotesize req.} (channel.west |- aliceUp);
\draw[sArrow] (bobUp) to node[auto,swap,yshift=-1] {\footnotesize req.} (channel.east |- bobUp);

\end{tikzpicture}

\end{centering}
\caption[Ideal system for quantum
authentication]{\label{fig:auth.ideal}The ideal quantum authentication
  system consisting of the constructed resource $\aS^m$ and
  $\bar{\aK}^{\nu,\nu+2m+2n}$, and the simulator $\sigma^\qauth_E$.}
\end{figure}
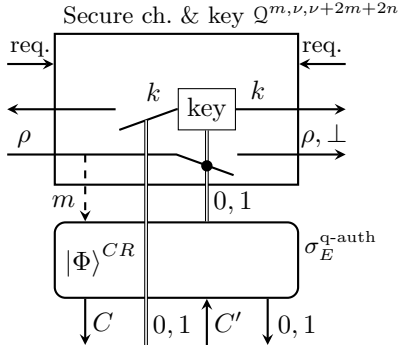

\begin{proof}
  It is trivial to show that correctness holds with error $0$,
  namely that
  \begin{equation} \label{eq:thm.cor} d\left(\pi^{\qauth}_{AB} \left(
        \aC \square_E \| \aA \lozenge_E \| \aK^{\nu+2m+2n}
      \right),\aQ^{m,\nu,\nu+2m+2n} \natural_E \right) = 0
    \,.\end{equation} We now prove the case of security, i.e.,
  \begin{equation} \label{eq:thm.sec} d\left(\pi^{\qauth}_{AB} \left(
        \aC \| \aA \| \aK^{\nu+2m+2n} \right), \aQ^{m,\nu,\nu+2m+2n}
      \sigma^\qauth_E \right) \leq \sqrt{\eps}+\eps/2\,.\end{equation}

      The real and ideal systems, drawn in Figures~\ref{fig:auth.real}
      and \ref{fig:auth.ideal} have $5$ inputs. The distinguisher thus
      has the choice between $5!$ possible orders for providing
      inputs. However, most of these orders are redundant and do not
      need to be analyzed. Providing the requests for the secret keys
      before they are ready is pointless. So it is sufficient to look
      at the case where these requests are made as soon as the keys
      are available for recycling, i.e., after Bob has received the
      message from Alice and after Alice has received the confirmation
      from Bob. What is more, neither sending Alice an error on the
      backwards authentic channel nor allowing her to get Bob's
      confirmation will help either way, since the distinguisher
      already knows what output Alice will produce, so we can
      completely ignore this input. That leaves only $2$ in-ports, and
      thus $2$ orders to analyze:
  \begin{enumerate}
    \item The distinguisher first inputs a message at Alice's
  interface, gets the cipher at Eve's interface, inputs a possibly
  modified cipher at Eve's interface, gets the output at Bob's
  interface, and requests the recycled key.
  \item The distinguisher first inputs a fake cipher at Eve's
    interface, gets the output at Bob's interface, makes a request for
    his recycled key, then inputs a message at Alice's interface and
    receives the cipher for that message.
  \end{enumerate}

  We start with the first case, the initial message is sent to
  Alice. The distinguisher prepares a message $\ket{\psi}^{ME}$ and
  inputs the $M$ part at Alice's interface. The ideal channel then
  notifies the simulator that a message has been input. The simulator
  prepares a maximally entangled state $\ket{\Phi}^{CR}$ of dimension
  $2^{2m+2n}$ and outputs the $C$ register at Eve's interface. The
  distinguisher now holds a bipartite state in $CE$, to which it
  applies a unitary $U^{CE}$. Without loss of generality, one may
  write the unitary as $U^{CE} = \sum_j P^C_j \tensor E^E_j$, where
  $P^C_j$ are Paulis acting on the cipher register $C$ and $E^E_j$ act
  on the distinguisher's internal memory $E$. The resulting state in
  the $C$ register is input back in the $E$ interface. The simulator
  now measures $CR$ in the Bell basis defined by the projectors
  $\{P_j \tensor I \proj{\Phi}^{CR} P_j \tensor I\}_j$. If the outcome
  is $j = 0$ \--- where $P_0 = I$ \--- it tells the resource
  $\aQ^{m,\nu,\nu+2m+2n}$ that the cipher was not modified. In which
  case the contents of the register $M$ is output at Bob's interface
  with an $\acc$ flag. Furthermore, $\aQ^{m,\nu,\nu+2m+2n}$ generates
  a fresh uniform key $(k,\ell)$, where $|k| = \nu$ and
  $|\ell| = 2m+2n$. If the outcome is $j \neq 0$, then the simulator
  notifies the resource to delete the message and output a $\rej$ flag,
  as well as prepare only the shorter key $k$. The
  distinguisher then sends a request to obtain the fresh key. So the
  final state held by the distinguisher interacting with the ideal
  system is
\begin{multline} \label{eq:auth.ideal} \zeta = \proj{\acc} \tensor
  \tau^{K} \tensor \tau^L \tensor \left[ \left( I^M \tensor E^E_0 \right)
    \proj{\psi}^{ME} \left( I^M \tensor \hconj{\left(E^E_0\right)} \right) \right] \\ {}+
  \sum_{j \neq 0} \proj{\rej} \tensor \tau^{K} \tensor E^E_j \rho^E
  \hconj{\left(E^E_j\right)}\,,\end{multline} where $\tau^K$ and $\tau^L$ are fully
mixed states and $\rho^E = \trace[M]{\proj{\psi}^{ME}}$. One could
append states $\bot^L$ and $\bot^M$ in the $\rej$ branch of
\eqnref{eq:auth.ideal} so that both terms have the same number of
registers; we omit them for simplicity.

In the real system, for the secret key $(k,\ell)$, the state before
Bob's measurement of the syndrome is given by
\begin{align*}\ket{\varphi_{k,\ell}}^{SME} & = \sum_j \left( \hconj{\left(U_k^{SM}\right)} P^{SM}_\ell P^{SM}_j P^{SM}_\ell U^{SM}_k \tensor
    E^E_j \right) \zero^S \ket{\psi}^{ME} \\ & = \sum_j (-1)^{\sym{j}{\ell}}\left( \hconj{\left(U_k^{SM}\right)}
    P^{SM}_j U^{SM}_k \tensor E^E_j \right) \zero^S \ket{\psi}^{ME}\,,
\end{align*}
where $\sym{\cdot}{\cdot}$ denotes the symplectic product defined in
\eqnref{eq:symplectic}. Let $\cJ^k_s$ be the set of indices $j$ such
that the error $P^{SM}_j$ produces a syndrome $s$ when code $k$ is
used, i.e.,
$\hconj{\left(U^{SM}_k \right)} P^{SM}_jU^{SM}_k = e^{i \theta_{k,j}}
P^{S}_{s,z} \tensor P^M_{j'}$ for some $\theta_{k,j}$
(see \eqnref{eq:stabilizer.pauli} and discussion thereafter). For
$j \in \cJ^k_s$, let
\begin{align*}\ket{s}^S\ket{\psi_{j,k}}^{ME} & \coloneqq \left( \hconj{\left(U^{SM}_k\right)} P^{SM}_j U^{SM}_k \tensor
  E^E_j \right) \zero^S \ket{\psi}^{ME} \\
  & \; = e^{i \theta_{k,j}} \left( P^{S}_{s,z} \tensor P^M_{j'} \tensor
  E^E_j \right) \zero^S \ket{\psi}^{ME} \,.\end{align*} Then
\begin{align*}
  \ket{\varphi_{k,\ell}} & = \sum_s \sum_{j \in \cJ^k_s} (-1)^{\sym{j}{\ell}}\left( \hconj{\left(U_k^{SM}\right)}
    P^{SM}_j U^{SM}_k \tensor E^E_j \right) \zero^S \ket{\psi}^{ME} \\
& = \sum_s \sum_{j \in \cJ^k_s} (-1)^{\sym{j}{\ell}} \ket{s}^S \ket{\psi_{j,k}}^{ME}\,.
\end{align*}
The next step in Bob's protocol consists in measuring the syndrome. If
$s=0$ is obtained, he outputs the message as well as the key
$(k,\ell)$ and a flag $\acc$. Otherwise he deletes the message,
outputs $k$ with the flag $\rej$. The final state held be the
distinguisher in this case is
\begin{align*}
\xi = & \proj{\acc} \tensor \frac{1}{2^{\nu+2m+2n}}
        \sum_{k,\ell}\proj{k,\ell} \\ &  \qquad \qquad \qquad \qquad \qquad \qquad \quad \tensor
\sum_{j_1,j_2 \in \cJ^k_0} (-1)^{\sym{j_1 \xor
    j_2}{\ell}}\ketbra{\psi_{j_1,k}}{\psi_{j_2,k}}^{ME} \\
& {}+ \proj{\rej} \tensor \frac{1}{2^{\nu+2m+2n}} \sum_{k,\ell}\proj{k}
 \\ &  \qquad \qquad \qquad \qquad \qquad \qquad \tensor \sum_{s \neq 0}
\sum_{j_1,j_2 \in \cJ^k_s} (-1)^{\sym{j_1 \xor
    j_2}{\ell}} E^E_{j_1} \rho^E \hconj{\left(E^E_{j_2}\right)}\,,
\end{align*}
where we have used $\ket{\psi_{j,k}}^{ME} = \left(V^M_{k,j} \tensor
  E^E_j\right) \ket{\psi}^{ME}$ for some unitary $V^M_{k,j}$.

Setting
\begin{align*}
\zeta^\acc & \coloneqq \left(I^M \tensor E^E_0\right) \proj{\psi}^{ME}
\left(I^M \tensor \hconj{\left(E_0^E\right)} \right)\,, \\
\zeta^\rej & \coloneqq \sum_{j \neq 0} E^E_j \rho^E \hconj{\left(E_j^E\right)}\,, \\
\xi^\acc_{k,\ell} & \coloneqq \sum_{j_1,j_2 \in \cJ^k_0} (-1)^{\sym{j_1 \xor
    j_2}{\ell}}\ketbra{\psi_{j_1,k}}{\psi_{j_2,k}}^{ME}\,, \\
\xi^\rej_{k} & \coloneqq \frac{1}{2^{2m+2n}} \sum_{\ell,s \neq 0} \sum_{j_1,j_2 \in \cJ^k_s} (-1)^{\sym{j_1 \xor
    j_2}{\ell}} E^E_{j_1} \rho^E \hconj{\left(E_{j_2}^E\right)}\,,
\end{align*} 
the distance between real and ideal systems may be written as
\[ \frac{1}{2} \trnorm{\zeta-\xi} =  \frac{1}{2\cdot 2^{\nu+2m+2n}}
\sum_{k,\ell}\trnorm{\zeta^\acc-\xi^\acc_{k,\ell}} + \frac{1}{2 \cdot 2^{\nu}}
\sum_{k}\trnorm{\zeta^\rej-\xi^\rej_{k}}\,.\]

$\zeta^\acc$ and $\xi^\acc_{k,\ell}$ are both pure states, so from
 \lemref{lem:trace2norm} we bound their distance as
\begin{align*}
 \frac{1}{2} \trnorm{\zeta^\acc-\xi^\acc_{k,\ell}} & \leq
\norm{\left(I^M \tensor E^E_0\right) \ket{\psi}^{ME} - \sum_{j \in \cJ^k_0}
  (-1)^{\sym{j}{\ell}} \ket{\psi_{j,k}}^{ME}} \\ & = \norm{\sum_{j \in \cJ^k_0
    \setminus \{0\}} (-1)^{\sym{j}{\ell}} \ket{\psi_{j,k}}^{ME}} \\
& = \sqrt{\sum_{j_1,j_2 \in \cJ^k_0
    \setminus \{0\}} (-1)^{\sym{j_1 \xor j_2}{\ell}} \braket{\psi_{j_1,k}}{\psi_{j_2,k}}}\,,
\end{align*}
where $\norm{\ket{a}} = \sqrt{\braket{a}{a}}$ is the vector $2$\=/norm
and we used the fact that
$\ket{\psi_{0,k}}^{ME} = \left(I^M \tensor E^E_0\right)
\ket{\psi}^{ME}$.
From Jensen's inequality and using \eqnref{eq:symplectic.trick} we
obtain
\begin{multline*}
 \frac{1}{2 \cdot 2^{\nu+2m+2n}} \sum_{k,\ell}
 \trnorm{\zeta^\acc-\xi^\acc_{k,\ell}} \\ \begin{split} & \leq
\sqrt{ \frac{1}{2^{\nu+2m+2n}} \sum_{k,\ell} \sum_{j_1,j_2 \in \cJ^k_0
    \setminus \{0\}} (-1)^{\sym{j_1 \xor j_2}{\ell}} \braket{\psi_{j_1,k}}{\psi_{j_2,k}}} \\
& = \sqrt{ \frac{1}{2^{\nu}} \sum_{k} \sum_{j \in \cJ^k_0
    \setminus \{0\}} \braket{\psi_{j,k}}{\psi_{j,k}}}\,. \end{split}
\end{multline*}
Finally, because the code is a strong purity testing code with error
$\eps$ and that
$\braket{\psi_{j,k}}{\psi_{j,k}} = \trace{E^E_j \rho^E
  \hconj{\left(E^E_j\right)}} \eqqcolon p_j$
with $\sum_j p_j = 1$, we get
\begin{align*}
    \frac{1}{2|\cK||\cL|} \sum_{k,\ell}
  \trnorm{\zeta^\acc-\xi^\acc_{k,\ell}} 
& \leq \sqrt{ \frac{1}{|\cK|} \sum_{j \neq 0} \sum_{k : j \in \cJ^k_0}
  \braket{\psi_{j,k}}{\psi_{j,k}}} \\
& = \sqrt{ \frac{1}{|\cK|} \sum_{j \neq 0} \sum_{k : j
    \in \cJ^k_0} p_j} \\
 & \leq \sqrt{\sum_{j \neq 0} \eps p_j } \leq \sqrt{\eps}\,.
\end{align*}

In the reject branch of the real system we have
\begin{align*}
\xi^\rej_{k} & = \frac{1}{2^{2m+2n}} \sum_{\ell,s \neq 0} \sum_{j_1,j_2 \in \cJ^k_s} (-1)^{\sym{j_1 \xor
    j_2}{\ell}} E^E_{j_1} \rho^E \hconj{\left(E^E _{j_2}\right)} \\
& = \sum_{s \neq 0} \sum_{j \in \cJ^k_s} E^E_{j} \rho^E \hconj{\left(E^E _{j}\right)} \\
& = \sum_{j \notin \cJ^k_0} E^E_{j} \rho^E \hconj{\left(E^E _{j}\right)} \,,
\end{align*}
where we used again \eqnref{eq:symplectic.trick}. Thus
\begin{align*}
 \frac{1}{2 \cdot 2^{\nu}} \sum_{k}\trnorm{\zeta^\rej-\xi^\rej_{k}}
& = \frac{1}{2 \cdot 2^{\nu}} \sum_k \trnorm{\sum_{j \in \cJ^k_0 \setminus
  \{0\}} E^E_{j} \rho^E \hconj{\left(E^E _{j}\right)}} \\
& \leq \frac{1}{2 \cdot 2^{\nu}} \sum_k \sum_{j \in \cJ^k_0 \setminus
  \{0\}} p_j \leq \eps/2\,.
\end{align*}
Putting all this together we get 
\[ \frac{1}{2} \trnorm{\zeta-\xi} \leq \sqrt{\eps} + \eps/2\,.\]

We now consider the second case: the distinguisher first prepares a
state $\ket{\psi}^{CE}$ and inputs the $C$ part at Eve's interface,
then obtains the output at Bob's interface. Note that in the ideal
case the channel always outputs a $\rej$ message at Bob's
interface. Thus, if the cipher is accepted by Bob \--- who outputs a
state $\zeta^\acc$ \--- the distinguisher must be interacting with the
real system and can already output this guess. In the case of a
rejection, it now holds a bipartite system $KE$ \--- the recycled key
$K$ and its purifying system $E$. It then applies an isometry
$U : \hilbert_{KE} \to \hilbert_{KME}$ to this system and inputs the
$M$ part of the resulting state at Alice's interface. After which it
obtains a cipher at Eve's interface and holds the tripartite system
$KCE$ \--- the recycled key $K$, the cipher $C$ and 
its internal memory $E$. We denote this state $\zeta$ in the ideal
case and $\xi^\rej$ in the real case, and we need to bound
\[ \frac{1}{2} \trnorm{\zeta - \xi^\rej} + \frac{1}{2}
\trnorm{\xi^\acc}\,.\]

In a first step, we assume that the state $\ket{\psi}^{CE}$ prepared
by the distinguisher is an antisymmetric fully entangled state, which
we denote
$\ket{\Psi^-}^{CE} = \sum_x (-1)^{w(x)}\ket{x,\bar{x}}^{CE}$, where
$w(x)$ is the Hamming weight of $x \in \{0,1\}^{m+n}$ and $\bar{x}$ is
the string $x$ with all bits flipped. In the ideal case the simulator
notifies the channel to reject the cipher, and the state
$\proj{\rej} \tensor \tau^K$ is output at Bob's interface. The
distinguisher then holds $\zeta = \tau^K \tensor \tau^E$. In the real
case, Bob applies the decoding algorithm, i.e., first a Pauli
$P^C_\ell$, then a unitary $\hconj{\left(U^C_k\right)}$ and finally
measures $n$ bits of the syndrome in the computational basis. Since
the antisymmetric state is invariant under $U \tensor U$, one could
equivalently apply the inverse operation, $P_\ell U_k$, to the $E$
system, i.e., the state after Bob's measurement is given by
\begin{multline*} \frac{1}{2^{\nu+3m+3n}} \sum_{k,\ell,s,x_1,x_2}
  (-1)^{w(x_1) \xor w(x_2)}  \proj{k,\ell} \\ {}\tensor
 \left( I^C \tensor P^E_\ell U^E_k \right)
\ketbra{s,x_1,\bar{s},\bar{x}_1}{s,x_2,\bar{s},\bar{x}_1}^{CE} \left( I^C
  \tensor \hconj{\left(U _k^E\right)} P^E_\ell \right)\,. \end{multline*}
If $s=0$ Bob accepts the cipher as being valid, which happens with probability
$2^{-n}$, i.e., $\trnorm{\xi^\acc} = 2^{-n}$. In the case where $s
\neq 0$, he deletes the cipher, so the remaining state is given by
\begin{multline*} \frac{1}{2^{\nu+3m+3n}} \sum_{k,\ell,s \neq 0,x}
  \proj{k,\ell} \tensor \left( I^C \tensor P^E_\ell U^E_k \right)
  \proj{\bar{s},\bar{x}}^{CE} \left( I^C \tensor \hconj{\left(U_k^E\right)} P^E_\ell \right)
  \\ = \tau^K \tensor \tau^L \tensor \tau^E -
  \rho^{KLE}\,, \end{multline*} where
\[ \rho^{KLE} = \frac{1}{2^{\nu+3m+3n}} \sum_{k,\ell,x}
\proj{k,\ell} \tensor P^E_\ell U^E_k \proj{\bar{0},\bar{x}}^{E} \hconj{\left(U^E_k\right)}
P^E_\ell \,,\]
$K$ is made public and the $L$ system is the part of the key kept
secret by the players.

Let $\cE$ denote the completely positive, trace\-/preserving
(CPTP) map consisting of the distinguisher's next
step \--- the isometry $U : \hilbert_{KE} \to \hilbert_{KME}$ \--- and
the final operation of the ideal system \--- deleting the message
system $M$ that is input at Alice's interface and outputting a fully
mixed state $\tau^C$. Let $\cF$ denote the CPTP map consisting of the
distinguisher's next step and the final operation of the real system
\--- encoding the message system $M$ according to the protocol and
outputting the resulting cipher. We have
$\zeta = \cE \left( \tau^K \tensor \tau^E \right)$ and
$\xi^\rej = \cF \left( \tau^K \tensor \tau^L \tensor \tau^E \right) -
\cF \left(\rho^{KLE}\right)$. Thus,
\[ \frac{1}{2} \trnorm{\zeta - \xi^\rej} \leq \frac{1}{2} \trnorm{\cE
  \left( \tau^K \tensor \tau^E \right) - \cF \left( \tau^K \tensor
    \tau^L \tensor \tau^E \right)} + \frac{1}{2}2^{-n}\,,\]
since $\trnorm{\rho^{KLE}} = 2^{-n}$. Finally, note that we have \[\cE
  \left( \tau^K \tensor \tau^E \right) = \cF \left( \tau^K \tensor
    \tau^L \tensor \tau^E \right) = \tau^C \tensor \sigma^{KE}\] for
  $\sigma^{KE} = \ktrace[M]{U \left( \tau^K \tensor \tau^E\right)
    \hconj{U}}$, since the random Pauli $P_\ell$ applied by the
  encryption algorithm completely decouples the cipher from
  $KE$. Putting this together, we get
  \[ \frac{1}{2} \trnorm{\zeta - \xi} \leq 2^{-n} \leq \sqrt{\eps} \,,\]
  since a strong purity testing code will always have an error
  $\eps \geq \frac{2^{2m+n}-1}{2^{2m+2n}-1} \geq 2^{-2n}$.

The final case that remains to consider is when the distinguisher prepares a state
$\ket{\psi}^{CE}$ that is not the antisymmetric state. We will reduce
this case to that of the entangled antisymmetric by using the
entangled state $\ket{\Psi^-}^{CE}$ to teleport the $C'$ part of any
state $\ket{\psi}^{C'E'}$. Let the teleportation scheme be given by the
projectors $\{M^{EC'}_a\}$ on $EC'$ which incur a Pauli correction $P^C_a$
on the $C$ system, i.e.,
\[ \trace[EC']{\sum_a P^C_a \tensor M^{EC'}_a \left(
    \proj{\Psi^-}^{CE} \tensor
  \proj{\psi}^{C'E'}\right) P^C_a \tensor M^{EC'}_a} = \proj{\psi}_{CE'} \,.\]
So the distinguisher prepares an entangled state $\ket{\Psi^-}_{CE}$ and the state
it wishes to send to Bob, $\ket{\psi}^{C'E'}$. It teleports the $C'$
register to the $C$ register, and sends this to Bob, who performs his
decryption operation $\hconj{\left(U_k^C\right)} P^C_\ell $. This results in
the shared state
\begin{align*} & \frac{1}{2^{\nu+2m+2n}} \sum_{k,\ell,a} \proj{k,\ell}
                 \tensor \Big[ \left(\hconj{\left(U_k^C\right)}
                 P^C_\ell P^C_a \tensor M^{EC'}_a\right)  \\
  & \qquad \qquad \qquad \qquad
   \left( \proj{\Psi^-}^{CE} \tensor \proj{\psi}^{C'E'}\right)
    \left(P^C_a P^C_\ell U^C_k \tensor M^{EC'}_a\right) \Big] \\
& \qquad = \frac{1}{2^{\nu+2m+2n}} \sum_{k,\ell,a} \proj{k,\ell \xor
  a} \tensor \Big[ \left(\hconj{\left(U_k^C\right)} P^C_\ell \tensor M^{EC'}_a\right)\\
  & \qquad \qquad \qquad \qquad
   \left( \proj{\Psi^-}^{CE} \tensor \proj{\psi}^{C'E'}\right)
    \left(P^C_\ell U^C_k \tensor M^{EC'}_a\right) \Big] \\
& \qquad = \frac{1}{2^{\nu+2m+2n}} \sum_{k,\ell,a} X^L_a \proj{k,\ell}
  X^L_a \tensor \Big[ \left(I^C \tensor M^{EC'}_a P^E_\ell U^E_k \right)\\
  & \qquad \qquad \qquad \qquad
\left( \proj{\Psi^-}^{CE} \tensor \proj{\psi}^{C'E'}\right)
    \left(I^C \tensor \hconj{\left(U_k^E\right)} P^E_\ell
    M^{EC'}_a\right) \Big]\,,\end{align*} where
  $X^L_a$ flips the bits of $\ell$ in the positions where $a_i = 1$. The
  teleportation of $\ket{\psi}_{C'E'}$ is thus equivalent to a
  measurement of the distinguisher's system followed by a correction
  of the secret key $\ell$. This may however be performed after Bob
  measures the syndrome and accepts or rejects the cipher he
  received. The probability of accepting the cipher is thus unchanged,
  and plugging in the result from the case where the distinguisher
  sends Bob half of the anti\-/symmetric entangled state, we find that
  the state of the real system in the rejection branch after Bob's
  operations is
  \begin{align*} & \sum_a \left( X^L_a \tensor M^{EC'}_a \right) \tau^K \tensor \tau^L \tensor
  \tau^E  \tensor \proj{\psi}^{C'E'} \left( X^L_a \tensor M^{EC'}_a \right) \\
    & \qquad \qquad \qquad \qquad {}- \sum_a
   \left( X^L_a \tensor M^{EC'}_a \right) \rho^{KLE} \tensor
   \proj{\psi}^{C'E'}  \left( X^L_a \tensor M^{EC'}_a \right) \\
    & \qquad = \tau^K \tensor \tau^L \tensor \sigma^{EC'E'} -
      \tilde{\rho}^{KLEC'E'}\,, \end{align*} where $\sigma^{EC'E'} = \sum_a
    M^{EC'}_a \left( \tau^E  \tensor \proj{\psi}^{C'E'}\right)  M^{EC'}_a$. And in the
    ideal system the state after the ideal channel outputs a rejection
    is $\tau^K \tensor \sigma^{EC'E'}$. We thus obtain the same bound
    on the distance between real and ideal systems as in the previous
    case.
\end{proof}

\begin{cor}
\label{cor:auth}
Let $\pi^{\qauth}_{AB}$ denote converteres corresponding to the
protocol from \figref{fig:protocol.2}. Then $\pi^{\qauth}_{AB}$
constructs the secure channel and secret key filtered resource
$(\aS^m \| \bar{\aK}^{\nu,\nu+2m+2n})_\flat$, given an insecure
quantum channel $\aC_{\square}$, a backwards authentic channel
$\aA_{\lozenge}$ and a secret key $\aK^{\nu+2m+2n}$, i.e.,
\[\aC_\square \| \aA_\lozenge \| \aK^{\nu+2m+2n}
\xrightarrow{\pi^\qauth_{AB},\eps^{\qauth}} (\aS^m \|
\bar{\aK}^{\nu,\nu+2m+2n})_\flat\,,\]
with $\eps^{\qauth} = \sqrt{\eps}+\eps/2$, where $\eps$ is the error
of the strong purity testing code.
\end{cor}

\begin{proof}
  $\aQ^{m,\nu_\rej,\nu_\acc}_\natural$ is a stronger resource than
  $(\aS^m \| \bar{\aK}^{\nu_\rej,\nu_\acc})_\flat$, and one trivially
  has
  \[ \aQ^{m,\nu_\rej,\nu_\acc}_\natural \xrightarrow{\id,0} (\aS^m \|
  \bar{\aK}^{\nu_\rej,\nu_\acc})_\flat\,,\]
  with a simulator that forwards everything between the distinguisher
  and Eve's interface of $\aS^m \| \bar{\aK}^{\nu_\rej,\nu_\acc}$
  except for the bit deciding the secret key length and whether the
  message is accepted, which is copied and sent to both $\aS^m$ and
  $\bar{\aK}^{\nu_\rej,\nu_\acc}$. The corollary follows immediately
  from this and the composition theorem.
\end{proof}

\subsection{Optimality of the recycled key length}
\label{sec:auth.optimality}

It follows from \lemref{lem:ete} that \thmref{thm:auth} and
\corref{cor:auth} are also proofs of security for the
encrypt\-/then\-/encode protocol from \figref{fig:protocol.1}, i.e.,
\[\aC_\square \| \aA_\lozenge \| \aK^{\nu+2m+n}
\xrightarrow{\bar{\pi}^\qauth_{AB},\eps^{\qauth}}
\aQ^{m,\nu,\nu+2m+n}_\natural \xrightarrow{\id,0} (\aS^m \|
\bar{\aK}^{\nu,\nu+2m+n})_\flat\,,\]
with $\eps^{\qauth} = \sqrt{\eps}+\eps/2$. Thus, in the case where the
message is not accepted by Bob, $2m+n$ bits of key are lost. We prove
here that this is optimal: one cannot recycle any extra bit of key.

\begin{lem}
\label{lem:auth.optimal}
There exists an adversarial strategy to obtain all the secret bits
that are not recycled in the encrypt\-/then\-/encode protocol.
\end{lem}

\begin{proof} The distinguisher prepares EPR pairs $\ket{\Phi}^{ME}$
  and provides the $M$ part to Alice. It then receives the cipher and
  thus holds the state
  \[ U^{SM}_k P^M_{\ell} \left( \ket{s}^S \tensor \ket{\Phi}^{ME}
  \right)\,,\]
  which it keeps. It then sends a bogus cipher to Bob, and obtains the
  key $k$ after Bob recycles it. It applies the decoding unitary
  $\hconj{\left( U^{SM}_k\right)}$, measures the $S$ register to get
  the secret key $s$ and measures the joint $ME$ register in the Bell
  basis to get the secret key $\ell$. \end{proof}

\subsection{Explicit constructions}
\label{sec:auth.explicit}

The protocols we have given in \secref{sec:auth.protocol} use strong purity
testing codes, and the parameters of the key used, key recycled and
error depend on the parameters of these codes. In this section we give
two constructions of purity testing codes. The first requires less
initial secret key, the second has a better error parameter. Both have
the same net consumption of secret key bits.

The first construction is from Barnum et al.~\cite{BCGST02}. They give
an explicit strong purity testing code with $\nu = n$ and
$\eps = \frac{2m/n+2}{2^n}$.\footnote{In fact, \cite{BCGST02} only
  prove that their construction is a purity testing code, not a strong
  one. But one can easily verify that it is strong with the same
  parameters. What is more, their construction has $\nu = \log(2^n+1)$
  and $\eps = \frac{2m/n+2}{2^n+1}$. We remove one of the keys (and
  thus increase the error), so as to get simpler final expressions.}
Plugging this in the parameters from \thmref{thm:auth} with the
encrypt\-/then\-/encode protocol, we get the following.

\begin{cor}
\label{cor:auth.1}
The encrypt\-/then\-/encode protocol with the purity testing code of
\cite{BCGST02} requires an initial key of length $2m+2n$. It recycles
all bits if the message is accepted, and $n$ bits if the message is
rejected. The error is \[\eps^{\qauth} =\sqrt{\frac{2m/n+2}{2^{n}}} +
\frac{m/n+1}{2^n} \,.\]
\end{cor}

The second construction we give is based on an explicit purity testing
code by Chau~\cite{Cha05} \--- though he does not name it this way.
Chau~\cite{Cha05} finds a set of unitaries $\cU = \{U_k\}$ in dimension $d$ such
that, if $k$ is chosen uniformly at random, any non\-/identity Pauli
is mapped to every non\-/identity Pauli with equal frequency, i.e.,
$\forall P_j,P_\ell$ with $P_j \neq I$ and $P_\ell \neq I$,
\[ \left| \left\{ U_k \in \cU : U_k P_j \hconj{U}_k = e^{i
      \theta_{j,k,\ell}} P_\ell \right\} \right| = \frac{|\cU|}{d^2-1}
\,,\] where $e^{i \theta_{j,k,\ell}}$ is some global phase.

We prove in \appendixref{app:design}, \lemref{lem:purity}, that this
is a strong purity testing code with $\eps = 2^{-n}$ for
$d = 2^{m + n}$. It also has
$|\cU| = 2^{m+n} \left(2^{2m+2n} - 1\right)$, hence
$\nu = m+n+\log \left(2^{2m+2n} -1 \right) \leq 3m+3n$. Note that when
composed with Paulis as in the encode\-/then\-/encrypt protocol,
$\{P_\ell U_k\}_{k,\ell}$ is a unitary
$2$\-/design~\cite{Dan05,DCEL09}. It follows that any (approximate)
unitary $t$\-/design is a good quantum authentication scheme
(see \appendixref{app:design} for a formal proof).

\begin{cor}
\label{cor:auth.2}
The encrypt\-/then\-/encode protocol with the purity testing code of
\cite{Cha05} requires an initial key of length $5m+4n$. It recycles
all bits if the message is accepted, and $3m+3n$ bits if the message
is rejected. The error is $\eps^{\qauth} = 2^{-n/2} + 2^{-n-1}$.
\end{cor}

\section{Complete construction}
\label{sec:construction}

We proved in \secref{sec:auth} that the quantum authentication
protocols from Figures~\ref{fig:protocol.1} and \ref{fig:protocol.2}
construct a secure channel and secret key filtered resource
$(\aS^m \| \bar{\aK}^{\nu_\rej,\nu_\acc})_\flat$ from a shared secret
key $\aK^{\mu}$, a noiseless insecure quantum channel $\aC_\square$
and a backwards authentic classical channel $\aA_\lozenge$, namely,
\[
\aC_\square \| \aA_\lozenge \| \aK^{\mu}
\xrightarrow{\pi^\qauth_{AB},\eps^{\qauth}} (\aS^m \|
\bar{\aK}^{\nu_\rej,\nu_\acc})_\flat\,.
\]
In this section we show how to construct the required resources from
nothing but shared secret key and noisy channels, then put it all
together to get the exact bounds of the composed protocols.

We discuss in \secref{sec:construction.auth} how to obtain the
authentic channel: it can be constructed from a shared secret key and
an insecure channel using any (classical) MAC\-/type authentication
scheme~\cite{Por14}. Channels are however usually not noiseless. This
is solved by using error correction: an error correction code
constructs a noiseless channel given a noisy channel (with known
noise), and is presented in \secref{sec:construction.noiseless}. Were
we to put things together at this point, we would construct the
desired secure quantum channel from nothing but shared secret key and
insecure noisy channels. But even in the case where no adversary is
present, we would still consume some secret key, because we do not
recycle the key from the backwards authentic channel. So in
\secref{sec:construction.morekey} we construct shared secret key given
a secure channel \--- we simply share secret key using the channel.
Combining all these pieces together, we obtain our secure quantum
channel without any net consumption of key in the case where the
adversary does not tamper with the messages. The security of the
composed scheme follows immediately from the security of each
component and the composition theorem of AC~\cite{MR11}. Finally, in
\secref{sec:construction.noauth} we discuss a setting in which the
backwards authentic channel is not needed, thus allowing a more
efficient use of the quantum channel \--- since we do not need it to
distribute key, and can thus use all of its capacity to send messages.

\subsection{Authentic channel}
\label{sec:construction.auth}

We used classical authentication as an example in \secref{sec:cc}: an
authentication protocol $\pi^\cauth_{AB}$ constructs an authentic
channel $\aA_\lozenge$ given a (noiseless) insecure channel
$\aC_{\square}$ and a secret key $\aK^{\eta}$,
\begin{equation}
\label{eq:security.classicalauth}
\aC_\square \| \aK^{\eta} \xrightarrow{\pi^\cauth_{AB},\eps^{\cauth}} \aA_\lozenge\,.
\end{equation}
The corresponding real and ideal systems were illustrated in
Figures~\ref{fig:classicalauth.real} and
\ref{fig:classicalauth.ideal}, respectively. Appending the MAC
$h_k(x)$ to the message $x$ is sufficient to construct the authentic
channel with error $\eps$ if the family of functions $\{h_k\}_k$ is
$\eps$\=/almost strongly $2$\-/universal~\cite{Por14}. In our case,
Bob only needs to send a $1$-bit message to Alice. If the key $k$ has
length $\eta$, a $2^{-\eta/2}$\=/almost strongly $2$\-/universal
family of functions for $1$-bit messages is given by $h_k(x) = k_x$,
where $k_0$ are the first $\eta/2$ bits of $k$ and $k_1$ are the last
$\eta/2$ bits.

\begin{lem}
\label{lem:classicalauth}
The authentication scheme described above satisfies
\eqnref{eq:security.classicalauth} with $\eps^{\cauth} = 2^{-\eta/2}$.
\end{lem}

\begin{proof} Follows from \cite[Lemma 9]{Por14}. \end{proof}

\subsection{Noiseless channel}
\label{sec:construction.noiseless}

Both the classical authentication protocol discussed in
\secsref{sec:cc} and \ref{sec:construction.auth}, as well as the
quantum authentication protocol analyzed in \secref{sec:auth} produce
an error message as soon as there is any disturbance on the
channel. Realistic channels are naturally noisy, so for such protocols
to even make sense, one needs an extra layer of error correction that
is designed to correct the specific noise on the channel. Here, we
formalize this as a constructive statement. Let $\pi^\ecc_A$ encode a
message with an error correcting code (ECC) given by the completely
positive, trace\-/preserving (CPTP) map $\cE$, and $\pi^\ecc_B$ decode
the message with the CPTP map $\cD$, as illustrated in
\figref{fig:noiseless.honest.real}. And let $\eps^\ecc$ be the error
of the ECC for noise given by a CPTP map $\cF$, i.e.,
\begin{equation} \label{eq:ecc.noise} \frac{1}{2} \dianorm{\cD \circ
    \cF \circ \cE - \id} \leq \eps^\ecc \,.\end{equation} Then
$\pi^{\ecc}_{AB} = (\pi^\ecc_A,\pi^\ecc_B)$ constructs a noiseless
(insecure) channel $\bar{\aC}_\square$ from a noisy (insecure) channel
$\aC_\sharp$, if the filter $\sharp_E$ introduces the noise $\cF$,
i.e.,
\begin{equation}
\label{eq:security.ecc}
\bar{\aC}_\sharp \xrightarrow{\pi^\ecc_{AB},\eps^\ecc} \aC_\square\,.
\end{equation}

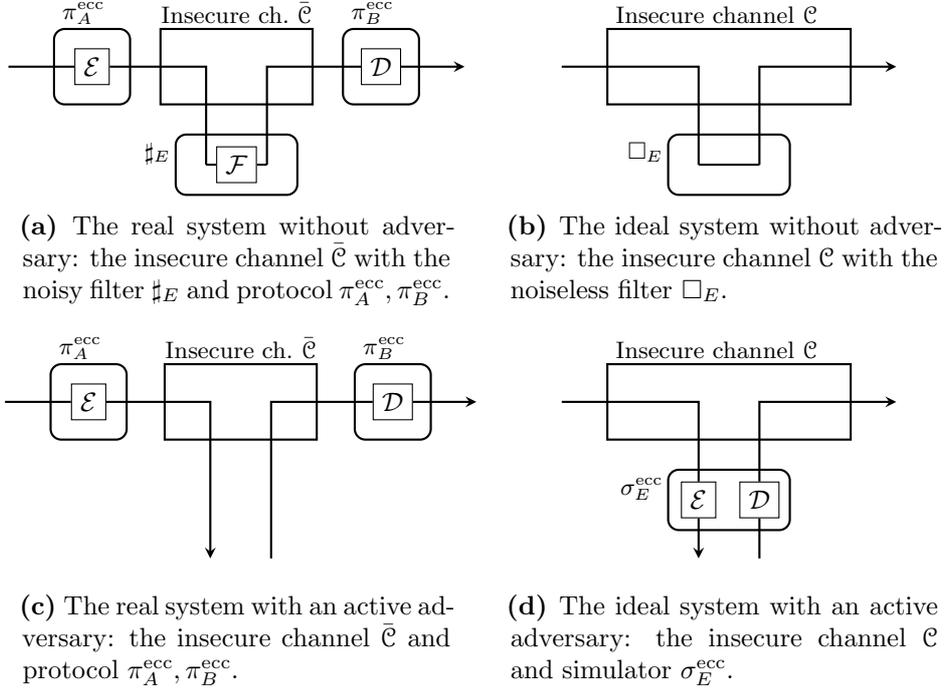
\begin{figure}[tb]
\subcaptionbox[Real noiseless system without
adversary]{\label{fig:noiseless.honest.real}The real system without
  adversary: the insecure channel $\bar{\aC}$ with the noisy filter
  $\sharp_E$ and protocol $\pi^\ecc_A,\pi^\ecc_B$.}[.5\textwidth][c]{
\begin{tikzpicture}[
resource/.style={draw,thick,minimum width=2cm,minimum height=1cm},
filter/.style={draw,thick,minimum width=1.6cm,minimum height=.8cm,rounded corners},
protocol/.style={draw,thick,minimum width=1cm,minimum height=1cm,rounded corners},
sArrow/.style={->,>=stealth,thick},
sLine/.style={-,thick}]

\small

\def\t{3} 
\def\a{1.9} 
\def\v{1.3} 
\def\s{.4}

\node[resource] (ch) at (0,0) {};
\node[yshift=-1.5,above] at (ch.north) {\footnotesize
  Insecure ch.\ $\bar{\aC}$};
\node[filter] (fil) at (0,-\v) {};
\node[draw] (noise) at (0,-\v) {$\cF$};
\node[xshift=2,below left] at (fil.north west) {\footnotesize $\sharp_E$};
\node[protocol] (protA) at (-\a,0) {};
\node[draw] (piA) at (-\a,0) {$\cE$};
\node[yshift=-2,above right] at (protA.north west) {\footnotesize $\pi^\ecc_A$};
\node[protocol] (protB) at (\a,0) {};
\node[draw] (piB) at (\a,0) {$\cD$};
\node[yshift=-2,above right] at (protB.north west) {\footnotesize $\pi^\ecc_B$};
\node (alice) at (-\t,0) {};
\node (bob) at (\t,0) {};

\draw[sLine] (alice.center) to (piA);
\draw[sLine] (piA) to (-\s,0) to (-\s,-\v);
\draw[sLine] (-\s,-\v) to (noise);
\draw[sLine] (noise) to (\s,-\v);
\draw[sLine] (\s,-\v) to (\s,0) to (piB);
\draw[sArrow] (piB) to (bob.center);

\end{tikzpicture}} \subcaptionbox[Ideal noiseless system without
adversary]{\label{fig:noiseless.honest.ideal}The ideal system without
  adversary: the insecure channel $\aC$ with the noiseless filter
  $\square_E$.}[.5\textwidth][c]{
\begin{tikzpicture}[
resource/.style={draw,thick,minimum width=3.2cm,minimum height=1cm},
filter/.style={draw,thick,minimum width=1.6cm,minimum height=.8cm,rounded corners},
sArrow/.style={->,>=stealth,thick},
sLine/.style={-,thick}]

\small

\def\t{2.2} 
\def\v{1.3} 
\def\s{.4}

\node[resource] (ch) at (0,0) {};
\node[yshift=-1.5,above right] at (ch.north west) {\footnotesize
  Insecure channel $\aC$};
\node[filter] (fil) at (0,-\v) {};
\node[xshift=2,below left] at (fil.north west) {\footnotesize $\square_E$};
\node (alice) at (-\t,0) {};
\node (bob) at (\t,0) {};

\draw[sLine] (alice.center) to (-\s,0) to (-\s,-\v);
\draw[sLine] (-\s,-\v) to (\s,-\v);
\draw[sArrow] (\s,-\v) to (\s,0) to (bob.center);

\end{tikzpicture}}

\vspace{6pt}

\subcaptionbox[Real noiseless system with
adversary]{\label{fig:noiseless.dishonest.real}The real system with an
  active adversary: the insecure channel $\bar{\aC}$ and protocol
  $\pi^\ecc_A,\pi^\ecc_B$.}[.5\textwidth][c]{
\begin{tikzpicture}[
resource/.style={draw,thick,minimum width=2cm,minimum height=1cm},
filter/.style={draw,thick,minimum width=1.6cm,minimum height=.8cm,rounded corners},
protocol/.style={draw,thick,minimum width=1cm,minimum height=1cm,rounded corners},
sArrow/.style={->,>=stealth,thick},
sLine/.style={-,thick}]

\small

\def\t{3.1} 
\def\a{2} 
\def\v{1.4} 
\def\s{.4}
\def\z{2.2}

\node[resource] (ch) at (0,0) {};
\node[yshift=-1.5,above] at (ch.north) {\footnotesize
  Insecure ch.\ $\bar{\aC}$};
\node[protocol] (protA) at (-\a,0) {};
\node[draw] (piA) at (-\a,0) {$\cE$};
\node[yshift=-2,above right] at (protA.north west) {\footnotesize $\pi^\ecc_A$};
\node[protocol] (protB) at (\a,0) {};
\node[draw] (piB) at (\a,0) {$\cD$};
\node[yshift=-2,above right] at (protB.north west) {\footnotesize $\pi^\ecc_B$};
\node (alice) at (-\t,0) {};
\node (bob) at (\t,0) {};
\node (eveL) at (-\s,-\z) {};
\node (eveR) at (\s,-\z) {};

\draw[sLine] (alice.center) to (piA);
\draw[sArrow] (piA) to (-\s,0) to (eveL);
\draw[sLine] (eveR) to (\s,0) to (piB);
\draw[sArrow] (piB) to (bob.center);

\end{tikzpicture}}
\subcaptionbox[Ideal noiseless system with
adversary]{\label{fig:noiseless.dishonest.ideal}The ideal system with
  an active adversary: the insecure channel $\aC$
  and simulator $\sigma^\ecc_E$.}[.5\textwidth][c]{
\begin{tikzpicture}[
resource/.style={draw,thick,minimum width=3.2cm,minimum height=1cm},
filter/.style={draw,thick,minimum width=1.6cm,minimum height=.8cm,rounded corners},
sArrow/.style={->,>=stealth,thick},
sLine/.style={-,thick}]

\small

\def\t{2.2} 
\def\v{1.3} 
\def\s{.4}
\def\z{2.2} 

\node[resource] (ch) at (0,0) {};
\node[yshift=-1.5,above right] at (ch.north west) {\footnotesize
  Insecure channel $\aC$};
\node[filter] (fil) at (0,-\v) {};
\node[xshift=2,below left] at (fil.north west) {\footnotesize $\sigma^\ecc_E$};
\node (alice) at (-\t,0) {};
\node (bob) at (\t,0) {};
\node (eveL) at (-\s,-\z) {};
\node (eveR) at (\s,-\z) {};

\node[draw] (simL) at (-\s,-\v) {$\cE$};
\node[draw] (simR) at (\s,-\v) {$\cD$};

\draw[sLine] (alice.center) to (-\s,0) to (simL);
\draw[sArrow] (simL) to (eveL);
\draw[sLine] (eveR) to (simR);
\draw[sArrow] (simR) to (\s,0) to (bob.center);

\end{tikzpicture}}

\caption[From noisy to noiseless]{\label{fig:noiseless}An error
  correcting code constructs a noiseless channel $\aC_\square$ from a
  noisy channel $\bar{\aC}_\sharp$.}
\end{figure}

Note that the resources $\aC$ and $\bar{\aC}$ are of different
dimension, since an error correcting code will map a quantum state to
a new one of larger dimension. In this work we generally we do not use
different notation for channels of different dimensions, since the
dimension is usually clear from the context, and we juste write $\aC$
for an insecure channel. We distinguish a noiseless from a noisy
channel by its filter, $\square_E$ and $\sharp_E$, respectively.

\begin{lem}
\label{lem:noiseless}
If the filter $\sharp_E$ introduces noise given by a CPTP map $\cF$,
and the encoding and decoding maps $\cE$ and $\cD$ satisfy
\eqnref{eq:ecc.noise}, then the protocol $\pi^{\ecc}_{AB}$ that uses
this ECC satisfies \eqnref{eq:security.ecc}.
\end{lem}

\begin{proof}
  We need to prove that
  $d(\pi^{\ecc}_{AB}\bar{\aC}\sharp_E, \aC\square_E) \leq \eps^\ecc$
  and $d(\pi^{\ecc}_{AB}\bar{\aC}, \aC\sigma^\ecc_E) \leq \eps^\ecc$
  for some simulator $\sigma^\ecc_E$ to satisfy the two conditions
  from \defref{def:security}. These systems are drawn in
  \figref{fig:noiseless}. One can easily check from the figure that
  the first condition holds because the ECC was designed to achieve
  exactly this.  $\pi^{\ecc}_{AB}\bar{\aC}$ and $\aC\sigma^\ecc_E$ are
  each a pair of channels, one performing the encoding operation
  $\cE$, the other the decoding operation $\cD$, so the second
  condition holds with distance $0$.
\end{proof}

Naturally, \lemref{lem:noiseless} only makes sense if there exists a
code than can correct the errors introduced by $\cF$, i.e., if there
exist maps $\cE,\cD$ that satisfy \eqnref{eq:ecc.noise}. In this work,
when we talk about a noisy channel resource, we always mean such a
channel that has non\-/zero capacity.

\subsection{Obtaining more key}
\label{sec:construction.morekey}

It is trivial to share secret key using a secure channel: Alice
generates a uniform string and sends it on the secure channel to
Bob. If the channel can transmit an $m+r$ qubit state and we use it to
share a $r$ bit key, then $m$ qubits can still be used to transmit a
message. Writing this using the AC resource theory notation, we get
\[ S^{m+r} \xrightarrow{} S^m \| \tilde{\aK}^{r}\,,\]
where $\tilde{\aK}^r$ generates an $r$-bit key, but allows Eve to
prevent Bob from receiving it.

Although this statement is correct, it is somewhat inconvenient, since
$\tilde{\aK}^{r}$ allows Eve to control whether Bob receives the key,
but Alice does not know whether he received it; whereas in the real
quantum authentication system analyzed, Alice actually learns whether
Bob receives her message or not, since he sends her a confirmation bit
on the backwards authentic channel. A stronger statement can be made
if instead of using the final secure channel $\aS^m$ constructed in
\corref{cor:auth} to share secret key, we use the intermediary secure
channel \& key resource $\aQ^{m,\nu_\rej,\nu_\acc}$ constructed in
\thmref{thm:auth} (and illustrated in \figref{fig:goal.alt}), in which
the secure channel $\aS^m$ and fresh key
$\bar{\aK}^{\nu_{\rej},\nu_{\acc}}$ are merged into one, with only one
input bit at Eve's interface controlling both the length of the key
and whether the message is delivered.

Using this secure channel to distribute key we get the system drawn in
\figref{fig:morekey}, where Alice's converter $\pi^\key_A$ only
outputs the $r$ bits of key she inserted on the channel if she gets
confirmation that Bob receives the message, i.e., if she obtains the
longer key.

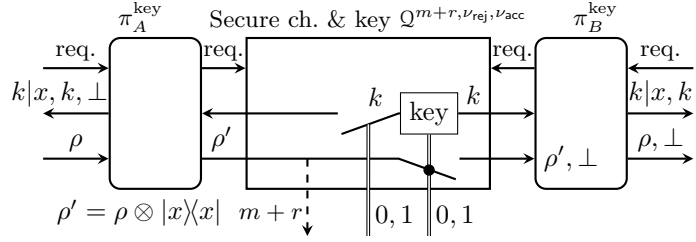
\begin{figure}[tb]
\begin{centering}

\begin{tikzpicture}[
resource/.style={draw,thick,minimum width=3.2cm,minimum height=2cm},
sArrow/.style={->,>=stealth,thick},
sLine/.style={-,thick},
sLeak/.style={->,>=stealth,thick,dashed},
protocol/.style={draw,thick,minimum width=1.2cm,minimum height=2cm,rounded corners}]

\small

\def\t{4.4} 
\def\a{2.8} 
\def\v{.6}
\def\w{.4}
\def\x{.8}
\def\z{1.75} 

\node[resource] (channel) at (0,0) {};
\node[yshift=-1.5,above] at (channel.north) {\footnotesize
  Secure ch.\ \& key $\aQ^{m+r,\nu_{\rej},\nu_{\acc}}$};
\node[protocol] (piA) at (-\a,0) {};
\node[yshift=-1.5,above right] at (piA.north west) {\footnotesize $\pi_A^\key$};
\node[below,xshift=-6] at (piA.south) {$\rho'=\rho\tensor\proj{x}$};

\node[protocol] (piB) at (\a,0) {};
\node[yshift=-1.5,above left] at (piB.north east) {\footnotesize $\pi_B^\key$};
\node (aliceUp) at (-\t,\v) {};
\node (aliceMiddle) at (-\t,0) {};
\node (aliceDown) at (-\t,-\v) {};
\node (bobUp) at (\t,\v) {};
\node (bobMiddle) at (\t,0) {};
\node (bobDown) at (\t,-\v) {};
\node (eveLeft) at (-\x,-\z) {};
\node (eveCenter) at (0,-\z) {};
\node (eveRight) at (\x,-\z) {};

\draw[sLine] (piA.east |- aliceDown) to node[auto,pos=.1,yshift=-1]
{$\rho'$} (\x-\w,-\v) to +(340:2*\w);
\draw[sArrow]  (\x+\w,-\v) to 
(piB.west |- bobDown);
\node[right] at (piB.west |- bobDown) {$\rho',\bot$};
\draw[sLeak] (eveLeft |- aliceDown) to node[auto,swap,pos=.7,xshift=2]
{\footnotesize $m+r$}
(eveLeft);

\draw[sArrow] (aliceDown) to node[auto,yshift=-1] {$\rho$} (piA.west |- aliceDown);
\draw[sArrow] (piB.east |- bobDown) to node[auto,pos=.5,yshift=-1] {$\rho,\bot$} (bobDown);

\node[draw] (key) at (\x,0) {key};
\draw[sArrow] (key) to node[auto,pos=.2] {$k$} (piB);
\draw[sLine] (key) to (\w,0) to node[auto,swap,pos=.1] {$k$}
node[pos=.53] (handle) {} +(200:2*\w);
\draw[sArrow] (-\w,0) to (piA);

\draw[sArrow] (piA) to node[auto,swap,pos=.75,yshift=-1] {$k|x,k,\bot$} (aliceMiddle);
\draw[sArrow] (piB) to node[auto,pos=.5,yshift=-1] {$k|x,k$} (bobMiddle);

\draw[double] (eveRight) to node[pos=.19,xshift=-1,auto,swap] {$0,1$}
node[circle,fill,pos=.65,inner sep=1.5] {}  (key);
\draw[double] (eveCenter) to node[pos=.17,xshift=-1,auto,swap]
{$0,1$} (handle.center);

\draw[sArrow] (piA.east |- aliceUp) to node[auto,yshift=-1] {\footnotesize req.} (channel.west |- aliceUp);
\draw[sArrow] (piB.west |- bobUp) to node[auto,swap,yshift=-1] {\footnotesize req.} (channel.east |- bobUp);
\draw[sArrow] (aliceUp) to node[auto,yshift=-1] {\footnotesize req.} (piA.west |- aliceUp);
\draw[sArrow] (bobUp) to node[auto,swap,yshift=-1] {\footnotesize req.} (piB.east |- bobUp);

\end{tikzpicture}

\end{centering}
\caption[Key vs. message length]{\label{fig:morekey}Distributing a
  secret key using the channel $\aQ^{m,\nu_{\rej},\nu_{\acc}}$
  constructed by a quantum authentication code. Alice generates a
  random string $x$ of length $r$ and appends it to an $m$ qubit
  message $\rho$. If Bob successfully authenticates the message, he
  appends $x$ to the recycled key. If Alice gets a confirmation from
  Bob that he received the correct message, she appends $x$ to the
  recycled key.}
\end{figure}

\begin{lem}
\label{lem:morekey}
Let $\pi^{\key}_{AB} = (\pi^{\key}_{A},\pi^{\key}_{A})$ denote the
protocol described above. It constructs a
$\aQ^{m,\nu_{\rej},\nu_{\acc}+r}_\natural$ secure channel from a
$\aQ^{m+r,\nu_{\rej},\nu_{\acc}}_\natural$ secure channel with no error,
i.e.,
\[ \aQ^{m+r,\nu_{\rej},\nu_{\acc}}_\natural \xrightarrow{\pi^{\key}_{AB},0}
\aQ^{m,\nu_{\rej},\nu_{\acc}+r}_\natural\,. \]
\end{lem}

\begin{proof}
This lemma trivially holds with a simulator that changes the length of
the message leaked at Eve's interface from $m+r$ to $m$ and forwards
the two control bits to the constructed $\aQ^{m,\nu_{\rej},\nu_{\acc}+r}$.
\end{proof}

As in \corref{cor:auth}, it follows that
\[ \aQ^{m+r,\nu_{\rej},\nu_{\acc}}_\natural
\xrightarrow{\pi^{\key}_{AB},0} (\aS^m \|
\bar{\aK}^{\nu_\rej,\nu_\acc+r})_\flat\,.\]

\subsection{Putting it together}
\label{sec:construction.together}

If we compose all the protocols described above, we obtain the system
depicted in \figref{fig:auth.complete}. Let $\pi_A$ denote the
composition of Alice's converters and $\pi_B$ denote the composition
of Bob's converters. We then immediately get that the
$\pi_{AB} = (\pi_A,\pi_B)$ constructs a secure channel and secret key
filtered resource, $(\aS^{m} \| \bar{\aK}^{\nu_{\rej},\nu_{\acc}+r})_\flat$, from secret keys
$\aK^{\mu}$ and $\aK^{\eta}$, and two noisy channels
$\overrightarrow{\aC_{\sharp}}$ and $\overleftarrow{\aC_{\sharp}}$
with error
$\eps = \max
\left\{\eps^{\qauth}+\eps^{\cauth},\overrightarrow{\eps}^\ecc +
    \overleftarrow{\eps}^\ecc\right\}$,
    where the arrows are used to distinguisher the forwards quantum
    channel and the backwards classical channel, i.e,
    \[ \aK^{\mu} \| \aK^{\eta} \| \overrightarrow{\aC_{\sharp}} \|
    \overleftarrow{\aC_{\sharp}} \xrightarrow{\pi_{AB},\eps} (\aS^{m}
    \| \bar{\aK}^{\nu_{\rej},\nu_{\acc}+r})_\flat\,.\]

\begin{figure}[tb]
\begin{centering}

\begin{tikzpicture}[
resource1/.style={draw,thick,minimum width=1.5cm,minimum height=1cm},
resource2/.style={draw,thick,minimum width=3cm,minimum height=1cm},
sArrow/.style={->,>=stealth,thick},
sLine/.style={-,thick},
protocol1/.style={draw,thick,minimum width=.6cm,minimum height=1cm,rounded corners},
protocol2/.style={draw,thick,minimum width=.6cm,minimum height=2.5cm,rounded corners},
protocol4/.style={draw,thick,minimum width=.6cm,minimum height=5.5cm,rounded corners},
pnode/.style={minimum width=.6cm,minimum height=1cm}]

\small

\def\v{1.5}
\def\w{.25}
\def\x{.75} 
\def\y{.75} 
\def\p{2.2} 
\def\c{1} 
\def\t{6.3}
\def\z{1.1} 

\node[resource2,green] (keyLong) at (0,3*\v) {};
\node[yshift=-1.5,above right] at (keyLong.north west) {\footnotesize
  Secret key $\aK^\mu$};
\node[resource2,red] (keyShort) at (0,2*\v) {};
\node[yshift=-1.5,above right] at (keyShort.north west) {\footnotesize
  Secret key $\aK^\eta$};
\node[resource2,violet] (chBack) at (0,\v) {};
\node[yshift=-1.5,above right] at (chBack.north west) {\footnotesize
  Insecure channel $\overleftarrow{\aC}$};
\node[resource1,blue] (chForward) at (-\x,0) {};
\node[yshift=-1.5,above right] at (chForward.north west) {\footnotesize
  Ins.\ ch.\ $\overrightarrow{\aC}$};

\node[protocol1,violet] (pECCbackA) at (-\p,\v) {};
\node[yshift=-1.5,above] at (pECCbackA.north) {\footnotesize $\overleftarrow{\pi}^\ecc_A$};
\node[protocol1,violet] (pECCbackB) at (\p,\v) {};
\node[yshift=-1.5,above] at (pECCbackB.north) {\footnotesize $\overleftarrow{\pi}^\ecc_B$};
\node[protocol1,blue] (pECCforwardA) at (-\p,0) {};
\node[yshift=-1.5,above] at (pECCforwardA.north) {\footnotesize $\overrightarrow{\pi}^\ecc_A$};
\node[protocol1,blue] (pECCforwardB) at (\p,0) {};
\node[yshift=-1.5,above] at (pECCforwardB.north) {\footnotesize $\overrightarrow{\pi}^\ecc_B$};
\node[protocol2,red] (pCauthA) at (-\p-\c,3*\v/2) {};
\node[yshift=-1.5,above] at (pCauthA.north) {\footnotesize $\pi^\cauth_A$};
\node[protocol2,red] (pCauthB) at (\p+\c,3*\v/2) {};
\node[yshift=-1.5,above] at (pCauthB.north) {\footnotesize $\pi^\cauth_B$};
\node[protocol4,green] (pQauthA) at (-\p-2*\c,3*\v/2) {};
\node[yshift=-1.5,above] at (pQauthA.north) {\footnotesize $\pi^\qauth_A$};
\node[protocol4,green] (pQauthB) at (\p+2*\c,3*\v/2) {};
\node[yshift=-1.5,above] at (pQauthB.north) {\footnotesize $\pi^\qauth_B$};
\node[protocol4,brown] (pKeyA) at (-\p-3*\c,3*\v/2) {};
\node[yshift=-1.5,above] (keyLabelA) at (pKeyA.north) {\footnotesize $\pi^\key_A$};
\node[protocol4,brown] (pKeyB) at (\p+3*\c,3*\v/2) {};
\node[yshift=-1.5,above] (keyLabelB) at (pKeyB.north) {\footnotesize $\pi^\key_B$};

\node[draw,thick,dashed,fit=(pECCforwardA)(keyLabelA)(pKeyA),inner sep=4] (compA) {};
\node[yshift=-1.5,above right] at (compA.north west) {\footnotesize
  Composed protocol $\pi_A$};
\node[draw,thick,dashed,fit=(pECCforwardB)(keyLabelB)(pKeyB),inner sep=4] (compB) {};
\node[yshift=-1.5,above right] at (compB.north west) {\footnotesize
  Composed protocol $\pi_B$};

\node (aliceUp) at (-\t,5*\v/2) {};
\node (aliceMiddle) at (-\t,3*\v/2) {};
\node (aliceDown) at (-\t,\v/2) {};
\node (bobUp) at (\t,5*\v/2) {};
\node (bobMiddle) at (\t,3*\v/2) {};
\node (bobDown) at (\t,\v/2) {};
\node (eveLeftL) at (-\x-\w,-\z) {};
\node (eveLeftR) at (-\x+\w,-\z) {};
\node (eveRightL) at (\y-\w,-\z) {};
\node (eveRightR) at (\y+\w,-\z) {};

\node[pnode] (aQ) at (-\p-2*\c,3*\v) {};
\node[pnode] (bQ) at (\p+2*\c,3*\v) {};
\node[pnode] (aK1) at (-\p+\c,3*\v) {};
\node[pnode] (bK1) at (\p-\c,3*\v) {};
\node[pnode] (aC) at (-\p-\c,2*\v) {};
\node[pnode] (bC) at (\p+\c,2*\v) {};
\node[pnode] (aK2) at (-\p+\c,2*\v) {};
\node[pnode] (bK2) at (\p-\c,2*\v) {};

\draw[sArrow] (aliceUp.center) to node[auto,pos=.4] {\footnotesize req.} (pKeyA.west |- aliceUp);
\draw[sArrow] (pKeyA.west |- aliceMiddle) to node[auto,swap,pos=.6] {$k$} (aliceMiddle.center);
\draw[sArrow] (aliceDown.center) to node[auto,pos=.4] {$\rho$} (pKeyA.west |- aliceDown);
\draw[sArrow] (bobUp.center) to node[auto,pos=.4,swap] {\footnotesize req.} (pKeyB.east |- aliceUp);
\draw[sArrow] (pKeyB.east |- aliceMiddle) to node[auto,pos=.6] {$k$} (bobMiddle.center);
\draw[sArrow] (pKeyB.east |- aliceDown) to node[auto,pos=.6] {$\rho$} (bobDown.center);

\draw[sArrow] (pKeyA.east |- aliceUp) to (pQauthA.west |- aliceUp);
\draw[sArrow] (pQauthA.west |- aliceMiddle) to (pKeyA.east |- aliceMiddle);
\draw[sArrow] (pKeyA.east |- aliceDown) to (pQauthA.west |- aliceDown);
\draw[sArrow] (pKeyB.west |- aliceUp) to (pQauthB.east |- aliceUp);
\draw[sArrow] (pQauthB.east |- aliceMiddle) to (pKeyB.west |- aliceMiddle);
\draw[sArrow] (pQauthB.east |- aliceDown) to (pKeyB.west |- aliceDown);

\draw[sArrow,bend left=15] (aK1) to (aQ);
\draw[sArrow,bend left=15] (aQ) to (aK1);
\draw[sArrow,bend left=15] (bK1) to (bQ);
\draw[sArrow,bend left=15] (bQ) to (bK1);
\draw[sArrow] (pCauthA) to (pQauthA);
\draw[sArrow] (pQauthB) to (pCauthB);
\draw[sArrow] (pQauthA.east |- chForward) to (pECCforwardA);
\draw[sArrow] (pECCforwardB) to (pQauthB.west |- chForward);

\draw[sArrow,bend left=20] (aK2) to (aC);
\draw[sArrow,bend left=20] (aC) to (aK2);
\draw[sArrow,bend left=20] (bK2) to (bC);
\draw[sArrow,bend left=20] (bC) to (bK2);
\draw[sArrow] (pECCbackA) to (pCauthA.east |- chBack) ;
\draw[sArrow] (pCauthB.west |- chBack) to (pECCbackB);

\draw[sArrow] (pECCforwardA) to (eveLeftL |- chForward) to (eveLeftL.center);
\draw[sArrow] (eveLeftR.center) to (eveLeftR |- chForward) to (pECCforwardB);
\draw[sArrow] (pECCbackB) to (eveRightR |- chBack) to (eveRightR.center);
\draw[sArrow] (eveRightL.center) to (eveRightL |- chBack) to (pECCbackA);

\end{tikzpicture}

\end{centering}
\caption[Complete system for quantum
authentication]{\label{fig:auth.complete} Composing all the protocols
  discussed in this work results in the system depicted here. In
  \textcolor{blue}{blue} we have drawn the systems constructing the
  noiseless insecure quantum channel, in \textcolor{violet}{violet}
  the systems constructing the backwards noiseless classical channel,
  in \textcolor{red}{red} (along with the \textcolor{violet}{violet}
  subprotocol) the systems constructing the backwards authentic
  channel, in \textcolor{green}{green} (along with the
  \textcolor{red}{red} and \textcolor{blue}{blue} subprotocols) the
  systems constructing the secure quantum channel that recycles the
  key $\aK^{\mu}$, and in \textcolor{brown}{brown} (along with the
  \textcolor{green}{green} subprotocol) the systems constructing the
  secure quantum channel that compensates for the lost key
  $\aK^{\eta}$.}
\end{figure}
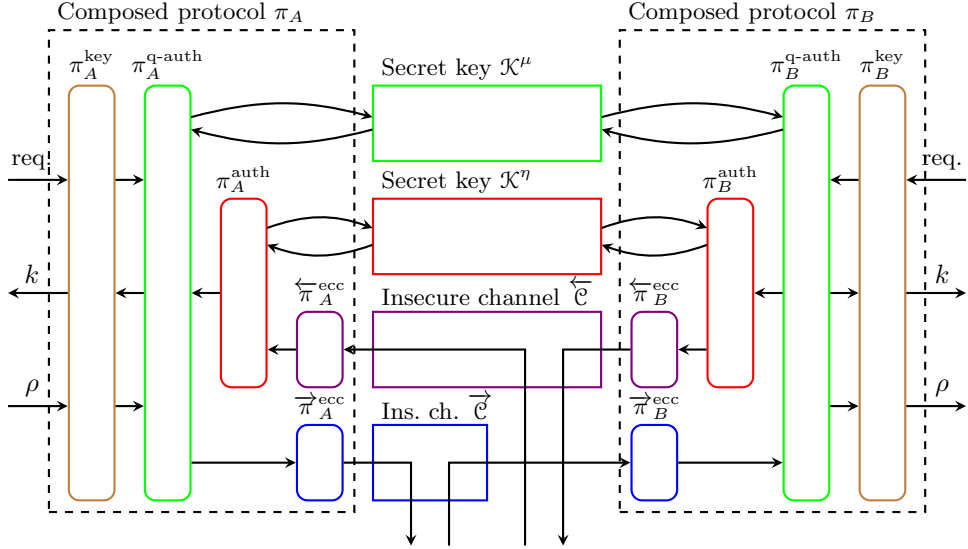

In the following we assume for simplicity that
$\overrightarrow{\eps}^\ecc + \overleftarrow{\eps}^\ecc \leq
\eps^{\qauth}+\eps^{\cauth}$
and take $\eps = \eps^{\qauth}+\eps^{\cauth}$. Plugging in
$r = \eta = n$ and the parameters from the two explicit quantum
authentication protocols from \secref{sec:auth.explicit}, we get the
following two corollaries. \corref{cor:construction.1} uses the
explicit construction proposed by Barnum et
al.~\cite{BCGST02}. \corref{cor:construction.2} uses the explicit
unitary $2$\-/design construction from Chau\cite{Cha05} (see
\secref{sec:auth.explicit} for details).

\begin{cor}
\label{cor:construction.1}
For any $m$ and $n$, there exist an explicit protocol that requires $2m+4n$ bits of
secret key, a forwards noisy insecure quantum channel and a backwards
noisy insecure classical channel, to construct a secure channel for an
$m$ qubit quantum message, which recycles all the key if the message
is accepted and $n$ bits if the message is rejected, and has error
\[\eps \leq 2^{-n/2} (1+\sqrt{2m/n+4}) + 2^{-n}(m/n+2)\,.\]
\end{cor}

\begin{proof}
  Follows from the composition theorem from AC~\cite{MR11} as well as
  \remref{rem:classicalmessage}, \corref{cor:auth.1},
  \lemref{lem:classicalauth}, \lemref{lem:noiseless}, and
  \lemref{lem:morekey}.
\end{proof}

\begin{cor}
\label{cor:construction.2}
For any $m$ and $n$, there exists an explicit protocol that requires
$5m+9n$ bits of secret key, a forwards noisy insecure quantum channel
and a backwards noisy insecure classical channel, to construct a secure
channel for an $m$ qubit quantum message, which recycles all the key
if the message is accepted and $3m+6n$ bits if the message is
rejected, and has error \[\eps \leq 2^{-n/2+1}+2^{-n-1}\,.\]
\end{cor}

\begin{proof}
  Follows from the composition theorem from AC~\cite{MR11} as well as
  \remref{rem:classicalmessage}, \corref{cor:auth.2},
  \lemref{lem:classicalauth}, \lemref{lem:noiseless}, and
  \lemref{lem:morekey}.
\end{proof}

Note that both schemes have a loss of $2m+3n$ key bits if the message
is not successfully authenticated: $2m$ bits used to one-time pad the
quantum message of length $m$, $2n$ bits to one-time pad the two
classical strings of $n$ bits \--- the $n$ bits of key sent to replace
those consumed by the backwards authentic channel and the $n$ bits
used as syndrome \--- and finally the $n$ bits consumed by the
backwards authentic channel cannot be replaced, so they are lost as
well.

\subsection{Removing the authentic channel}
\label{sec:construction.noauth}

The backwards authentic channel is crucial in the construction of the
secure channel with key recycling analyzed in this section, because
without it, Alice would not know that her message arrived \--- let
alone whether it was accepted or not \--- and thus not be able to
recycle the key. One can however skip the authentic channel if the
players share a stronger resource, e.g., another secure quantum
channel that recycles key.

This is the case if the players construct secure quantum channels in
both directions and alternate between the two: first Alice sends a quantum
message to Bob, then Bob to Alice, then Alice to Bob, etc. Let them
share two sets of keys, one set is used for the forward communication,
the other is used for the backward communication. If at any point a
message is not successfully authenticated, this means that an
eavesdropper is disturbing the communication, and the players abort
and stop communicating. Thus, Bob only sends his next message to Alice
if he successfully received her message. So if Alice successfully
authenticates Bob's message, she knows that hers was received, and can
recycle all of her key to send the next message. We thus avoid any
explicit confirmation of reception, since sending the next message is
in itself the confirmation.

One can easily show that the composition of $n$ rounds of this
protocol, each round sending an $m$ qubit quantum message, constructs
a secure channel for $nm$ qubits, which we denote $\aR^{m,n}$. Unlike
the resource $\aS^{m}$ which gave the adversary a $1$ bit input to
decide if the message is delivered or not, $\aR^{m,n}$ provides the
adversary with a bit of input for every block of $m$ qubits \--- but
if one block is prevented from being delivered, none of the subsequent
messages are delivered either. The error of this construction is $n$
times the error of each round, $\eps = n\eps^{\qauth}$. This may
continue arbitrary long if no adversary introduces noise on the
channel, since the quantum authentication protocol recycles every bit
of key.

\section{Discussion and open questions}
\label{sec:conclusion}

The family of quantum authentication protocols of Barnum et
al.~\cite{BCGST02} as well as the subset analyzed in this work are
large classes, which include many protocols appearing independently in
the literature. The signed polynomial code \cite{BCGHS06,ABE10}, the
Clifford code \cite{ABE10,DNS12,BW16} (which is a unitary
$3$\-/design~\cite{Zhu15,Web15}) and the unitary $8$\-/design scheme from
\cite{GYZ16} are all instances which use a strong purity testing
code. Our results apply directly to the Clifford and unitary
$8$\-/design schemes \--- which have in the same error as the unitary
$2$\-/design scheme from \corref{cor:auth.2}. But the signed polynomial code
uses an ECC on qudits, not qubits, so our proof does not cover this
case, and would have to be adapted to do so.

The trap code~\cite{BGS13,BW16} is an example of a quantum
authentication scheme that uses a purity testing code that is not a
strong purity testing code, i.e., errors which do not modify the
message do not necessarily provoke an abort. For example, if the
adversary performs a simple bit flip in one position, this will
provoke an abort with probability $2/3$ in the variant from
\cite{BGS13} and with probability $1/3$ in the variant from
\cite{BW16}, but leaves the message unmodified if no abort occurs. If
the adversary learns whether Bob accepted the message or not, she will
learn whether the ECC used detects that specific bit flip or not, and
thus learn something about the key used to select the ECC. Hence, the
players cannot recycle the entire key, even in the case where the
message is accepted. The restriction to strong purity testing codes is
thus necessary to recycle every bit. It remains open how many bits of
key can be recycled with the trap code, but we conjecture that this
bit leaked due the decision to abort or not is the only part of the
key leaked, and the rest can be recycled.

Another quantum authentication scheme, Auth-QFT-Auth, has been
proposed in \cite{GYZ16}, where the authors prove that some of the key
can be recycled as well. We do not know if this scheme fits in the family
from~\cite{BCGST02} or not.

In the classical case, almost strongly $2$\-/universal hash
functions~\cite{WC81,Sti94} are used for authentication, and any new
family of such functions immediately yields a new MAC. Likewise, any
new purity testing code provides a new quantum authentication
scheme. However, it is unknown whether all quantum authentication
schemes can be modeled as a combination of a one-time pad and a purity
testing code, or whether there exist interesting schemes following a
different pattern.

We have proven that a loss of $2m+n$ bits of key is inevitable with
these schemes if the adversary tampers with the channel. In the case
of the unitary $2$\-/design scheme, which has the smallest error, this
is $2m + 2 \log 1/\eps + 2$ bits of key which are consumed. A loss of
$2m$ bits will always occur, since these are required to one-time pad
the message. It remains open whether there exist other schemes \---
which do not fit the one-time pad + purity testing code model \---
which recycle more key.

The initial preprint of this work suggested that one should also
investigate whether it is possible to find a prepare\-/and\-/measure
scheme to encrypt and authenticate a classical message in a quantum
state, so that all of the key may be recycled if it is successfully
authenticated. At the time of writing, a possible solution had already
been found by Fehr and Salvail~\cite{FS17}. Their protocol is however
not known to be composable, and it remains open to prove that it
achieves the desired result in such a setting.

\section*{Acknowledgments}
\pdfbookmark[1]{Acknowledgments}{sec:ack}

CP would like to thank Anne Broadbent, Fr\'ed\'eric Dupuis and Debbie
Leung for useful discussions.

CP is supported by the European Commission FP7 Project RAQUEL
(grant No.~323970), US Air Force Office of Scientific Research (AFOSR)
via grant~FA9550-16-1-0245, the Swiss National Science Foundation (via
the National Centre of Competence in Research `Quantum Science and
Technology') and the European Research Council -- ERC (grant
No.~258932).

\appendix
\appendixpage
\phantomsection
\label{app}
\pdfbookmark[1]{Appendices}{app}

In \appendixref{app:ac} we provide a longer introduction to AC framework
than already present in \secref{sec:cc}. In \appendixref{app:design}
we give an introduction to unitary $2$-designs and prove that they are
strong purity testing codes. \appendixref{app:lemmas} contains some
technical lemmas used in the body of this work. And
in \appendixref{app:norecycle} we provide a security proof for quantum
authentication without key recycling.

\section{Abstract cryptography}
\label{app:ac}

As already mentioned in \secref{sec:intro.contributions}, the AC
framework~\cite{MR11} models cryptography as a resource theory. The AC
framework does however not explicitly define these resources. It
follows a top\-/down paradigm and only specifies on each level of
abstraction the properties of objects that are absolutely essential
\--- the axioms these objects must satisfy. This simplifies the
framework by removing unnecessary and cumbersome information \---
e.g., a model of computation \--- and results in more general
framework that is not hard\-/coded with a specific communication or
scheduling model. In this section we give a brief introduction to
AC. We illustrate this with an example taken from \cite{Por14}, namely
that appending a message authentication code (MAC) to a classical
message is sufficient to construct a classical authentic channel given
an insecure channel and a shared secret key. We refer the interested
reader to the original AC paper~\cite{MR11} for a detailed treatment
of the abstract layer and to \cite{Mau12,PR14} for more gentle
introductions to AC.

\subsection{Resources}
\label{app:ac.objects}

An $\cI$\=/\emph{resource} is an (abstract) system with interfaces
specified by a set $\cI$, e.g., $\cI = \{A,B,E\}$. Each interface
$i \in \cI$ is accessible to a user $i$. The objects depicted in
\figref{fig:channel.dishonest} are examples of resources. The insecure
channel in \figref{fig:channel.dishonest.insecure} allows Alice to
input a message at her interface on the left and allows Bob to receive
a message at his interface on the right. Eve can intercept Alice's
message and insert a message of her choosing at her
interface. Mathematically, this can be captured by two identity
channels, one from Alice to Eve and one from Eve to Bob, which may be
used independently. The authentic channel resource depicted in
\figref{fig:channel.dishonest.authentic} also allows Alice to send a
message and Bob to receive a message, but Eve's interface is more
limited than for the insecure channel: she can only decide if Bob
receives the message, an error symbol or nothing at all \--- by
inputing $0$, $1$, or nothing, respectively, at her interface \--- but
not tamper with the message being sent. This can be modeled
mathematically by two maps. The first goes from Alice to Eve and some
internal memory of the system: upon receiving Alice's message, it
stores a copy in the memory sends the other to Eve. The second map
goes from Eve and the internal memory to Bob: Bob either receives a
copy of the memory or an error message depending on Eve's input. Note
that if the maps are activated in the other order, Bob will receive an
error regardless of Eve's input value. The key resource drawn in
\figref{fig:resource.key} provides each player with a secret key when
requested. This can be modeled by two maps, which each take the input
\texttt{req.} from a player and return the key value (which is chosen
uniformly at random). In \appendixref{app:ac.instantiation} we discuss
how to model resources in general as mathematical objects.

Additionally, resources are equipped with a parallel composition
operator, $\|$, that maps two resources to another resource. This is
to be understood as both resources being merged into one, that can be
accessed in any arbitrary order. For example, if players share a
secret key resource $\aK$ and a channel resource $\aC$, they have the
resource $\aK \| \aC$ available, as depicted in
\figref{fig:resource.composition}. Given access to $\aK \| \aC$, the
players could, e.g., get a key from $\aK$ and use it to encode a
message that is sent on $\aC$

\emph{Converters} capture operations that a player might perform
locally at her interface. These are (abstract) systems with two
interfaces, an \emph{inside} interface and an \emph{outside}
interface. The inside interface connects to an interface of a
resource, and the outside interface becomes the new interface of the
resource resulting from the connection of this converter and
resource. For example, in the setting described a paragraph higher,
Alice might decide to append a MAC to her message. This is modeled as
a converter $\pi^\cauth_A$ that obtains the message $x$ at the outside
interface, obtains a key at the inside interface from a key resource
$\aK$ and sends $(x,h_k(x))$ on the insecure channel $\aC$, where
$h_k$ is taken from a family of strongly $2$\-/universal hash
functions~\cite{WC81,Sti94}. We illustrate this in
\figref{fig:classicalauth.real}. Converters are always drawn with
rounded corners.

If a converter $\alpha_i$ is connected to the $i$ interface of a
resource $\aR$, we write $\alpha_i\aR$ or $\aR\alpha_i$ for the new
resource obtained by connecting the two \--- in this work we adopt the
convention of writing converters at the $A$ and $B$ interfaces on the
left and converters at the $E$ interface on the right, though there is
no mathematical difference between $\alpha_i\aR$ and $\aR\alpha_i$.
Serial and parallel composition of converters is defined as follows:
\[(\alpha _i\beta _i) \aR \coloneqq \alpha_i (\beta_i \aR) \qquad
\text{and} \qquad (\alpha_i \| \beta_i) (\aR\|\aS) \coloneqq (\alpha_i
\aR) \| (\beta_i \aS) \,. \]
By definition, converters at different interfaces commute, i.e.,
$\alpha_i \beta_j \aR = \beta_j \alpha_i \aR$ if $i \neq j$. This
allows us to draw systems as in \figref{fig:classicalauth.real}
without having to specify an order in which $\pi^\cauth_A$ and
$\pi^\cauth_B$ are connected to the resource $\aK \| \aC$.

A protocol is then defined by a set of converters, one for every
honest player. Another type of converter that we need is a
\emph{filter}. The resources illustrated in
\figref{fig:channel.dishonest} depict a setting with an adversary that
has some control over these resources. For a cryptographic protocol to
be useful it is not sufficient to provide guarantees on what happens
when an adversary is present, one also has to provide a guarantee on
what happens when no adversary is present, e.g., if no adversary
tampers with the message on the insecure channel, then Bob will
receive the message that Alice sent. We model this setting by covering
the adversarial interface with a filter that emulates an honest
behavior.\footnote{More generally, a filter covers the inputs and
  outputs that are only accessible to a dishonest player, but provides
  access to those that should be available to an honest player. The
  dishonest player can remove the filter to have more control over the
  resource. We however do not need such a feature in this work, since
  we only consider resources with $E$\=/interfaces that are blank if
  the adversary is not active.} In \figref{fig:channel.honest} we draw
an insecure and an authentic channel with filters $\sharp_E$ and
$\lozenge_E$ that transmit the message to Bob. In the case of the
insecure channel, one may want to model an honest noisy channel when
no adversary is present. This is done by having the filter $\sharp_E$
add some noise to the message. A dishonest player removes this and has
access to a noiseless channel as in
\figref{fig:channel.dishonest.insecure}.

We use the term \emph{filtered resource} to refer to a pair of a
resource $\aR$ and a filter $\sharp_E$, and often write
$\aR_\sharp = (\aR,\sharp_E)$. Such an object can be thought of as
having two modes: it is characterized by the resource $\aR\sharp_E$
when no adversary is present and by the resource $\aR$ when the
adversary is present. Parallel composition of filtered resources
naturally follows from parallel composition of resources and
converters:
\[ \aR_\sharp \| \aS_\lozenge \coloneqq (\aR \| \aS)_{\sharp \| \lozenge}\,.\]

The final object that is required by the AC framework to define the
notion of construction and prove that it is composable, is a
(pseudo\-/)metric defined on the space of resources that measures how
close two resources are. In the following, we use a distinguisher
based metric, i.e., the maximum advantage a distinguisher has in
guessing whether it is interacting with resource $\aR$ or $\aS$, which
we write $d(\aR,\aS)$. This is discussed further in
\appendixref{app:ac.instantiation}.

\subsection{Security definition}
\label{app:ac.security}

We are now ready to define the security of a cryptographic
protocol. We do so in the three player setting, for honest Alice and
Bob, and dishonest Eve. Thus, in the following, all resources have
three interfaces, denoted $A$, $B$ and $E$, and a protocol is then
given by a pair of converters $(\pi_A,\pi_B)$ for the honest
players. We refer to \cite{MR11} for the general case, when arbitrary
players can be dishonest. For convenience, we reproduce here
\defref{def:security}.

\begin{deff}[Cryptographic security~\cite{MR11}] \label{def:security.app}
  Let $\pi_{AB} = (\pi_A,\pi_B)$ be a protocol and $\aR_\sharp =
  (\aR,\sharp)$ and $\aS_\lozenge = (\aS,\lozenge)$ denote two
  filtered resources.  We say that \emph{$\pi_{AB}$ constructs
    $\aS_{\lozenge}$ from $\aR_\sharp$ within $\eps$}, which we write
  $\aR_\sharp \xrightarrow{\pi,\eps} \aS_\lozenge$, if the two
  following conditions hold:
\begin{enumerate}[label=\roman*), ref=\roman*]
\item \label{eq:def.cor.app} We have
  \[d(\pi_{AB}\aR\sharp_E,\aS\lozenge_E) \leq \eps \,.\]
\item \label{eq:def.sec.app} There exists a converter\footnote{For a
    protocol with information\-/theoretic security to be composable
    with a protocol that has computational security, one additionally
    requires the simulator to be efficient.} $\sigma_E$ \--- which we
  call simulator \--- such that
  \[  d(\pi_{AB}\aR,\cS\sigma_E) \leq \eps \,.\]
\end{enumerate}
If it is clear from the context what filtered resources $\aR_\sharp$
and $\aS_\lozenge$ are meant, we simply say that $\pi_{AB}$ is
$\eps$\=/secure.
\end{deff}

The first of these two conditions measures how close the constructed
resource is to the ideal resource in the case where no malicious
player is intervening, which is often called \emph{correctness} in the
literature. The second condition captures \emph{security} in the
presence of an adversary. For example, to prove that the MAC protocol
$\pi^\cauth_{AB}$ constructs an authentic channel $\aA_{\lozenge}$
from a (noiseless) insecure channel $\aC_\square$ and a secret key
$\aK$ within $\eps$, we need to prove that the real system (with
filters) $\pi^\cauth_{AB}(\aK\|\aC\square_E)$ cannot be distinguished
from the ideal system $\aA\lozenge_E$ with advantage greater than
$\eps$, and we need to find a converter $\sigma^\cauth_E$ such that
the real system (without filters) $\pi^\cauth_{AB}(\aK\|\aC)$ cannot
be distinguished from the ideal system $\aA\sigma^\cauth_E$ with
advantage greater than $\eps$. For the MAC protocol, correctness is
satisfied with error $0$ and the simulator $\sigma^\cauth_E$ drawn in
\figref{fig:classicalauth.ideal} satisfies the second requirement if
the family of hash functions $\{h_k\}_k$ is $\eps$\=/almost strongly
$2$\-/universal~\cite{Por14}.

It follows from the composition theorem of the AC
framework~\cite{MR11} that if two protocols $\pi$ and $\pi'$ are
$\eps$- and $\eps'$\=/secure, the composition of the two is
$(\eps+\eps')$\=/secure. More precisely, let protocols $\pi$ and
$\pi'$ construct $\aS_\lozenge$ from $\aR_\sharp$ and $\aT_\square$
from $\aS_\lozenge$ within $\eps$ and $\eps'$, respectively, i.e.,
\[\aR_\sharp \xrightarrow{\pi,\eps}\aS_\lozenge \qquad \text{and} \qquad \aS_\lozenge
\xrightarrow{\pi',\eps'}\aT_\square \,.\]
It then follows from the triangle inequality of the
metric that $\pi'\pi$ constructs $\aT_\square$ from $\aR_\sharp$
within $\eps+\eps'$,
\[\aR_\sharp \xrightarrow{\pi'\pi,\eps+\eps'}\aT_\square \,.\]
A similarly statement holds for parallel composition. Let $\pi$ and
$\pi'$ construct $\aS_\lozenge$ and $\aS'_\square$ from $\aR_\sharp$
and $\aR'_\flat$ within $\eps$ and $\eps'$, respectively, i.e.,
\[\aR_\sharp \xrightarrow{\pi,\eps}\aS_\lozenge \qquad \text{and}
\qquad \aR'_\flat \xrightarrow{\pi',\eps'} \aS'_\square \,.\]
If these resources and protocols are composed in parallel, we find
that $\pi \| \pi'$ constructs $\aS_\lozenge \| \aS'_\square$ from
$\aR_\sharp \| \aR'_\flat$ within $\eps+\eps'$,
\[\aR_\sharp \| \aR'_\flat \xrightarrow{\pi \| \pi',\eps+\eps'}
\aS_\lozenge \| \aS'_\square \,.\]
Proofs of these statements can be found in \cite{MR11,Mau12}.

\subsection{Instantiation}
\label{app:ac.instantiation}

As stated at the beginning of this section, the AC framework
\cite{MR11} specifies only the necessary axioms that resources and
converters must satisfy so that one can prove that the resulting
notion of construction is composable. Modeling concrete systems such
as those in Figures~\ref{fig:channel.dishonest},
\ref{fig:classicalauth.real}, \ref{fig:channel.honest} or
\ref{fig:classicalauth.ideal}, requires an instantiation of the
framework with mathematical objects that capture interactive quantum
information\-/processing systems. Such an instantiation has been given
in \cite{PMMRT17} and proven to satisfy the axioms of AC, where the
interactive systems are called \emph{causal boxes}.


Unlike the model of systems used in quantum UC~\cite{Unr10}, in which
the output of a system is given by a quantum message and a classical
string denoting a recipient, causal boxes allow messages to be sent to
a superposition of different players; they even allow superpositions
of different numbers of messages to be generated in superpositions of
different
orders~\cite{PMMRT17}.  
This generality is however not needed in the current work, because all
converters and resources involved in the construction of secure
channels have a very simple structure. For this work, a system $\aS$
can be modeld as having internal memory $\hilbert_{\text{mem}}$, and
sets of in-ports $\mathsf{In}$ and out-ports $\mathsf{Out}$ with
message spaces $\{\hilbert^{\text{in}}_i\}_{i \in \mathsf{In}}$ and
$\{\hilbert^{\text{out}}_o\}_{o \in \mathsf{Out}}$,
respectively. Furthermore, upon receiving a message at
$i \in \mathsf{In}$, $\aS$ always produces outputs at a fixed set of
out-ports $\cO_i \subseteq \mathsf{Out}$ such that
$\cO_i \cap \cO_j = \emptyset$ if $i \neq j$. A system $\aS$ is thus
entirely described by a set of completely positive, trace\-/preserving
(CPTP) maps
\begin{equation}
\label{eq:system}
\left\{ \cE_i : \mathcal{L}\left( \hilbert^{\text{in}}_i \tensor
    \hilbert_{\text{mem}} \right) \to \mathcal{L}\left(
    \hilbert_{\text{mem}} \otimes \hilbert^{\text{out}}_{\cO_i} \right)
\right\}_{i \in \mathsf{In}}\,,\end{equation}
where
$\hilbert^{\text{out}}_{\cO_i} = \bigotimes_{o \in \cO_i}
\hilbert^{\text{out}}_o$
and $\lo{}$ is the space of linear operators on $\hilbert$. Upon
receiving a message at some port $i \in \mathsf{In}$, the
system $\aS$ then applies the map $\cE_i$ and outputs the messages in the
registers $\cO_i$.

For a fixed order of messages received, this specific type of system
has been called a \emph{quantum strategy}~\cite{GW07,Gut12},
\emph{quantum comb}~\cite{CDP09} or \emph{operator
  tensor}~\cite{Har11,Har12,Har15} \--- here we use the terminology
from \cite{CDP09}, namely \emph{comb}. A comb can be represented more
compactly as a single CPTP map
$\cE : \lo{\mathsf{In}} \to \lo{\mathsf{Out}}$, and using the
composition rules for combs~\cite{CDP09,Har12} or causal
boxes~\cite{PMMRT17}, two such systems can be composed to obtain a new
system of the same type. The exact formula for composing systems is
not needed in the current work; in all special cases where we connect
two systems, the resulting system can easily be worked out by hand. We
refer the interested reader to \cite{CDP09,Har12,PMMRT17} for the
generic cases.

As mentioned at the end of \appendixref{app:ac.objects} we use a
distinguisher metric to define the distance between two systems $\aR$
and $\aS$. This means that another system $\aD$, a
\emph{distinguisher}, is given access to either $\aR$ or $\aS$, and
has to guess to which of the two it is connected. Let $\aD[\aR]$ and
$\aD[\aS]$ be the binary random variables corresponding to $\aD$'s
guess, then the distinguishing advantage between $\aR$ and $\aS$ for
this specific distinguisher is given by
\[ d_{\aD}(\aR,\aS) \coloneqq \left| \Pr\left[ \aD[\aR] = 0 \right] -
  \Pr\left[ \aD[\aS] = 0 \right] \right|\,,\]
and the distance is given by
\[ d(\aR,\aS) \coloneqq \sup_{\aD} d_{\aD}(\aR,\aS)\,,\]
where the supremum is taken over all distinguishers allowed by quantum
mechanics\footnote{In the computational setting one would restrict the
  set of distinguishers to those that are efficient. In our
  information\-/theoretic setting the distance is defined over
  unbounded distinguishers as well.} \--- it has been proven in
\cite{PMMRT17} that $d( \cdot, \cdot)$ is indeed a metric.

Due to the simple structure of the systems considered in this work,
the distinguishing metric can be reduced to the following strategy
\--- for the general case of the distinguishing metric between causal
boxes we refer the reader to \cite{PMMRT17}. Let the distinguisher
have internal memory $\hilbert_{R}$. It choses an in-port $i_1$, and
prepares a state $\rho_{RA_1}$. The $A_1$ part is sent to the system
at the $i_1$ port. It then receives the output on ports $\cO_{i_1}$,
which it appends to its internal memory. It measures its internal
memory to decide on the next in-port $i_2$, applies a map
$\cF_1 : \lo{R} \to \lo{RA_2}$, and inputs the $A_2$ part at the
corresponding port. This process is repeated until there are no more
unused ports, after which it measures its internal memory and produces
one bit of output, its final guess. In the case where the two systems
being compared only have one in-port, this metric reduces to the
diamond norm. And if the systems have no in-port (or one trivial
in-port of dimension $1$), this results in the trace\-/distance
between the states output by the two systems.

\section{Unitary designs}
\label{app:design}

The concept of a unitary $2$-design was originally proposed in
\cite{Dan05,DCEL09}. The following (equivalent) definition is taken
from \cite{GAE07}.

\begin{deff}[Unitary $2$-design~\cite{Dan05,DCEL09,GAE07}]
A unitary $2$-design is a finite set of unitaries $\{V_j\}_{j \in
  \cJ}$ on $\hilbert = \complex^d$ such that for all $\rho \in
\cL(\hilbert \tensor \hilbert)$
\begin{equation}
\label{eq:design.1}
\frac{1}{|\cJ|} \sum_{j \in \cJ} ( V_j \tensor V_j ) \rho
( \hconj{V}_j \tensor \hconj{V}_j )  = \int_{V(d)} ( V \tensor V ) \rho
( \hconj{V} \tensor \hconj{V} ) dV\,,
\end{equation}
where $dV$ is the Haar measure. Equivalently, $\{V_j\}_{j \in \cJ}$ is
a unitary $2$-design if for any quantum channel
$\Lambda : \lo{} \to \lo{}$ and state $\rho \in \lo{}$,
\begin{equation}
\label{eq:design.2}
\frac{1}{|\cJ|} \sum_{j \in \cJ} \hconj{V}_j \Lambda \left( V_j \rho
\hconj{V}_j \right) V_j = \int_{V(d)} \hconj{V} \Lambda \left( V \rho
\hconj{V} \right) V dV\,.
\end{equation}
\end{deff}

One way to construct a unitary $2$-design is by finding a set of
unitaries $\{U_k\}_{k \in \cK}$ that map all non\-/identity Paulis to each other
with equal frequency, i.e., $\forall P_\ell,P_{\ell'}$ such that
$P_\ell \neq I$ and $P_{\ell'} \neq I$,
\begin{equation}
\label{eq:design.uniformmapping}
\left| \left\{k \in \cK : U_k P_\ell \hconj{U}_k = e^{i
      \theta_{\ell,\ell',k}} P_{\ell'} \right\} \right| =
\frac{|\cK|}{d^2-1} \,,
\end{equation}
where $e^{i\theta_{\ell,\ell',k}}$ is some global phase and $d$ is the
dimension of the Hilbert space. A unitary $2$\-/design is then
obtained by composing these unitaries with Paulis, i.e., the set
$\{P_\ell U_k\}_{\ell,k}$ is a unitary $2$\-/design. This has been
used in \cite{DCEL09} to show that the Clifford group is a unitary
$2$-design, and is further discussed in \cite{GAE07}.

Chau~\cite{Cha05} finds a set $\{U_k\}_{k \in \cK}$ satisfying
\eqnref{eq:design.uniformmapping}. To understand his construction, we
must view the indices $x$ and $z$ of a Pauli operator $P_{x,z}$ as
elements of a Galois field $x,z \in \mathrm{GF}(d)$. Let
$M = \twobytwo{\alpha}{\beta}{\delta}{\gamma} \in \mathrm{SL}(2,d)$ be any
element of the special linear group of $2 \times 2$ matrices over the
finite field $\mathrm{GF}(d)$, i.e., matrices with determinant $1$. Chau then finds
unitaries $U_{M}$ such that
\[ U_M P_{x,z}
\hconj{U}_{M} = e^{i\theta} P_{\alpha x +
  \beta z, \delta x + \gamma z}\,,\]
for some global phase $e^{i \theta}$ that may depend on
$M,x,z$, where the arithmetic in the indices
is done in $\mathrm{GF}(d)$. By considering the entire set
$\mathrm{SL}(2,d)$ one can verify that
\eqnref{eq:design.uniformmapping} is satisfied. Since
$|\mathrm{SL}(2,d)| = d^3-d$, we need $\log (d^3 - d) \leq 3 \log d$
bits of key to chose the unitary. We now show that this set is a
strong purity testing code.

\begin{lem}
\label{lem:purity}
Any set $\{U_k\}_{k \in \cK}$ satisfying
\eqnref{eq:design.uniformmapping} with $d = 2^{m+n}$ is a strong
purity testing code with $\eps = 2^{-n}$.
\end{lem}

\begin{proof}
An error $P_\ell$ is not detected if it is mapped to $P_{\ell'} =
P_{x,z} \tensor P_{0,z'}$. There are $2^{2m+n}-1$ such Paulis
$P_{\ell'}$ which are not identity. Since the unitaries $\{U_k\}_{k
  \in \cK}$ are constructed to map $P_\ell$ to all non\-/identity
$P_{\ell'}$ with equal frequency, then $\frac{2^{2m+n}-1}{2^{2m+2n}-1}
\leq 2^{-n}$ of them will not detect $P_\ell$.
\end{proof}

A unitary $t$\-/design is defined similarly to a unitary $2$\-/design,
except that it has a $t$\-/fold tensor product instead of a $2$\-/fold
tensor product in \eqnref{eq:design.1}. Intuitively, if a unitary
$2$\-/design is a good quantum authentication scheme, then so should
any (approximate) unitary $t$\-/design for any $t \geq 2$. One can
however not directly apply our proof to unitary $t$\-/designs. This is
because we use \eqnref{eq:design.uniformmapping} as an intermediary
step, to show that a unitary $2$\-/design is a strong purity testing
code. \eqnref{eq:design.uniformmapping} is also satisfied by unitary
$3$\-/designs (given by the Clifford group~\cite{Zhu15,Web15}), but
not necessarily for $t \geq 4$, where the unitaries are not elements
of the Clifford group anymore. One can however show directly from
\eqnref{eq:design.2} that a unitary $2$\-/design is a strong purity
testing code, and since all (approximate) unitary $t$\-/designs
(approximately) satisfy \eqnref{eq:design.2}, they can all be used to
construct quantum authentication schemes that have the same key
recycling properties as unitary $2$\-/designs.\footnote{The reason we
  used \eqnref{eq:design.uniformmapping} and not \eqnref{eq:design.2}
  to prove the security of the unitary $2$\-/design scheme, is that
  Chau's construction~\cite{Cha05} is a subset of a unitary
  $2$\-/design that satisfies \eqnref{eq:design.uniformmapping} but
  not \eqnref{eq:design.2}. We still have to compose it with a random
  Pauli to obtain the unitary $2$\-/design.}

\begin{lem}
\label{lem:purity.2}
Any $\delta$\-/approximate $t$\=/design with $t \geq 2$ is a strong
purity testing code with error $\delta + 2^{-n}$.
\end{lem}

\begin{proof}
  To prove that a set of unitaries is a strong purity testing code,
  one has to show that it can be used to detect all non\-/identity
  Pauli errors with high probability. Setting
  $\Lambda(\rho) = P_\ell \rho P_\ell$ for a non\-/identity Pauli
  $P_\ell$ in \eqnref{eq:design.2}, one can show that that the RHS
  becomes (see, e.g., \cite{DCEL09}),
  \[ \frac{d}{d^2-1} \1 - \frac{1}{d^2-1} \rho\,.\]
  If $d = 2^{m+n}$, $\rho = \rho' \tensor \proj{0}$ where $\rho'$ is
  the first $m$ qubits of $\rho$ and the last $n$ qubits are used as
  syndrome, then the probability of obtaining $0$ when performing a
  measuring in the computational basis on the syndrome is
\[ \frac{d^2 2^{-n}}{d^2-1} - \frac{1}{d^2-1} \leq 2^{-n} \,.\]
Thus, if the distance between the LHS and RHS of
\eqnref{eq:design.2} is $\delta$, then the probability of not
detecting a Pauli error is less than $\delta + 2^{-n}$.
\end{proof}

\section{Technical lemma}
\label{app:lemmas}

The following lemma is used in the proof of \thmref{thm:auth} to
bound the trace distance between two states with the $2$-norm.

\begin{lem}
\label{lem:trace2norm}
Let $\ket{\psi}$ and $\ket{\varphi}$ be two subnormalized states. Then
\begin{equation}
\label{eq:trace2norm}
 \frac{1}{2} \trnorm{\proj{\psi} - \proj{\varphi}} \leq
\norm{\ket{\psi}-\ket{\varphi}} \,,\end{equation}
where $\norm{\ket{a}} = \sqrt{\braket{a}{a}}$ is the vector
$2$\=/norm.
\end{lem}

\begin{proof}
It was shown in \cite[Lemma~A.2.3]{Ren05} that if
$\braket{\psi}{\varphi}$ is real, then
\begin{equation*}
  \trnorm{\proj{\psi} - \proj{\varphi}} =
  \norm{\ket{\psi}-\ket{\varphi}} \cdot \norm{\ket{\psi}+\ket{\varphi}}
  \,.
\end{equation*}
For complex $\braket{\psi}{\varphi}$ we define
$\ket{\varphi'} \coloneqq
\frac{\braket{\varphi}{\psi}}{|\braket{\varphi}{\psi}|}
\ket{\varphi}$. It then follows from \cite[Lemma~A.2.3]{Ren05} that
\begin{equation*}
  \trnorm{\proj{\psi} - \proj{\varphi}} = \trnorm{\proj{\psi} - \proj{\varphi'}} =
  \norm{\ket{\psi}-\ket{\varphi'}} \cdot \norm{\ket{\psi}+\ket{\varphi'}}
  \,.
\end{equation*}
To prove this lemma it remains to show that 
\[ \norm{\ket{\psi}-\ket{\varphi'}} \cdot
\norm{\ket{\psi}+\ket{\varphi'}} \leq \norm{\ket{\psi}-\ket{\varphi}}
\cdot \norm{\ket{\psi}+\ket{\varphi}}\,,\]
since combining this with $\norm{\ket{\psi}+\ket{\varphi}} \leq 2$ we
get \eqnref{eq:trace2norm}.

Writing out the norms with the scalar product we obtain
\begin{align*}
& \norm{\ket{\psi}-\ket{\varphi}}^2 \cdot \norm{\ket{\psi}+\ket{\varphi}}^2\\ 
& \qquad = \left(\braket{\psi}{\psi} + \braket{\varphi}{\varphi} -
  \braket{\psi}{\varphi} - \braket{\varphi}{\psi}\right)
\left(\braket{\psi}{\psi} + \braket{\varphi}{\varphi} +
  \braket{\psi}{\varphi} + \braket{\varphi}{\psi}\right) \\
& \qquad = \left(\braket{\psi}{\psi} + \braket{\varphi}{\varphi}\right)^2 -
  \left(\braket{\psi}{\varphi} + \braket{\varphi}{\psi}\right)^2\,.
\end{align*}
Thus, using $|\braket{\psi}{\varphi}| = \braket{\psi}{\varphi'}$ we get
\begin{align*}
& \norm{\ket{\psi}-\ket{\varphi}}^2 \cdot
  \norm{\ket{\psi}+\ket{\varphi}}^2 - \norm{\ket{\psi}-\ket{\varphi'}}^2
  \cdot \norm{\ket{\psi}+\ket{\varphi'}}^2 \\ 
& \qquad = \left(\braket{\psi}{\varphi'} + \braket{\varphi'}{\psi}\right)^2 -
  \left(\braket{\psi}{\varphi} + \braket{\varphi}{\psi}\right)^2 \\
& \qquad = \braket{\psi}{\varphi'}^2 + \braket{\varphi'}{\psi}^2 -
  \braket{\psi}{\varphi}^2 - \braket{\varphi}{\psi}^2 \\
& \qquad = 2 |\braket{\psi}{\varphi}|^2  -
  \braket{\psi}{\varphi}^2 - \braket{\varphi}{\psi}^2 \geq 0 \,. \qedhere
\end{align*}
\end{proof}

\section{Authentication without key recycling}
\label{app:norecycle}

The proof for authentication with key recycling provided in
\secref{sec:auth} is automatically a proof for authentication without
key recycling, since the players do not have to reuse the key if they
do not want to. But the parameters are not optimal in this case,
because recycling the key causes the error to change from $\eps$ to
$\sqrt{\eps}$. What is more, the proof given is only valid for strong
purity testing codes, since these are essential to be able to recycle
all the key. But if the users are not interested in recycling key, it
is sufficient to use weak purity testing codes. For completeness, we
provide a proof here that the entire family of quantum
authentication protocols from \cite{BCGST02} is secure, i.e., they
construct the secure channel $\aS^m_\flat$, where $\aS^m$ is the
resource considered before \--- which is drawn in
\figref{fig:channel.secure} \--- and $\flat_E$ is the obvious filter
which lets the message through.

Like in the case of key recycling, we can consider an
encrypt\-/then\-/encode and an encode\-/then\-/encrypt protocol, which
are identical to Figures~\ref{fig:protocol.1} and
\ref{fig:protocol.2}, except that the players use a weak purity
testing code, do not recycle any key and do not use an backwards
authentic channel. Let the weak purity testing code
$\{U_k\}_{k \in \aK}$ have size $\log |\aK| = \nu$ and error $\eps$,
and encode an $m$ qubit message in an $m+n$ qubit cipher. As
previously, let $\pi^{\qauth}_{AB}$ denote the encode\-/then\-/encrypt
converters and $\bar{\pi}^{\qauth}_{AB}$ denote the
encrypt\-/then\-/encode version.

\begin{lem}
\label{lem:weak.auth}
The converters $\pi^{\qauth}_{AB}$ for the encode\-/then\-/encrypt
protocol without key recycling construct the secure channel
$\aS^m_\flat$, given an insecure quantum channel $\aC_{\square}$ and a
secret key $\aK^{\nu+2m+2n}$, i.e.,
\[\aC_\square \| \aK^{\nu+2m+2n}
\xrightarrow{\pi^\qauth_{AB},\eps^{\qauth}} \aS^m_\flat\,,\]
with $\eps^{\qauth} = \max\{\eps,2^{-n}\}$.
\end{lem}

\begin{proof}
  Correctness of the protocol is trivial, so we only need to consider
  security. Just as in the proof of \thmref{thm:auth}, the
  distinguisher has the choice between providing the inputs in two
  orders, first a message for Alice, then a (possibly modified) cipher
  at Eve's interface, or first the cipher then the message. In the
  latter case, the simulator always tells the ideal channel to output
  an error; then when the ideal channel notifies the simulator that a
  message has been input at Alice's interface, the simulator outputs a
  fully mixed state. This is exactly the same behavior as the
  simulator used in the proof of \thmref{thm:auth}, except that no key
  is output at any point. We proved back then that the distinguisher
  has an advantage of at most $2^{-n}$ at distinguishing the real and
  ideal systems. The proof does not depend on the purity testing code,
  it follows directly from the random Paulis $P_\ell$ used for
  decrypting and encrypting. So in particular, it is also valid when
  the unitaries $\{U_k\}_{k \in \aK}$ form a weak purity testing
  code. Hence in the case of weak purity testing codes without key
  recycling, the distinguishing advantage for this order of messages
  is also bounded by $2^{-n}$.

  The first case \--- when the distinguisher first provides a message
  at Alice's interface, then modifies the cipher \--- requires a
  different simulator and different proof than that of
  \thmref{thm:auth} to go through with weak purity testing codes. The
  simulator we use also prepares $n+m$ EPR pairs $\ket{\Phi}^{CR}$,
  but this time it picks a key $k$ uniformly at random and applies
  $U_k$ to the halves in the $C$ system, which it outputs at Eve's
  interface. Upon receiving the (possibly modified) system $C'$ back
  from the distinguisher, it applies the inverse $\hconj{U}_k$, then
  measures in the Bell basis. Let the outcome of the measurement be
  $j'$. If the Pauli $P_{j'}$ acts trivially on the $M'$ subsystem of
  $C' = S'M'$ and only flips phases on $S'$ \--- i.e.,
  $P_{j'} = P_{0,z} \tensor P_{0,0}$ \--- then the simulator tells the
  ideal resource to output the message, otherwise it should output an
  error.

Putting this together with the distinguisher that first
prepares a state $\ket{\psi}^{ME}$, inputs the $M$ part at Alice's
interface, receives some cipher in the system $C$, applies a unitary
$U^{CE} = \sum_j P_j^C \tensor E^E_j$ to the $CE$ system, and inputs
the modified $C$ system back on the channel, we get the following
final state in the ideal case:
\begin{multline} \label{eq:auth.ideal.weak} \zeta = \proj{\acc}
  \tensor \frac{1}{2^\nu} \sum_{k} \sum_{j \in \cQ_k} \left[ \left(
      I^M \tensor E^E_j \right) \proj{\psi}^{ME} \left( I^M \tensor
      \hconj{\left(E^E_j\right)} \right) \right] \\ {}+ \proj{\rej}
  \tensor \frac{1}{2^\nu} \sum_{k} \sum_{j \notin \cQ_k} E^E_j \rho^E
  \hconj{\left(E^E_j\right)}\,,\end{multline} where $\cQ_k$ is the set
of Paulis that are not detected by the code and act trivially on
$M$, i.e., the $j$ for which
$\hconj{U}_k P^{SM}_j U_k = e^{i \theta_{k,j}} P^S_{0,z} \tensor
P^M_{0,0}$.

In the real system, for the secret key $(k,\ell)$, the state before
Bob's measurement of the syndrome is given by
\begin{align*}\ket{\varphi_{k,\ell}}^{SME} & = \sum_j \left( \hconj{\left(U_k^{SM}\right)} P^{SM}_\ell P^{SM}_j P^{SM}_\ell U^{SM}_k \tensor
    E^E_j \right) \zero^S \ket{\psi}^{ME} \\ & = \sum_j (-1)^{\sym{j}{\ell}}\left( \hconj{\left(U_k^{SM}\right)}
    P^{SM}_j U^{SM}_k \tensor E^E_j \right) \zero^S \ket{\psi}^{ME} \\ & = \sum_j (-1)^{\sym{j}{\ell}}
    e^{i \theta_{k,j}} \left( P^{SM}_{k(j)} \tensor E^E_j \right) \zero^S \ket{\psi}^{ME}\,,
\end{align*}
where we denote by $k(j)$ the index of the Pauli operator
such that $\hconj{U}_k P_j U_k = e^{i \theta_{k,j}} P_{k(j)}$. Summing
over $k$ and $\ell$, the state before Bob's measurement is given by
\begin{align*} & \frac{1}{2^{\nu+2m+2n}} \sum_{k,\ell} \proj{\varphi_{k,\ell}}^{SME} \\
& \qquad = \frac{1}{2^{\nu+2m+2n}} \sum_{j_1,j_2,k,\ell} (-1)^{\sym{j_1 \xor j_2}{\ell}}
e^{i \theta_{k,j_1} - i \theta_{k,j_2} } \left( P^{SM}_{k(j_1)}
  \tensor E^E_{j_1} \right) \\
& \qquad \qquad \qquad \qquad \qquad \qquad \qquad \left( \proj{0}^S  \tensor
  \proj{\psi}^{ME} \right) \left( P^{SM}_{k(j_2)}
  \tensor \hconj{\left( E^E_{j_2} \right)} \right) \\
& \qquad = \frac{1}{2^{\nu}} \sum_{j,k} \left( P^{SM}_{k(j)}
  \tensor E^E_{j} \right) \left( \proj{0}^S  \tensor \proj{\psi}^{ME} \right) \left( P^{SM}_{k(j)}
  \tensor \hconj{\left( E^E_{j} \right)} \right) \,,
\end{align*}
where we used \eqnref{eq:symplectic.trick}. The measurement of $S$
yields $\ket{\acc}$ if $P^{SM}_{k(j)} = P^S_{0,z'} \tensor P^M_{x,z}$,
i.e., $j \in \cP_k$, where $\cP_k$ denotes the set of Paulis that are not
detected by the code $k$. Hence, the final state held by the
distinguisher after Bob's measurement is
given by 
\begin{multline} \label{eq:auth.real.weak} \xi = \proj{\acc}
  \tensor \frac{1}{2^\nu} \sum_{k} \sum_{j \in \cP_k} \left[ \left(
      P^M_{k(j)} \tensor E^E_j \right) \proj{\psi}^{ME} \left( P^M_{k(j)} \tensor
      \hconj{\left(E^E_j\right)} \right) \right] \\ {}+ \proj{\rej}
  \tensor \frac{1}{2^\nu} \sum_{k} \sum_{j \notin \cP_k} E^E_j \rho^E
  \hconj{\left(E^E_j\right)}\,.\end{multline}

The distinguishability between the real and ideal systems is given by
the trace distance between \eqnsref{eq:auth.ideal.weak} and
\eqref{eq:auth.real.weak}, namely
\begin{align*}
 \frac{1}{2} \trnorm{\xi - \zeta} & = \frac{1}{2}
                                    \trnorm{\frac{1}{2^\nu} \sum_{k}
                                    \sum_{j \in \cP_k \setminus Q_k} \left[ \left(
      P^M_{k(j)} \tensor E^E_j \right) \proj{\psi}^{ME} \left( P^M_{k(j)} \tensor
      \hconj{\left(E^E_j\right)} \right) \right]} \\ 
    & \qquad \qquad \qquad \qquad \quad \qquad \qquad {}+\frac{1}{2}
                                    \trnorm{ \frac{1}{2^\nu} \sum_{k}
      \sum_{j \in \cP_k \setminus \cQ_k} E^E_j \rho^E
  \hconj{\left(E^E_j\right)}} \\
& \leq  \frac{1}{2^\nu} \sum_{k}
      \sum_{j \in \cP_k \setminus \cQ_k} \trace{E^E_j \rho^E
  \hconj{\left(E^E_j\right)}} \\
& = \sum_j \trace{E^E_j \rho^E
  \hconj{\left(E^E_j\right)}} \frac{\left|\left\{ k : j \in \cP_k \setminus \cQ_k
      \right\}\right|}{2^{\nu}} \leq \eps \,.  \qedhere \end{align*}
\end{proof}

It follows from this and \lemref{lem:ete} that the
encrypt\-/then\-/encode protocol is also secure.

\begin{cor}
\label{cor:weak.auth}
The converters $\bar{\pi}^{\qauth}_{AB}$ for the encrypt\-/then\-/encode
protocol without key recycling construct the secure channel
$\aS^m_\flat$, given an insecure quantum channel $\aC_{\square}$ and a
secret key $\aK^{\nu+2m+n}$, i.e.,
\[\aC_\square \| \aK^{\nu+2m+n}
\xrightarrow{\bar{\pi}^\qauth_{AB},\eps^{\qauth}} \aS^m_\flat\,,\]
with $\eps^{\qauth} = \max\{\eps,2^{-n}\}$.
\end{cor}




\begin{thebibliography}{SBPC{\etalchar{+}}09}

\bibitem[ABE10]{ABE10}
Dorit Aharonov, Michael {Ben-Or}, and Elad Eban.
\newblock Interactive proofs for quantum computations.
\newblock In {\em Proceedings of Innovations in Computer Science, ICS 2010},
  pages 453--469. Tsinghua University Press, 2010.
\newblock [\epfmt{arxiv}{0810.5375}].

\bibitem[AM16]{AM16}
Gorjan Alagic and Christian Majenz.
\newblock Quantum non-malleability and authentication.
\newblock Eprint, 2016.
\newblock [\epfmt{arxiv}{1610.04214}].

\bibitem[BB84]{BB84}
Charles~H. Bennett and Gilles Brassard.
\newblock Quantum cryptography: Public key distribution and coin tossing.
\newblock In {\em Proceedings of IEEE International Conference on Computers,
  Systems, and Signal Processing}, pages 175--179, 1984.

\bibitem[BBB82]{BBB82}
Charles~H. Bennett, Gilles Brassard, and Seth Breidbart.
\newblock Quantum cryptography {II}: How to re-use a one-time pad safely even
  if {P=NP}.
\newblock Original unpublished manuscript uploaded to arXiv in 2014, 1982.
\newblock [\epfmt{arxiv}{1407.0451}].

\bibitem[BBCM95]{BBCM95}
Charles~H. Bennett, Gilles Brassard, Claude Cr{\'e}peau, and Ueli Maurer.
\newblock Generalized privacy amplification.
\newblock {\em IEEE Transaction on Information Theory}, 41(6):1915--1923,
  November 1995.
\newblock [\epfmtdoi{10.1109/18.476316}].

\bibitem[BCG{\etalchar{+}}02]{BCGST02}
Howard Barnum, Claude Cr{\'e}peau, Daniel Gottesman, Adam Smith, and Alain
  Tapp.
\newblock Authentication of quantum messages.
\newblock In {\em Proceedings of the 43rd Symposium on Foundations of Computer
  Science, FOCS~'02}, pages 449--458. IEEE, 2002.
\newblock [\epfmtdoi{10.1109/SFCS.2002.1181969},
  \epfmt{arxiv}{quant-ph/0205128}].

\bibitem[BCG{\etalchar{+}}06]{BCGHS06}
Michael {Ben-Or}, Claude Cr\'epeau, Daniel Gottesman, Avinatan Hassidim, and
  Adam Smith.
\newblock Secure multiparty quantum computation with (only) a strict honest
  majority.
\newblock In {\em Proceedings of the 47th Symposium on Foundations of Computer
  Science, FOCS~'06}, pages 249--260, 2006.
\newblock [\epfmtdoi{10.1109/FOCS.2006.68}, \epfmt{arxiv}{0801.1544}].

\bibitem[BGS13]{BGS13}
Anne Broadbent, Gus Gutoski, and Douglas Stebila.
\newblock Quantum one-time programs.
\newblock In {\em Advances in Cryptology -- CRYPTO 2013}, volume 8043 of {\em
  Lecture Notes in Computer Science}, pages 344--360. Springer, 2013.
\newblock [\epfmtdoi{10.1007/978-3-642-40084-1_20}, \epfmt{arxiv}{1211.1080}].

\bibitem[BW16]{BW16}
Anne Broadbent and Evelyn Wainewright.
\newblock Efficient simulation for quantum message authentication.
\newblock In {\em Proceedings of the 9th International Conference on
  Information Theoretic Security, ICITS 2016}, pages 72--91. Springer, 2016.
\newblock [\epfmtdoi{10.1007/978-3-319-49175-2_4}, \epfmt{arxiv}{1607.03075}].

\bibitem[Can01]{Can01}
Ran Canetti.
\newblock Universally composable security: A new paradigm for cryptographic
  protocols.
\newblock In {\em Proceedings of the 42nd Symposium on Foundations of Computer
  Science, FOCS~'01}, pages 136--145. IEEE, 2001.
\newblock [\epfmtdoi{10.1109/SFCS.2001.959888}].

\bibitem[Can13]{Can13}
Ran Canetti.
\newblock Universally composable security: A new paradigm for cryptographic
  protocols.
\newblock Cryptology ePrint Archive, Report 2000/067, 2013.
\newblock Updated version of~\cite{Can01}.
\newblock [\epfmt{cryptoeprint}{2000/067}].

\bibitem[CDP09]{CDP09}
Giulio Chiribella, Giacomo~Mauro D'Ariano, and Paolo Perinotti.
\newblock Theoretical framework for quantum networks.
\newblock {\em Physical Review A}, 80:022339, August 2009.
\newblock [\epfmtdoi{10.1103/PhysRevA.80.022339}, \epfmt{arxiv}{0904.4483}].

\bibitem[Cha05]{Cha05}
Hoi~Fung Chau.
\newblock Unconditionally secure key distribution in higher dimensions by
  depolarization.
\newblock {\em IEEE Transactions on Information Theory}, 51(4):1451--1468,
  April 2005.
\newblock [\epfmtdoi{10.1109/TIT.2005.844076},
  \epfmt{arxiv}{quant-ph/0405016}].

\bibitem[Dan05]{Dan05}
Christoph Dankert.
\newblock Efficient simulation of random quantum states and operators.
\newblock Master's thesis, University of Waterloo, 2005.
\newblock [\epfmt{arxiv}{quant-ph/0512217}].

\bibitem[DCEL09]{DCEL09}
Christoph Dankert, Richard Cleve, Joseph Emerson, and Etera Livine.
\newblock Exact and approximate unitary 2-designs and their application to
  fidelity estimation.
\newblock {\em Physical Review A}, 80:012304, July 2009.
\newblock [\epfmtdoi{10.1103/PhysRevA.80.012304},
  \epfmt{arxiv}{quant-ph/0606161}].

\bibitem[DNS12]{DNS12}
Fr{\'e}d{\'e}ric Dupuis, Jesper~Buus Nielsen, and Louis Salvail.
\newblock Actively secure two-party evaluation of any quantum operation.
\newblock In {\em Advances in Cryptology -- CRYPTO 2012}, volume 7417 of {\em
  Lecture Notes in Computer Science}, pages 794--811. Springer, 2012.
\newblock [\epfmtdoi{10.1007/978-3-642-32009-5_46},
  \epfmt{cryptoeprint}{2012/304}].

\bibitem[DPS05]{DPS05}
Ivan Damg{\aa}rd, Thomas~Brochmann Pedersen, and Louis Salvail.
\newblock A quantum cipher with near optimal key-recycling.
\newblock In {\em Advances in Cryptology -- CRYPTO 2005}, volume 3621 of {\em
  Lecture Notes in Computer Science}, pages 494--510. Springer, 2005.
\newblock [\epfmtdoi{10.1007/11535218_30}].

\bibitem[DPS14]{DPS14}
Ivan Damg{\aa}rd, Thomas~Brochmann Pedersen, and Louis Salvail.
\newblock How to re-use a one-time pad safely and almost optimally even if {P =
  NP}.
\newblock {\em Natural Computing}, 13(4):469--486, December 2014.
\newblock [\epfmtdoi{10.1007/s11047-014-9454-5}].

\bibitem[FS17]{FS17}
Serge Fehr and Louis Salvail.
\newblock Quantum authentication and encryption with key recycling.
\newblock To appear in the proceedings of EUROCRYPT, 2017.
\newblock [\epfmt{arxiv}{1610.05614}].

\bibitem[GAE07]{GAE07}
David Gross, Koenraad Audenaert, and Jens Eisert.
\newblock Evenly distributed unitaries: on the structure of unitary designs.
\newblock {\em Journal of Mathematical Physics}, 48:052104, 2007.
\newblock [\epfmtdoi{10.1063/1.2716992}, \epfmt{arxiv}{quant-ph/0611002}].

\bibitem[Got03]{Got03}
Daniel Gottesman.
\newblock Uncloneable encryption.
\newblock {\em Quantum Information and Computation}, 3:581, 2003.
\newblock [\epfmt{arxiv}{quant-ph/0210062}].

\bibitem[Gut12]{Gut12}
Gus Gutoski.
\newblock On a measure of distance for quantum strategies.
\newblock {\em Journal of Mathematical Physics}, 53(3):032202, 2012.
\newblock [\epfmtdoi{10.1063/1.3693621}, \epfmt{arxiv}{1008.4636}].

\bibitem[GW07]{GW07}
Gus Gutoski and John Watrous.
\newblock Toward a general theory of quantum games.
\newblock In {\em Proceedings of the 39th Symposium on Theory of Computing,
  STOC~'07}, pages 565--574. ACM, 2007.
\newblock [\epfmtdoi{10.1145/1250790.1250873},
  \epfmt{arxiv}{quant-ph/0611234}].

\bibitem[GYZ16]{GYZ16}
Sumegha Garg, Henry Yuen, and Mark Zhandry.
\newblock New security notions and feasibility results for authentication of
  quantum data.
\newblock Eprint, 2016.
\newblock [\epfmt{arxiv}{1607.07759}].

\bibitem[Har11]{Har11}
Lucien Hardy.
\newblock Reformulating and reconstructing quantum theory.
\newblock Eprint, 2011.
\newblock [\epfmt{arxiv}{1104.2066}].

\bibitem[Har12]{Har12}
Lucien Hardy.
\newblock The operator tensor formulation of quantum theory.
\newblock {\em Philosophical Transactions of the Royal Society of London A:
  Mathematical, Physical and Engineering Sciences}, 370(1971):3385--3417, 2012.
\newblock [\epfmtdoi{10.1098/rsta.2011.0326}, \epfmt{arxiv}{1201.4390}].

\bibitem[Har15]{Har15}
Lucien Hardy.
\newblock Quantum theory with bold operator tensors.
\newblock {\em Philosophical Transactions of the Royal Society of London A:
  Mathematical, Physical and Engineering Sciences}, 373(2047), 2015.
\newblock [\epfmtdoi{10.1098/rsta.2014.0239}].

\bibitem[HLM11]{HLM11}
Patrick Hayden, Debbie Leung, and Dominic Mayers.
\newblock The universal composable security of quantum message authentication
  with key recycling.
\newblock Eprint, presented at QCrypt 2011, 2011.
\newblock [\epfmt{arxiv}{1610.09434}].

\bibitem[Mau12]{Mau12}
Ueli Maurer.
\newblock Constructive cryptography---a new paradigm for security definitions
  and proofs.
\newblock In {\em Proceedings of Theory of Security and Applications, TOSCA
  2011}, volume 6993 of {\em Lecture Notes in Computer Science}, pages 33--56.
  Springer, 2012.
\newblock [\epfmtdoi{10.1007/978-3-642-27375-9_3}].

\bibitem[MR11]{MR11}
Ueli Maurer and Renato Renner.
\newblock Abstract cryptography.
\newblock In {\em Proceedings of Innovations in Computer Science, ICS 2011},
  pages 1--21. Tsinghua University Press, 2011.

\bibitem[MR16]{MR16}
Ueli Maurer and Renato Renner.
\newblock From indifferentiability to constructive cryptography (and back).
\newblock In {\em Theory of Cryptography, Proceedings of {TCC} 2016-B, Part
  {I}}, volume 9985 of {\em Lecture Notes in Computer Science}, pages 3--24.
  Springer, 2016.
\newblock [\epfmtdoi{10.1007/978-3-662-53641-4_1},
  \epfmt{cryptoeprint}{2016/903}].

\bibitem[OH05]{OH05}
Jonathan Oppenheim and Micha\l{} Horodecki.
\newblock How to reuse a one-time pad and other notes on authentication,
  encryption, and protection of quantum information.
\newblock {\em Physical Review A}, 72:042309, October 2005.
\newblock [\epfmtdoi{10.1103/PhysRevA.72.042309},
  \epfmt{arxiv}{quant-ph/0306161}].

\bibitem[PMM{\etalchar{+}}17]{PMMRT17}
Christopher Portmann, Christian Matt, Ueli Maurer, Renato Renner, and Bj\"orn
  Tackmann.
\newblock Causal boxes: Quantum information-processing systems closed under
  composition.
\newblock To appear in IEEE Trans.\ Inf.\ Theory, 2017.
\newblock [\epfmt{arxiv}{1512.02240}].

\bibitem[Por14]{Por14}
Christopher Portmann.
\newblock Key recycling in authentication.
\newblock {\em IEEE Transactions on Information Theory}, 60(7):4383--4396, July
  2014.
\newblock [\epfmtdoi{10.1109/TIT.2014.2317312}, \epfmt{arxiv}{1202.1229}].

\bibitem[PR14]{PR14}
Christopher Portmann and Renato Renner.
\newblock Cryptographic security of quantum key distribution.
\newblock Eprint, 2014.
\newblock [\epfmt{arxiv}{1409.3525}].

\bibitem[Ren05]{Ren05}
Renato Renner.
\newblock {\em Security of Quantum Key Distribution}.
\newblock PhD thesis, Swiss Federal Institute of Technology (ETH) Zurich,
  September 2005.
\newblock [\epfmt{arxiv}{quant-ph/0512258}].

\bibitem[RK05]{RK05}
Renato Renner and Robert K\"onig.
\newblock Universally composable privacy amplification against quantum
  adversaries.
\newblock In Joe Kilian, editor, {\em Theory of Cryptography, Proceedings of
  TCC 2005}, volume 3378 of {\em Lecture Notes in Computer Science}, pages
  407--425. Springer, 2005.
\newblock [\epfmtdoi{10.1007/978-3-540-30576-7_22},
  \epfmt{arxiv}{quant-ph/0403133}].

\bibitem[SBPC{\etalchar{+}}09]{SBCDLP09}
Valerio Scarani, Helle Bechmann-Pasquinucci, Nicolas~J. Cerf, Miloslav
  Du\ifmmode~\check{s}\else \v{s}\fi{}ek, Norbert L\"utkenhaus, and Momtchil
  Peev.
\newblock The security of practical quantum key distribution.
\newblock {\em Reviews of Modern Physics}, 81:1301--1350, September 2009.
\newblock [\epfmtdoi{10.1103/RevModPhys.81.1301}, \epfmt{arxiv}{0802.4155}].

\bibitem[Sha49]{Sha49}
Claude Shannon.
\newblock Communication theory of secrecy systems.
\newblock {\em Bell System Technical Journal}, 28(4):656--715, 1949.

\bibitem[Sim85]{Sim85}
Gustavus~J. Simmons.
\newblock Authentication theory/coding theory.
\newblock In {\em Advances in Cryptology -- CRYPTO~'84}, pages 411--431.
  Springer, 1985.
\newblock [\epfmtdoi{10.1007/3-540-39568-7_32}].

\bibitem[Sim88]{Sim88}
Gustavus~J. Simmons.
\newblock A survey of information authentication.
\newblock {\em Proceedings of the IEEE}, 76(5):603--620, May 1988.
\newblock [\epfmtdoi{10.1109/5.4445}].

\bibitem[SP00]{SP00}
Peter~W. Shor and John Preskill.
\newblock Simple proof of security of the {BB84} quantum key distribution
  protocol.
\newblock {\em Physical Review Letters}, 85:441--444, 2000.
\newblock [\epfmtdoi{10.1103/PhysRevLett.85.441},
  \epfmt{arxiv}{quant-ph/0003004}].

\bibitem[Sti90]{Sti90}
Douglas~R. Stinson.
\newblock The combinatorics of authentication and secrecy codes.
\newblock {\em Journal of Cryptology}, 2(1):23--49, 1990.
\newblock [\epfmtdoi{10.1007/BF02252868}].

\bibitem[Sti94]{Sti94}
Douglas~R. Stinson.
\newblock Universal hashing and authentication codes.
\newblock {\em Designs, Codes and Cryptography}, 4(3):369--380, 1994.
\newblock A preliminary version appeared at CRYPTO~'91.
\newblock [\epfmtdoi{10.1007/BF01388651}].

\bibitem[Unr10]{Unr10}
Dominique Unruh.
\newblock Universally composable quantum multi-party computation.
\newblock In {\em Advances in Cryptology -- EUROCRYPT 2010}, volume 6110 of
  {\em Lecture Notes in Computer Science}, pages 486--505. Springer, 2010.
\newblock [\epfmtdoi{10.1007/978-3-642-13190-5_25}, \epfmt{arxiv}{0910.2912}].

\bibitem[WC81]{WC81}
Mark~N. Wegman and Larry Carter.
\newblock New hash functions and their use in authentication and set equality.
\newblock {\em Journal of Computer and System Sciences}, 22(3):265--279, 1981.

\bibitem[Web15]{Web15}
Zak Webb.
\newblock The {Clifford} group forms a unitary 3-design.
\newblock {\em Quantum Information {\&} Computation}, 16(15{\&}16):1379--1400,
  2015.
\newblock [\epfmt{arxiv}{1510.02769}].

\bibitem[Zhu15]{Zhu15}
Huangjun Zhu.
\newblock Multiqubit {Clifford} groups are unitary 3-designs.
\newblock Eprint, 2015.
\newblock [\epfmt{arxiv}{1510.02619}].

\end{thebibliography}

\newcommand{\etalchar}[1]{$^{#1}$}
\providecommand{\bibhead}[1]{}
\expandafter\ifx\csname pdfbookmark\endcsname\relax%
  \providecommand{\tocrefpdfbookmark}{}
\else\providecommand{\tocrefpdfbookmark}{%
   \phantomsection%
   \addcontentsline{toc}{section}{\refname}}%
\fi

\tocrefpdfbookmark

\end{document}